\documentclass[journal,onecolumn,12pt]{IEEEtran}
\usepackage{setspace}
\usepackage[T1]{fontenc}
\usepackage{amsmath,amsfonts,amssymb}
\usepackage{algorithmic}
\usepackage{algorithm}
\usepackage{array}
\usepackage{multirow}
\usepackage[caption=false,font=normalsize,labelfont=sf,textfont=sf]{subfig}
\usepackage{textcomp}
\usepackage{stfloats}
\usepackage{url}
\usepackage{verbatim}
\usepackage{graphicx}
\usepackage{cite}
\usepackage{mathtools}
\usepackage{stmaryrd}
\usepackage{amsthm}
\theoremstyle{definition}   % bold heading, upright body
\newtheorem{lemma}{Lemma}
\newtheorem{proposition}{Proposition}
\newtheorem{corollary}{Corollary}
\theoremstyle{remark}
\newtheorem{remark}{Remark}
\usepackage[acronym]{glossaries}
\usepackage{tikz}
\usepackage[none]{hyphenat}
\usetikzlibrary{arrows.meta,positioning,calc,decorations.pathreplacing,shapes.geometric,fit}
\usepackage{pgfplots}
\pgfplotsset{compat=1.18}
\usepgfplotslibrary{fillbetween,groupplots}

\definecolor{cALS}{RGB}{0,84,159}      % ALS + ESPRIT, MDL rank
\definecolor{cALSK}{RGB}{90,90,90}     % ALS + ESPRIT, genie rank
\definecolor{cMUSIC}{RGB}{191,60,30}   % tensor MUSIC FB, MDL rank
\definecolor{cRthree}{RGB}{0,128,84}   % \hat L = 3 (correct)
\definecolor{cRtwo}{RGB}{118,58,150}   % \hat L = 2
\definecolor{cRone}{RGB}{184,134,11}   % \hat L = 1
\pgfplotsset{
  resultaxis/.style={
    width=\columnwidth, height=0.72\columnwidth,
    scale only axis=false,
    grid=both,
    grid style={gray!25},
    tick label style={font=\footnotesize},
    label style={font=\footnotesize},
    legend style={font=\footnotesize, fill=white, fill opacity=0.85,
                  text opacity=1, draw=gray!50, row sep=1pt, inner sep=2pt},
    every axis plot/.append style={line width=1.0pt, mark size=2.8pt},
  },
  als/.style     ={cALS,   mark=*,         solid},
  alsknown/.style={cALSK,  mark=square*,   solid},
  music/.style   ={cMUSIC, mark=triangle*, solid, mark options={rotate=0, solid}},
  rthree/.style  ={cRthree, mark=diamond*,  solid},
  rtwo/.style    ={cRtwo,   mark=pentagon*, solid},
  rone/.style    ={cRone,   mark=x,         solid},
  snrlow/.style  ={solid},
  snrmid/.style  ={dashed,         mark options={solid}},
  snrhigh/.style ={densely dotted, mark options={solid}},
}
\newcommand{\cedata}[1]{data/channel_estimation/#1}
\newcommand{\bfdata}[1]{data/precoder_design/#1}

\definecolor{cGenie}{RGB}{0,110,110}    % perfect CSI + the same search
\definecolor{cTrad}{RGB}{150,60,140}    % fixed 4x4 UPA, identical budget
\definecolor{cZF}{RGB}{0,84,159}
\definecolor{cMMSE}{RGB}{0,128,84}
\definecolor{cMRC}{RGB}{191,60,30}
\pgfplotsset{
  genie/.style ={cGenie, mark=diamond,  solid, mark size=4.2pt},
  trad/.style  ={cTrad,  mark=x,        densely dashed},
  zf/.style    ={cZF,   mark=*,         solid},
  mmse/.style  ={cMMSE, mark=square*,   solid},
  mrc/.style   ={cMRC,  mark=triangle*, solid},
  ptransl/.style={solid},
  pboth/.style  ={dashed,         mark options={solid}},
  pmorph/.style ={densely dotted, mark options={solid}},
}
\colorlet{revcolor}{black}
\tikzset{
  cgroup/.style ={draw=revcolor!75, line width=0.6pt, ellipse, inner sep=0pt},
  cglabel/.style={font=\scriptsize, align=center, inner sep=1.5pt, text=revcolor,
                  fill=white, fill opacity=0.85, text opacity=1},
  cgarrow/.style={draw=revcolor!75, line width=0.6pt, -{Stealth[length=4pt]},
                  shorten >=0.5pt},
}
\newcommand{\curvegroup}[6][0]{%
  \node[cgroup, rotate=#1, minimum width=#3, minimum height=#4]
    (cg) at (axis cs:#2) {};
  \node[cglabel] (cgl) at (axis cs:#5) {#6};
  \draw[cgarrow] (cgl) -- (cg);%
}
\usepackage{microtype}
\newacronym{RIS}{RIS}{reconfigurable intelligent surface}
\newacronym{FIM}{FIM}{flexible intelligent metasurface}
\newacronym{GFIM}{G-FIM}{generalized flexible intelligent metasurface}
\newacronym{BS}{BS}{base station}
\newacronym{NG}{NG}{next-generation}
\newacronym{CP}{CP}{canonical polyadic}
\newacronym{PARAFAC}{PARAFAC}{parallel factor}
\newacronym{ALS}{ALS}{alternating least squares}
\newacronym{NMSE}{NMSE}{normalized mean square error}
\newacronym{RMSE}{RMSE}{root mean square error}
\newacronym{LS}{LS}{least squares}
\newacronym{MMSE}{MMSE}{minimum mean square error}
\newacronym{CSI}{CSI}{channel state information}
\newacronym{AoA}{AoA}{angle of arrival}
\newacronym{MDL}{MDL}{minimum description length}
\newacronym{CRB}{CRB}{Cram\'er--Rao bound}
\newacronym{SNR}{SNR}{signal-to-noise ratio}
\newacronym{UPA}{UPA}{uniform planar array}
\newacronym{MUSIC}{MUSIC}{multiple signal classification}
\newacronym{ESPRIT}{ESPRIT}{estimation of signal parameters via rotational invariance techniques}
\newacronym{EVD}{EVD}{eigenvalue decomposition}
\newacronym{GEVD}{GEVD}{generalized eigenvalue decomposition}
\newacronym{ZF}{ZF}{zero-forcing}
\newacronym{MRC}{MRC}{maximum ratio combining}
\newacronym{TPC}{TPC}{transmit precoder}
\newacronym{SINR}{SINR}{signal-to-interference-plus-noise ratio}
\newacronym{MA}{MA}{movable antenna}
\newacronym{FAS}{FAS}{fluid antenna system}
\newacronym{sixdma}{6DMA}{six-dimensional movable antenna}
\newacronym{SIM}{SIM}{stacked intelligent metasurface}
\newacronym{DMA}{DMA}{dynamic metasurface antenna}
\newacronym{CS}{CS}{compressed sensing}
\newacronym{OMP}{OMP}{orthogonal matching pursuit}
\newacronym{ISAC}{ISAC}{integrated sensing and communications}
\newacronym{SE}{SE}{spectral efficiency}

\DeclareMathOperator{\vecop}{vec}
\newcommand{\kron}{\otimes}
\newcommand{\khatri}{\odot}
\newcommand{\hadamard}{\ast}
\newcommand{\bbR}{\mathbb{R}}
\newcommand{\bbC}{\mathbb{C}}

\begin{document}

% Paper title: edit only here; \title and the running header (\markboth) both use it.
\newcommand{\papertitle}{Training-Efficient Channel Estimation for the Multiuser Uplink Relying on Generalized Flexible Intelligent Metasurfaces}
\title{\papertitle}

\author{Daniel~C.~Ara\'ujo, Andr\'e~L.~F.~de~Almeida, and Lajos~Hanzo}

\maketitle

\begin{abstract}
  We consider uplink multiuser channel estimation for a generalized
  flexible intelligent metasurface at the base station, under narrowband
  far-field propagation and a twin-time-scale mechanical protocol: on the
  slow time-scale the panel translates rigidly over a planar grid, while on the
  fast time-scale the surface morphs along its normal direction. From
  orthogonal multiuser pilots, we show that each user's channel admits a
  fourth-order canonical polyadic tensor structure. We develop two
  complementary estimators: alternating least squares factorization followed
  by paired rotational-invariance wavevector extraction, and a tensor
  spatial-smoothing multiple signal classification estimator that preserves
  the multilinear array manifold. Both share a minimum description length
  front-end for automatic tensor-rank detection. We derive Kruskal
  identifiability conditions that relate the numbers of elements, morphing
  states and grid positions to the maximum number of recoverable multipath
  components, and show that the tensor model reconstructs the channel at panel
  configurations beyond the training grid. This prediction capability drives a
  configuration-selective multiuser transmit precoder that searches the
  entire motion grid without additional pilot overhead. Numerical results
  confirm the identifiability limits and show that, once the model order is
  resolved, the configuration search makes the same choice from estimated as
  from perfect channel state information. The reconfiguration gain over a fixed
  panel of identical aperture and training budget grows with the user load,
  and fast morphing proves valuable for conditioning the tensor factorization
  rather than as a source of beamforming gain.
\end{abstract}


\glsresetall
\section{Introduction}
\IEEEPARstart{R}{econfigurable} apertures have become a central theme of
\gls{NG} wireless research, because the spatial degrees of freedom of
a \gls{BS} are set by the geometry of its array and that geometry has
traditionally been frozen at deployment~\cite{zhu2024ma_magazine}. \Glspl{RIS} are capable of ameliorating the propagation
environment by programmable control~\cite{basar2019wireless,wu2021intelligent}, \glspl{MA}~\cite{zhu2024ma_magazine} and \glspl{FAS}~\cite{wong2021fas} made the
element positions themselves a design variable, and \glspl{FIM}~\cite{yang2025fim_arch,an2025fim_miso} added a
morphological dimension by allowing each element to be displaced along the surface
normal. Every one of these architectures shares the same prerequisite, namely
accurate \gls{CSI}. However, their configuration space is far larger than the one a fixed
array has to learn, and it has to be acquired within a single coherence time interval. Training overhead, rather than hardware, is therefore what
limits how much of the promised reconfigurability can actually be exploited.

In this paper we tackle this problem for a \gls{GFIM} deployed at the
\gls{BS}, which combines rigid panel motion with local morphing to create rich
spatial diversity for uplink channel recovery. More explicitly, we consider uplink multiuser
channel estimation under a narrowband far-field propagation model and a
refined twin-time-scale mechanical protocol, in which the \gls{GFIM} panel
undergoes rigid planar translation over a two-dimensional grid on the slow
time-scale, while the surface itself executes prompt morphing displacements along the
normal direction. The superposition of these two motions results in a 4th-order
\gls{CP} structure in the per-user channel
tensor, which we exploit through a pair of complementary estimation pipelines,
relying on \gls{ALS} along with paired \gls{ESPRIT}~\cite{roy1989esprit} and tensor spatial-smoothing
\gls{MUSIC}~\cite{schmidt1986music}. Both are equipped with automatic rank detection via the \gls{MDL}
criterion. The parametric estimates that these pipelines return are
generative rather than specific, so the \gls{BS} can evaluate the channel
at panel configurations it never visited during training and pick the
configuration that maximizes the downlink sum-rate, which closes the loop
between the estimation model and the \gls{TPC} that exploits it. The
contribution sits at the confluence of four lines of work, namely on flexible
metasurfaces, on movable antennas, on twin-time-scale training and on tensor
channel estimation, which we review in turn before detailing what is new.

\subsection{State-of-the-Art for RISs and Flexible Intelligent Metasurfaces}

Reconfigurable intelligent surfaces (\glspl{RIS}) have been extensively studied as a means of passively shaping the propagation environment through programmable phase shifts applied to a large number of passive elements~\cite{basar2019wireless,wu2021intelligent,pan2021ris6g,pan2020multicell,pan2020swipt}. While a \gls{RIS} provides a valuable electromagnetic control layer, its elements are fixed in space, and its performance gains are fundamentally tied to a favorable propagation geometry. Flexible intelligent metasurfaces~\cite{yang2025fim_arch,an2025fim_miso} extend this paradigm by endowing each element with a normal-direction morphing degree of freedom, so that the surface can dynamically deform its three-dimensional shape~\cite{yang2025fim_arch,hu2025fim_radio}. This morphological dimension yields a richer set of steering vectors and, crucially, allows the system to exploit spatial diversity along the normal axis, i.e., away from the surface plane, which is entirely absent in rigid \gls{RIS} designs. Recent work has shown that a \gls{FIM} achieves significant channel gain improvement over a rigid \gls{RIS} under matched propagation conditions~\cite{hu2025fim_radio}, while the surface spectral efficiency under statistical \gls{CSI} degrades gracefully with spatial channel correlation~\cite{kumar2025fim_stat}. \Gls{SIM} architectures, in which multiple programmable metasurface layers are
cascaded, further multiply the spatial degrees of freedom available for
wave-domain beamforming without requiring additional RF
chains~\cite{stacked2025fim}. From an architectural standpoint, each \gls{FIM}
element is implemented via a reprogrammable metallic filament subject to
distributed Lorentz forces, which enables sub-wavelength actuation without
active RF chains~\cite{yang2025fim_arch}. On the design side, the morphing
degree of freedom has been optimized jointly with the \gls{TPC} for
downlink multiuser \gls{SE}~\cite{an2025fim_miso} and, more recently, for
\gls{ISAC} objectives~\cite{wang2026fim_isac}. Both of these lines assume
perfect \gls{CSI} over the deformation space, which is precisely the
assumption that the present paper sets out to earn, rather than to postulate.

Despite this progress, channel estimation for \glspl{FIM} remains a critical
bottleneck. The deformation space is continuous and high-dimensional, and it
couples element-level morphing displacements with the array propagation
response in a non-linear fashion. Yang~\textit{et al.}~\cite{yang2025fim_arch}
cast channel estimation as a \gls{CS} problem and recover cascaded and direct
channel components via a variational sparse Bayesian learning algorithm.
Rezvani and Adve~\cite{xiao2025fim_ce} address channel estimation for
dynamically reconfigurable metasurface antennas by exploiting the fast
reconfiguration capability of the surface for multi-measurement diversity, and
propose practical \gls{MMSE}-based algorithms. Jiang~\textit{et
al.}~\cite{jiang2026fim_omp} retain the sparse-recovery formulation but choose
the surface shape so as to minimize the column coherence of the resulting
\gls{OMP} measurement matrix, which makes the morphing itself part of the
sensing design. Xiao~\textit{et al.}~\cite{xiao2026fim_neuralop} move away from
a fixed dictionary altogether and learn the continuous operator that maps a
\gls{FIM} shape to its channel response with a hierarchical Fourier neural
operator, so that the response at an unseen shape is interpolated, rather than
re-estimated. That last objective is the same one we pursue here, but our
route to it is a structural rather than a learned one, since the parametric
tensor model of Section~\ref{sec:sysmodel} extrapolates by evaluating an
identified wavevector at a new configuration and therefore needs neither
training data nor a shape-space dictionary.

\subsection{State-of-the-Art for Movable Antennas and Fluid Antenna Systems}

A parallel line of research replaces fixed antenna positions by flexibly
repositionable elements in order to exploit the small-scale spatial variation
of the channel. The \gls{MA} concept~\cite{zhu2024ma_magazine,zhu2024ma_multiuser}
allows each antenna to slide within a predefined region, and joint antenna-position
and beamforming optimization substantially improves \gls{SE} by aligning
elements with favorable channel conditions. The \gls{FAS}~\cite{wong2021fas}
pursues a similar objective by using liquid-metal or reconfigurable-port elements,
and the \gls{sixdma}~\cite{shao2024_6dma_magazine,shao2025_6dma_twc}
generalizes both to full position-and-rotation control of distributed
sub-arrays, which provides both translational and rotational spatial degrees of
freedom at once. Channel
estimation for these architectures exploits the multi-path field-response
structure, in that the angles of departure, angles of arrival, and complex gains
are jointly recovered by either \gls{CS} or deep learning~\cite{ma_ce_cs,ma_ce_dl}.
Shao~\textit{et al.}~\cite{shao2024_6dma_sparsity} exploit a directional
sparsity across the discrete set of candidate position-rotation pairs and
split the acquisition into a statistical stage and an instantaneous one, which
is the closest \gls{MA} counterpart to the twin-time-scale idea pursued here.
Their diversity is nonetheless indexed by an unstructured catalogue of poses,
whereas our slow grid is a rigid translation whose phase progression is
of Vandermonde nature by construction, so the pose index becomes an algebraic mode of a
tensor rather than a dictionary entry. A further architectural
difference is that \gls{MA} research targets transmit-side repositioning for
beamforming gain, whereas here the \gls{GFIM} operates at the receive side as
a passive aperture whose rigid-panel translation serves as a controlled
diversity mechanism, while element-level morphing modulates the normal-axis
projection of each wavevector.

\subsection{State-of-the-Art for Twin-Time-Scale Training}

Training overhead is one of the gravest practical limitations of
large-aperture intelligent surface-aided systems. The pilot duration required for
full instantaneous channel estimation scales linearly with the number of
surface elements and users, and it quickly exceeds the coherence time in dense
deployments~\cite{zheng2022ris_overhead}. Twin-time-scale frameworks address
this by separating a slow component of large coherence time, such as the
\gls{BS}-surface geometry or the statistical \gls{CSI}, from a fast component
of short coherence time such as the users' effective
channels~\cite{ris_twoscale_springer,ris_twoscale_nearfield}. Within one
slow-scale period, the fast-scale component is estimated several times at much
lower cost, which substantially reduces the aggregate pilot overhead relative
to single-scale baselines~\cite{ris_twoscale_nearfield}. Our work formalizes
an analogous twin-time-scale protocol for the \gls{GFIM}, in which the slow
scale is a rigid planar translation grid and the fast scale is a morphing
state, and shows that this structure directly induces a separable tensor model
that further compresses the estimation problem. The separation here is
mechanical rather than statistical, so both scales carry instantaneous
information and the slow loop is not relegated to second-order statistics.

\subsection{State-of-the-Art for Tensor-Based Channel Estimation}

Tensor decomposition provides a principled framework for exploiting the
multilinear structure of wireless channel
data~\cite{kolda2009tensor,sidiropoulos2017tensor}. The \gls{CP} model, also
known as \gls{PARAFAC}, entered sensor array processing through the
observation that multiple shift invariances of an array are exactly the
multilinear structure that the model captures~\cite{sidiropoulos2000sensor},
and it was subsequently developed into a unified modelling tool for wireless
receivers~\cite{dealmeida2007parafac}.
Tensor-based channel estimation has also been developed extensively for
\gls{RIS}-aided MIMO. De Ara\'{u}jo et al. formulate IRS-assisted MIMO channel
estimation using tensor models and exploit the resulting multilinear
structure~\cite{ris_tensor_ce_2021}. Ardah et al. introduce TRICE, an efficient
tensor-based channel-estimation framework for \gls{RIS}-aided mmWave
MIMO~\cite{trice_ris_2021}. For double-\gls{RIS} systems, coupled tensor
decompositions exploit shared factors to estimate the single- and
double-reflection links separately~\cite{double_ris_tensor_2025}. For
beyond-diagonal \gls{RIS}, one line of work uses block-Tucker models to decouple
channel estimates, while another fixes the group-connected topology and
applies a \gls{PARAFAC}-based alternating least-squares receiver to recover the
individual channels~\cite{bdris_tensor_tucker_2025,bdris_parafac_2026}.
Complementary semi-blind receivers jointly estimate channels and symbols for
beyond-diagonal \gls{RIS} using PARATUCK-to-\gls{PARAFAC} or Tucker models, and for hybrid
reflecting-and-sensing \gls{RIS} using tensor-coded signals observed at both the
surface and the \gls{BS}~\cite{bdris_semiblind_2026,hris_semiblind_2025}. Tensor
methods have also reached related apertures: \gls{RIS}-aided fluid-antenna systems
use \gls{PARAFAC} or nested-PARAFAC2 models to exploit joint surface and
antenna-port diversity, while TenSIM uses parity-dependent \gls{PARAFAC} and Tucker
models for stacked intelligent
metasurfaces~\cite{fluid_ris_semiblind_2026,tensim_2026}.
In parallel, tensor methods have also been applied to multiuser uplink \gls{CSI} estimation in mmWave massive MIMO systems, where the received pilot
tensor fits a low-rank \gls{CP} structure whose factor matrices encode spatial
frequencies, user directions, and path
gains~\cite{parafac_mmwave_multiuser,tensor_ce_mmwave_ofdm,parafac_timevary_2024}.
Identifiability of the \gls{CP} factorization under mild rank and k-rank
conditions provides stronger uniqueness guarantees than matrix-based methods
and enables blind or semi-blind recovery with reduced pilot
overhead~\cite{parafac_mmwave_multiuser}. Deep-learning-aided
\gls{ALS}~\cite{dl_als_mimo2024} further accelerates the convergence of the
\gls{CP} fit under noise. The \gls{ESPRIT} algorithm~\cite{roy1989esprit}, and its unitary
form~\cite{haardt1995unitary}, exploit the shift-invariance of
Vandermonde factor matrices to extract spatial frequencies in closed form,
which makes it a natural post-processing companion to the \gls{CP}
decomposition whenever the array geometry induces uniform phase progressions.
In parallel, \gls{MUSIC}~\cite{schmidt1986music,stoica1989music} provides a
subspace-based alternative that separates the signal and noise subspaces by
eigen-decomposition of a sample covariance matrix. Spatial
smoothing~\cite{shan1985smoothing} and its forward--backward
refinement~\cite{pillai1989fb} restore the rank of that covariance under
coherent or single-snapshot conditions, and the extension of both to the
tensor domain via multi-mode smoothing~\cite{haardt2008tensor_music} builds
higher-rank covariance matrices from structured sub-tensors, which improves
both the resolution and the robustness of multi-dimensional array settings.

Tensor methods have very recently reached the flexible-aperture literature
itself, and three contributions are close enough to the present one to deserve
an explicit comparison. Romano~\textit{et al.}~\cite{romano2026fim_parafac}
model a point-to-point MIMO link with active \gls{FIM} arrays at both ends and
morph the two surfaces in sequence, which yields a 4-way \gls{PARAFAC} model
that an \gls{ALS} receiver factors into steering matrices and path gains.
Alcantara~\textit{et al.}~\cite{alcantara2026mis} translate a small movable
metasurface layer over a larger fixed one and combine its positions with
electronically controlled phase patterns, which produces a 4-way \gls{PARAFAC}
model whose position-dependent factor is estimated blindly together with the
two cascaded channels. De Ara\'ujo~\textit{et
al.}~\cite{araujo2025ma_parafac2} also consider the multiuser uplink, but move to a
\gls{PARAFAC}2 receiver that estimates channels and symbols jointly from a
movable array at the \gls{BS}. By contrast, Tapie~\textit{et
al.}~\cite{tapie2026dma_tensor} show with an $18$~GHz prototype that a
block-Tucker \gls{ALS} recovers the channel of a strongly coupled \gls{DMA}.
Every one of these models draws its extra tensor modes from a single physical
mechanism, regardless of whether that is morphing, a translated sub-aperture, or an
electronic phase pattern, and each stops at the estimate.
Table~\ref{tab:novelty} boldly contrasts our contributions with the existing body of work. The
problem that remains open, namely the one this paper targets, is the case in which two
distinct mechanical motions act on the same aperture at two different
time-scales, the 4th-order \gls{CP} structure that their superposition
induces, the Kruskal conditions that this structure implies, and the use of
the identified parameters to select a panel configuration that was never
trained.

\subsection{Novel Contributions}

Against this background, the main contributions of this work are summarized as
follows.
\begin{enumerate}
\item We formalize a rigid-panel planar translation model for the \gls{GFIM}
  that is distinct both from the element-wise planar motion models studied
  elsewhere and from the single-axis morphing-only model of prior \gls{FIM}
  studies. Rigid translation on the slow time-scale and element morphing on the
  fast time-scale are superimposed on the same aperture, and together they
  yield a richer diversity space under fewer hardware constraints than either
  motion alone.
\item We exploit the separability of the rigid translation and the morphing
  displacement under narrowband far-field propagation conditions to derive a 4th-order
  \gls{CP} model for the channel tensor of each user, with modes indexed by
  the array elements, the morphing states, and the two planar-grid axes. To
  the best of our knowledge, no earlier tensor formulation draws its extra modes from
  two mechanical motions acting at two different time-scales.
\item We develop a pair of per-user estimation algorithms that operate on the
  pilot-separated 4-way tensor. The first combines \gls{ALS}-based \gls{CP}
  decomposition with paired \gls{ESPRIT}, which extracts the wavevector
  components from the Vandermonde structure of each factor matrix and
  therefore returns super-resolution parametric estimates. The second employs
  tensor spatial-smoothing \gls{MUSIC}, which reshapes the tensor to expose
  the \gls{UPA} structure and builds a multi-mode covariance matrix from
  overlapping sub-tensors, so that a subspace estimate is obtained without any
  explicit tensor factorization and without a convergence failure mode.
\item We equip both pipelines with a common \gls{MDL} front-end that
  estimates the model order from the eigenvalue spectrum of the largest-mode
  unfolding. We write the \gls{MDL} cost in the tensor context and show that
  the 4-way structure supplies an effective sample size that grows
  multiplicatively with the mode dimensions, which is what makes the rank
  decision reliable at the modest per-mode dimensions of a practical panel.
\item We derive closed-form Kruskal identifiability conditions that relate the
  system dimensions $(N,K,I,J)$ to the maximum number of recoverable multipath
  components, and we show that the structured model decouples the user
  separation step from the spatial estimation step. Channel reconstruction at
  arbitrary panel configurations beyond the training grid follows, and we turn
  that reconstruction into a configuration-selective multiuser \gls{TPC} whose
  search costs no pilot resource at all.
\end{enumerate}

The rest of this paper is organized as follows. Section~\ref{sec:sysmodel}
introduces the rigid-panel motion model and the twin-time-scale training
protocol. Furthermore, it derives the 4th-order \gls{CP} structure of the per-user channel
tensor. Section~\ref{sec:estimation} develops the two estimation pipelines
together with the \gls{MDL} front-end, the channel reconstruction rule, the
identifiability conditions and the training allocation that they imply. Section~\ref{sec:beamforming} formulates the
downlink multiuser precoding problem and the configuration search that the
parametric estimates enable. It also states the fixed-panel baseline and
the complexity of both stages. Section~\ref{sec:results} reports the numerical
campaign, in which each subsection tests a specific derivation of the two
preceding sections. Finally, Section~\ref{sec:conclusion} concludes the paper.

The following notation is used throughout the paper. Scalars are denoted by $a$, vectors by $\mathbf a$, matrices by $\mathbf A$, and tensors by calligraphic letters $\mathcal A$. The vectorization operator is $\vecop(\cdot)$, $\mathbf A_{(n)}$ denotes the mode-$n$ unfolding of the tensor $\mathcal A$, that is, the matrix whose columns are the mode-$n$ fibers of $\mathcal A$, and $\mathcal A\times_n\mathbf M$ is the mode-$n$ product of $\mathcal A$ by matrix $\mathbf M$. The Kronecker, Khatri--Rao, and Hadamard (element-wise) products are $\kron$, $\khatri$, and $\hadamard$, respectively. The \gls{CP} decomposition of a $d$-way tensor is written as $\mathcal X \approx \llbracket \mathbf A^{(1)},\dots,\mathbf A^{(d)}\rrbracket = \sum_{\ell=1}^R \mathbf a^{(1)}_\ell\circ \cdots \circ \mathbf a^{(d)}_\ell$. We consider $U$ uplink users and a receiver-side \gls{GFIM} with $N$ elements. The slow grid has $I$ positions along $x$ and $J$ positions along $z$, the fast morphing has $K$ states, and each user transmits $P$ pilot symbols. Indices are $n$ for elements, $k$ for fast states, and $(i,j)$ for slow grid
positions, and the imaginary unit is denoted by $\mathrm{j}$. Time indexing
follows a strict twin-time-scale protocol, in which the fast morphing index
$k\in\{1,\dots,K\}$ runs within a fixed planar location while the slow planar
grid indices $i\in\{1,\dots,I\}$ and $j\in\{1,\dots,J\}$ move the panel from one
location to the next. Each user $u\in\{1,\dots,U\}$ is seen through
$L_u$ propagation paths, a quantity that is unknown at the receiver and is
estimated in Section~\ref{sec:mdl}.

% Positioning of this work against the closest prior art.  Referenced from
% the Related Work section (Sec. I-A) and revisited in the Conclusion.
\begin{table*}[t]
\centering
\caption{Positioning of this work against the closest prior art. A dash marks a dimension the reference does not address.}
\label{tab:novelty}
\renewcommand{\arraystretch}{1.15}
\scriptsize
\setlength{\tabcolsep}{4pt}
\begin{tabular}{|l|p{0.20\textwidth}|c|p{0.27\textwidth}|c|c|c|}
\hline
\textbf{Ref.} & \textbf{Flexibility mechanism} & \textbf{Time-} &
\textbf{Signal model and estimator} & \textbf{Multi-} & \textbf{Off-grid} &
\textbf{Precoder}\\
 & & \textbf{scales} & & \textbf{user} & \textbf{prediction} & \textbf{design}\\
\hline
\cite{yang2025fim_arch}      & Element morphing along $y$          & One & Sparse recovery, variational SBL              & --  & --  & -- \\
\cite{xiao2025fim_ce}        & Electronic reconfiguration          & One & Multi-measurement \gls{MMSE}                  & --  & --  & -- \\
\cite{jiang2026fim_omp}      & Element morphing along $y$          & One & Coherence-designed \gls{OMP}                  & --  & --  & Yes \\
\cite{xiao2026fim_neuralop}  & Element morphing along $y$          & One & Learned neural operator                       & --  & Yes & -- \\
\cite{an2025fim_miso}        & Element morphing along $y$          & One & Perfect \gls{CSI} assumed                     & Yes & --  & Yes \\
\cite{romano2026fim_parafac} & Morphing at both link ends          & One & 4-way \gls{CP}, \gls{ALS}                     & --  & --  & -- \\
\cite{alcantara2026mis}      & Movable layer over a fixed layer    & One & 4-way \gls{CP}, \gls{ALS}                     & --  & --  & -- \\
\cite{araujo2025ma_parafac2} & \gls{MA} positions at the \gls{BS}  & One & \gls{PARAFAC}2, semi-blind                    & Yes & --  & -- \\
\cite{tapie2026dma_tensor}   & \gls{DMA}, coupled, static          & One & Block-Tucker \gls{ALS}                        & --  & --  & -- \\
\cite{shao2024_6dma_sparsity}& \gls{sixdma} position and rotation  & Two & Directional-sparsity recovery                 & Yes & --  & Yes \\
\cite{ris_twoscale_nearfield}& \gls{RIS} phases, static aperture   & Two & \glsentryshort{LS} on a statistical/fast split          & Yes & --  & Yes \\
\hline
This work                    & Rigid panel translation with element morphing
                                                                   & Two & 4-way \gls{CP}, with \gls{ALS}\,+\,\gls{ESPRIT} and tensor \gls{MUSIC}, \gls{MDL} rank
                                                                                                                   & Yes & Yes & Yes \\
\hline
\end{tabular}
\end{table*}

\section{System Model and Problem Formulation}
\label{sec:sysmodel}

In this section we first introduce the rigid-panel motion model of the
\gls{GFIM} and the twin-time-scale training protocol that it induces. We then
derive the 4th-order \gls{CP} structure of the per-user channel tensor and
state the pair of problems that the remainder of the paper solves.

We consider a \gls{BS} equipped with a \gls{GFIM} comprising $N$ elements. The panel is subject to rigid planar translation over a two-dimensional grid with $I$ steps along $x$ and $J$ steps along $z$ (the slow time scale), where each element can undergo morphing displacements along the normal direction $y$ across $K$ states (the fast time scale). Element positions are described relative to a reference configuration and the common panel displacement. We formalize this next.

Let the reference, untranslated, planar coordinates of element $n$ be
\begin{equation}
\mathbf p_n^{(0)}=\begin{bmatrix}x_n^{(0)}\\ 0 \\ z_n^{(0)}\end{bmatrix}\in\bbR^3,\qquad n=1,\dots,N.
\label{eq:B0_refpos}
\end{equation}
At slow planar grid position $(i,j)$, the entire panel is translated by the common displacement of
\begin{equation}
\boldsymbol{\delta}[i,j]\triangleq \begin{bmatrix}\Delta_x[i]\\ 0\\ \Delta_z[j]\end{bmatrix}\in\bbR^3,
\qquad i=1,\dots,I,\ j=1,\dots,J,
\label{eq:B0_panelshift}
\end{equation}
where $(\Delta_x[i],\Delta_z[j])$ are the known controlled translation steps. At fast morphing state $k$, element $n$ experiences a normal-direction displacement $y_n[k]$. Thus, the element position at configuration $(k,i,j)$ is
\begin{equation}
\boxed{
\mathbf p_n[k,i,j]=\mathbf p_n^{(0)}+\boldsymbol{\delta}[i,j]+\begin{bmatrix}0\\ y_n[k]\\ 0\end{bmatrix}.
}
\label{eq:B0_fullpos}
\end{equation}
Hence, the planar motion is rigid and common to all elements, whereas the
morphing is element-dependent in general. Throughout this paper we restrict the
morphing to a rigid piston mode, in which the displacement is common to all
elements at each state $k$, so that $y_n[k]=y[k]$ for every $n$ and the
morphing phase factor depends only on $k$. This is the assumption that makes
the fourth mode separable from the element mode, and it is what distinguishes
the fast motion of the present protocol from the free surface deformation of
the \gls{FIM} literature~\cite{yang2025fim_arch,an2025fim_miso}. A
non-uniform profile $y_n[k]$ would couple the $n$ and $k$ indices and reduce
the model of Eq.~\eqref{eq:B_CP_user} to a 3-way one based on the pair $(n,k)$.

For each planar grid location $(i,j)$ in the slow loop, the panel remains fixed on the $x$-$z$ plane, while the morphing index $k$ cycles through all $K$ states in the fast loop. This yields training configurations indexed by $(k,i,j)$, and the total number of configurations is $K I J$. At each configuration, the $U$ users transmit uplink pilots. Each user $u$ transmits a pilot vector $\mathbf s_u\in\bbC^{P}$, a sequence of $P$ symbols, and we collect them as $\mathbf S=[\mathbf s_1,\dots,\mathbf s_U]\in\bbC^{P\times U}$. Throughout, we assume the pilots are mutually orthogonal, $\mathbf S^H\mathbf S=P\,\mathbf I_U$ with $P\geq U$, so the users are separated at the receiver by a matched-filter projection along the pilot mode, as detailed in Sec.~\ref{sec:ortho_pilots}. The model is therefore developed per user, and the remainder of this section focuses on the 4th-order tensor of a single user. At configuration $(k,i,j)$, the signal received at the $N$ \gls{GFIM} elements is
\begin{equation}
\mathbf Y[k,i,j]=\sum_{u=1}^U \sum_{\ell=1}^{L_u} \beta_{u,\ell}\, \mathbf a_{u,\ell}(k,i,j)\, \mathbf s_u^T+\mathbf N[k,i,j],
\label{eq:B_Y}
\end{equation}
where $\mathbf Y[k,i,j]\in\bbC^{N\times P}$ and $\mathbf N[k,i,j]$ is noise. Each user $u$ has $L_u$ paths, where $\beta_{u,\ell}\in\bbC$ is the complex gain of path $\ell$ of user $u$ and $\mathbf a_{u,\ell}(k,i,j)\in\bbC^N$ is the steering vector for that path at $(k,i,j)$, as defined below. The channel vector of user $u$ is $\mathbf h_u[k,i,j]=\sum_{\ell=1}^{L_u}\beta_{u,\ell}\,\mathbf a_{u,\ell}(k,i,j)$, so that $\mathbf Y[k,i,j]=\mathbf H[k,i,j]\mathbf S^T+\mathbf N[k,i,j]$ with $\mathbf H[k,i,j]=[\mathbf h_1[k,i,j],\dots,\mathbf h_U[k,i,j]]$.

We assume multipath propagation and model each user's channel as a sum over steering vectors weighted by path gains. The channel vector of user $u$ at configuration $(k,i,j)$ is
\begin{equation}
\mathbf h_u[k,i,j]=\sum_{\ell=1}^{L_u}\beta_{u,\ell}\,\mathbf a_{u,\ell}(k,i,j),
\label{eq:B_geom}
\end{equation}
where $\beta_{u,\ell}\in\bbC$ is the complex gain of path $\ell$ of user $u$ and $\mathbf a_{u,\ell}(k,i,j)\in\bbC^N$ is the steering vector for that path at $(k,i,j)$. Using the element position from Eq.~\eqref{eq:B0_fullpos}, the $n$-th entry of the steering vector for path $\ell$ of user $u$ can be written as a product of three exponentials. Under a far-field plane-wave assumption, with wavevector $\mathbf q_{u,\ell}$ for path $\ell$ of user $u$, we have:
\begin{equation}
\begin{split}
[\mathbf a_{u,\ell}(k,i,j)]_n&=\exp\!\left(-\mathrm{j}\mathbf q_{u,\ell}^T\boldsymbol{\delta}[i,j]\right)\exp\!\left(-\mathrm{j}\mathbf q_{u,\ell}^T\mathbf p_n^{(0)}\right)\\
&\quad\times\exp\!\left(-\mathrm{j} q^{(y)}_{u,\ell} y_n[k]\right).
\end{split}
\label{eq:B_planewave_fact}
\end{equation}
The three factors correspond to the decomposition of $\mathbf p_n[k,i,j]$ in Eq.~\eqref{eq:B0_fullpos}: reference position $\mathbf p_n^{(0)}$, morphing $y_n[k]$ along $y$, and planar translation $\boldsymbol{\delta}[i,j]$. The planar factor splits into its two coordinates. Writing $\mathbf q_{u,\ell}=[q^{(x)}_{u,\ell},q^{(y)}_{u,\ell},q^{(z)}_{u,\ell}]^T$ and using the panel shift of Eq.~\eqref{eq:B0_panelshift}, which is additive in $x$ and $z$, the planar phase becomes $\mathbf q_{u,\ell}^T\boldsymbol{\delta}[i,j]=q^{(x)}_{u,\ell}\Delta_x[i]+q^{(z)}_{u,\ell}\Delta_z[j]$ and therefore factors exactly into an $i$-dependent and a $j$-dependent term. The steering entry is thus a product of four scalar phases, one per index, and we accordingly define the spatial vector $\mathbf e_{u,\ell,n}\in\bbC^N$ with $[\mathbf e_{u,\ell,n}]_n=\exp(-\mathrm{j}\mathbf q_{u,\ell}^T\mathbf p_n^{(0)})$, the morphing vector $\mathbf e_{u,\ell,k}\in\bbC^K$ with $[\mathbf e_{u,\ell,k}]_k=\exp(-\mathrm{j} q^{(y)}_{u,\ell} y[k])$, and the two planar vectors $\mathbf e_{u,\ell,i}\in\bbC^{I}$, $\mathbf e_{u,\ell,j}\in\bbC^{J}$ with $[\mathbf e_{u,\ell,i}]_i=\exp(-\mathrm{j} q^{(x)}_{u,\ell}\Delta_x[i])$ and $[\mathbf e_{u,\ell,j}]_j=\exp(-\mathrm{j} q^{(z)}_{u,\ell}\Delta_z[j])$. Note that this four-way separation requires no further assumption beyond the
piston morphing already stated, since it follows directly from the rigid-panel
model of Eq.~\eqref{eq:B0_fullpos}, whose translation term is additive and
common to every element. With these definitions, $[\mathbf a_{u,\ell}(k,i,j)]_n=[\mathbf e_{u,\ell,n}]_n\,[\mathbf e_{u,\ell,k}]_k\,[\mathbf e_{u,\ell,i}]_i\,[\mathbf e_{u,\ell,j}]_j$, so the array of entries over $(n,k,i,j)$ is the Kronecker product of the four vectors,
\begin{equation}
\mathbf a_{u,\ell}(\cdot,\cdot,\cdot)=\mathbf e_{u,\ell,n}\kron\mathbf e_{u,\ell,k}\kron\mathbf e_{u,\ell,i}\kron\mathbf e_{u,\ell,j}.
\label{eq:B_kron_steering}
\end{equation}
Thus the steering structure yields a 4th-order tensor with modes $N$, $K$, $I$, $J$.

We next substitute the channel and steering model, Eqs.~\eqref{eq:B_geom} and~\eqref{eq:B_planewave_fact}, into Eq.~\eqref{eq:B_Y}. The noiseless part of $\mathbf Y[k,i,j]$ is a sum over users and paths of terms $\beta_{u,\ell}\,\mathbf a_{u,\ell}(k,i,j)\,\mathbf s_u^T$, with $\mathbf a_{u,\ell}$ given by the Kronecker structure in Eq.~\eqref{eq:B_kron_steering}. For each user $u$, we collect the channel vectors $\mathbf h_u[k,i,j]$ over all configurations $(k,i,j)$ into a 4-way tensor
\begin{equation}
\mathcal H_u\in\bbC^{N\times K\times I\times J},\qquad \mathcal H_u(:,k,i,j)=\mathbf h_u[k,i,j].
\label{eq:B_tensor_user}
\end{equation}
From the four-vector Kronecker structure in Eq.~\eqref{eq:B_kron_steering}, the four modes $N$, $K$, $I$, $J$ are separable and $\mathcal H_u$ is a sum of $L_u$ rank-1 tensors, i.e., a 4th-order \gls{PARAFAC} of rank at most $L_u$,
\begin{equation}
\begin{split}
\mathcal H_u&=\sum_{\ell=1}^{L_u}\beta_{u,\ell}\,\mathbf e_{u,\ell,n}\circ\mathbf e_{u,\ell,k}\circ\mathbf e_{u,\ell,i}\circ\mathbf e_{u,\ell,j}\\
&=\llbracket \boldsymbol\beta_u;\mathbf A_u,\mathbf B_u,\mathbf C_u,\mathbf D_u\rrbracket,
\end{split}
\label{eq:B_CP_user}
\end{equation}
where $\boldsymbol\beta_u=[\beta_{u,1},\dots,\beta_{u,L_u}]^T$ collects the path gains and the four factor matrices stack the corresponding vectors over $\ell$, one per mode: $\mathbf A_u=[\mathbf e_{u,1,n},\dots,\mathbf e_{u,L_u,n}]\in\bbC^{N\times L_u}$ for the spatial mode, and likewise $\mathbf B_u\in\bbC^{K\times L_u}$, $\mathbf C_u\in\bbC^{I\times L_u}$ and $\mathbf D_u\in\bbC^{J\times L_u}$ for the morphing and the two planar modes. Each factor thus carries a single subscript, namely the user index.
Fig.~\ref{fig:tensor_cp} provides a graphical representation of this
decomposition.

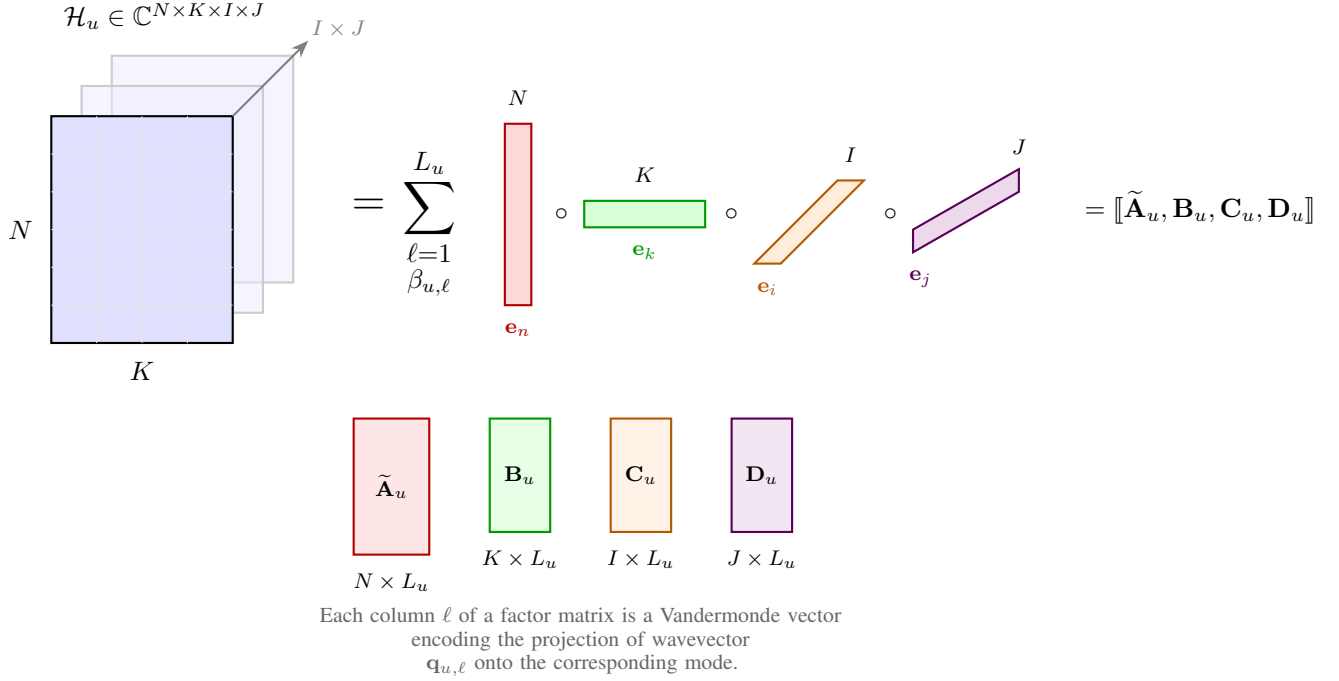
\begin{figure*}[t]
\centering
\begin{tikzpicture}[
  >=Stealth,
  every node/.style={font=\small},
  annot/.style={font=\scriptsize, text=black!60},
]

% ===== 4-D Tensor H_u =====
\begin{scope}[xshift=0cm]
  % Back slabs for depth (I x J)
  \foreach \d in {0.8, 0.4} {
    \fill[blue!8, opacity=0.5]
      (\d+0.0, \d+0.0) rectangle (\d+2.4, \d+3.0);
    \draw[thick, gray!40]
      (\d+0.0, \d+0.0) rectangle (\d+2.4, \d+3.0);
  }
  % Front slab
  \fill[blue!12] (0,0) rectangle (2.4,3.0);
  \draw[thick] (0,0) rectangle (2.4,3.0);

  % Grid lines (N rows, K columns)
  \foreach \y in {0.5, 1.0, 1.5, 2.0, 2.5} {
    \draw[gray!20, thin] (0,\y) -- (2.4,\y);
  }
  \foreach \x in {0.6, 1.2, 1.8} {
    \draw[gray!20, thin] (\x,0) -- (\x,3.0);
  }

  % Mode labels
  \node[left] at (-0.1, 1.5) {$N$};
  \node[below] at (1.2, -0.1) {$K$};
  \draw[->, gray, thick] (2.4, 3.0) -- (3.4, 4.0);
  \node[above right, font=\scriptsize, gray] at (3.3, 3.9) {$I \times J$};

  % Tensor name
  \node[above] at (1.5, 4.0) {$\mathcal H_u \in \mathbb C^{N \times K \times I \times J}$};
\end{scope}

% ===== Equals sign =====
\node[font=\Large] at (4.2, 1.8) {$=$};

% ===== Sum of rank-1 terms =====
\node[font=\large] at (5.0, 1.8) {$\displaystyle\sum_{\ell=1}^{L_u}$};
\node[font=\small] at (5.0, 0.8) {$\beta_{u,\ell}$};

% ===== Outer product vectors =====
% Orientation follows the tensor drawing: N vertical, K horizontal, and the
% two planar modes I and J along two depth diagonals.
% Vector e_n (vertical)
\begin{scope}[xshift=6.0cm, yshift=0.5cm]
  \fill[red!15] (0, 0) rectangle (0.35, 2.4);
  \draw[thick, red!70!black] (0, 0) rectangle (0.35, 2.4);
  \node[below, font=\scriptsize, red!70!black] at (0.175, -0.1) {$\mathbf e_{n}$};
  \node[above, font=\scriptsize] at (0.175, 2.5) {$N$};
\end{scope}

\node at (6.75, 1.7) {$\circ$};

% Vector e_k (horizontal)
\begin{scope}[xshift=7.05cm, yshift=1.53cm]
  \fill[green!15] (0, 0) rectangle (1.6, 0.35);
  \draw[thick, green!60!black] (0, 0) rectangle (1.6, 0.35);
  \node[below, font=\scriptsize, green!60!black] at (0.8, -0.1) {$\mathbf e_{k}$};
  \node[above, font=\scriptsize] at (0.8, 0.45) {$K$};
\end{scope}

\node at (9.0, 1.7) {$\circ$};

% Vector e_i (depth)
\begin{scope}[xshift=9.3cm, yshift=1.05cm]
  \fill[orange!15] (0,0) -- (0.35,0) -- (1.45,1.1) -- (1.1,1.1) -- cycle;
  \draw[thick, orange!70!black] (0,0) -- (0.35,0) -- (1.45,1.1) -- (1.1,1.1) -- cycle;
  \node[below, font=\scriptsize, orange!70!black] at (0.175, -0.1) {$\mathbf e_{i}$};
  \node[above, font=\scriptsize] at (1.275, 1.2) {$I$};
\end{scope}

\node at (11.1, 1.7) {$\circ$};

% Vector e_j (depth, along a shallower diagonal than e_i)
\begin{scope}[xshift=11.4cm, yshift=1.2cm]
  \fill[violet!15] (0,0) -- (1.4,0.8) -- (1.4,1.1) -- (0,0.3) -- cycle;
  \draw[thick, violet!70!black] (0,0) -- (1.4,0.8) -- (1.4,1.1) -- (0,0.3) -- cycle;
  \node[below, font=\scriptsize, violet!70!black] at (0.1, -0.1) {$\mathbf e_{j}$};
  \node[above, font=\scriptsize] at (1.4, 1.15) {$J$};
\end{scope}

% ===== CP notation =====
\node[font=\small] at (15.2, 1.8) {$= \llbracket \widetilde{\mathbf A}_u, \mathbf B_u, \mathbf C_u, \mathbf D_u \rrbracket$};

% ===== Factor matrices below =====
\begin{scope}[yshift=-2.8cm, xshift=4.0cm]
  % A_u
  \fill[red!10] (0,0) rectangle (1.0, 1.8);
  \draw[thick, red!70!black] (0,0) rectangle (1.0, 1.8);
  \node[font=\scriptsize] at (0.5, 0.9) {$\widetilde{\mathbf A}_u$};
  \node[below, font=\scriptsize] at (0.5, -0.1) {$N \times L_u$};

  % B_u
  \fill[green!10] (1.8,0.3) rectangle (2.6, 1.8);
  \draw[thick, green!60!black] (1.8,0.3) rectangle (2.6, 1.8);
  \node[font=\scriptsize] at (2.2, 1.05) {$\mathbf B_u$};
  \node[below, font=\scriptsize] at (2.2, 0.2) {$K \times L_u$};

  % C_u
  \fill[orange!10] (3.4,0.3) rectangle (4.2, 1.8);
  \draw[thick, orange!70!black] (3.4,0.3) rectangle (4.2, 1.8);
  \node[font=\scriptsize] at (3.8, 1.05) {$\mathbf C_u$};
  \node[below, font=\scriptsize] at (3.8, 0.2) {$I \times L_u$};

  % D_u
  \fill[violet!10] (5.0,0.3) rectangle (5.8, 1.8);
  \draw[thick, violet!70!black] (5.0,0.3) rectangle (5.8, 1.8);
  \node[font=\scriptsize] at (5.4, 1.05) {$\mathbf D_u$};
  \node[below, font=\scriptsize] at (5.4, 0.2) {$J \times L_u$};

  % Annotation
  \node[annot, text width=7cm, align=center] at (3.0, -1.15)
    {Each column $\ell$ of a factor matrix is a Vandermonde vector\\
     encoding the projection of wavevector $\mathbf q_{u,\ell}$ onto the corresponding mode.};
\end{scope}

\end{tikzpicture}
\caption{4th-order \gls{CP} structure of the per-user channel tensor $\mathcal H_u$ of Eq.~\eqref{eq:B_CP_user}, with the factor matrices of Eq.~\eqref{eq:B_CP_user_weighted_equiv}.}
\label{fig:tensor_cp}
\end{figure*}

We define the training phase as the period during which, at each configuration $(k,i,j)$, every user $u$ transmits its pilot sequence $\mathbf s_u\in\bbC^P$ and the \gls{BS} collects the received matrix $\mathbf Y[k,i,j]\in\bbC^{N\times P}$ over all $K I J$ configurations. Thus training is characterized by the twin-time-scale protocol, in which the slow loop sweeps the planar grid and the fast loop sweeps the morphing states, with simultaneous multiuser pilot transmission at each configuration, so that one received block is obtained per configuration. Stacking these matrices into a 5-way tensor, we obtain the received-signal tensor
\begin{equation}
\mathcal Y\in\bbC^{N\times P\times K\times I\times J},\qquad \mathcal Y(:,:,k,i,j)=\mathbf Y[k,i,j],
\label{eq:B_tensor_Y}
\end{equation}
i.e., the $(n,p,k,i,j)$-entry is the signal associated with element $n$, pilot symbol index $p$, and configuration $(k,i,j)$. Let $\mathcal H\in\bbC^{N\times U\times K\times I\times J}$ denote the multiuser channel tensor with $\mathcal H(:,u,k,i,j)=\mathbf h_u[k,i,j]$, and let $\mathcal N\in\bbC^{N\times P\times K\times I\times J}$ be the noise tensor with $\mathcal N(:,:,k,i,j)=\mathbf N[k,i,j]$. The noiseless received slice at $(k,i,j)$ is $\mathbf H[k,i,j]\mathbf S^T=\sum_{u=1}^U \mathbf h_u[k,i,j]\mathbf s_u^T$, so entrywise $\mathcal Y_0(n,p,k,i,j)=\sum_{u=1}^U \mathcal H(n,u,k,i,j)\,[\mathbf S]_{p,u}$. Upon using the mode-$2$ (user) product of $\mathcal H$ with $\mathbf S^T\in\bbC^{P\times U}$, we have
\begin{equation}
\mathcal Y=\mathcal Y_0+\mathcal N,\qquad \mathcal Y_0=\mathcal H\times_2\mathbf S^T,
\label{eq:B_Y_tensor}
\end{equation}
where $\mathcal Y_0\in\bbC^{N\times P\times K\times I\times J}$ and $(\mathcal H\times_2\mathbf S^T)(n,p,k,i,j)=\sum_{u=1}^U \mathcal H(n,u,k,i,j)\,[\mathbf S]_{p,u}$, with $[\mathbf S]_{p,u}$ the $(p,u)$-entry of $\mathbf S$. Eq.~\eqref{eq:B_Y_tensor} is the tensor expression for the received training signal, in which the multiuser channel tensor $\mathcal H$ is contracted along the user mode with the pilot matrix $\mathbf S$ to produce the 5-way received tensor whose observations are used for channel estimation.

\begin{figure*}[t]
\centering
\resizebox{\textwidth}{!}{%
\begin{tikzpicture}[
  font=\footnotesize,
  >={Latex[length=2.2mm]},
  slowblock/.style={draw, thick, fill=blue!10, minimum height=\slowH, minimum width=\slowW, rounded corners=1pt},
  fastslot/.style={draw, fill=orange!15, minimum height=\fastH, minimum width=\fastW, rounded corners=1pt},
  pilotblock/.style={draw, fill=green!12, minimum height=\pilotH, minimum width=\pilotW, rounded corners=1pt},
  rxblock/.style={draw, thick, fill=yellow!18, minimum height=8mm, minimum width=18mm, rounded corners=2pt},
  braceup/.style={decorate, decoration={brace, amplitude=4pt, raise=2pt}},
  bracedn/.style={decorate, decoration={brace, amplitude=4pt, mirror, raise=2pt}}
]

% Layout parameters (single-point control)
\def\xL{0.9}
\def\xS{1.75}
\def\xM{4.45}
\def\xR{8.25}
\def\ySlow{4.00}
\def\yFast{1.95}
\def\yPilot{-0.35}
\def\titleOffSlow{0.95}
\def\titleOffFast{0.68}
\def\titleOffPilot{0.95}
\def\braceOffDown{0.48}
\def\rightX{11.35}
\def\rightTopY{2.45}
\def\rightBotY{-0.95}
\def\rightScale{0.62}

\def\slowW{13mm}
\def\slowH{7.5mm}
\def\fastW{12mm}
\def\fastH{7mm}
\def\pilotW{11mm}
\def\pilotH{6mm}

% Slow loop
\node[font=\small\bfseries, anchor=west, align=left, text width=9.0cm]
  at (-0.2,\ySlow+\titleOffSlow)
  {Slow time-scale: rigid panel translation over grid positions};

\node[slowblock]               (s11)   at (\xL,\ySlow) {$(1,1)$};
\node[slowblock]               (s12)   at (\xL+\xS,\ySlow) {$(1,2)$};
\node                                  at (\xL+2*\xS,\ySlow) {$\cdots$};
\node[slowblock, fill=blue!25] (sij)   at (\xM+\xS,\ySlow) {$(i,j)$};
\node                                  at (\xM+2*\xS,\ySlow) {$\cdots$};
\node[slowblock]               (sNxNz) at (\xR,\ySlow) {$(I,J)$};

\draw[bracedn] (\xL-0.65,\ySlow-\braceOffDown) -- (\xR+0.65,\ySlow-\braceOffDown)
  node[midway,below=5pt] {$IJ$ panel positions};

% Fast loop zoom
\draw[->,thick,densely dashed,blue!60]
  (sij.south) -- ++(0,-0.30) -| (\xL+0.6,\yFast+0.35);
\draw[->,thick,densely dashed,blue!60]
  (sij.south) -- ++(0,-0.30) -| (\xR+0.2,\yFast+0.35);

\node[font=\small\bfseries, anchor=west, align=left, text width=9.0cm]
  at (-0.2,\yFast+\titleOffFast)
  {Fast time-scale at fixed $(i,j)$: morphing states};

\node[fastslot]                (f1) at (\xL+0.8,\yFast) {$k=1$};
\node[fastslot]                (f2) at (\xL+2.45,\yFast) {$k=2$};
\node                                 at (\xL+3.75,\yFast) {$\cdots$};
\node[fastslot, fill=orange!30](fk) at (\xM+1.55,\yFast) {$k$};
\node                                 at (\xM+2.75,\yFast) {$\cdots$};
\node[fastslot]                (fK) at (\xR+0.1,\yFast) {$k=K$};

\draw[bracedn] (\xL+0.1,\yFast-\braceOffDown) -- (\xR+0.75,\yFast-\braceOffDown)
  node[midway,below=5pt] {$K$ morphing states};

% Pilot transmission at one configuration
\draw[->,thick,densely dashed,orange!70]
  (fk.south) -- ++(0,-0.30) -| (\xL+0.05,\yPilot+0.25);
\draw[->,thick,densely dashed,orange!70]
  (fk.south) -- ++(0,-0.30) -| (\xM+1.35,\yPilot+0.25);

\node[font=\small\bfseries, anchor=west, align=left, text width=9.0cm]
  at (-0.2,\yPilot+\titleOffPilot)
  {At configuration $(k,i,j)$: simultaneous pilot transmission};

\node[pilotblock, fill=red!15]    (p1) at (\xL,\yPilot) {$\mathbf s_1^T$};
\node[pilotblock, fill=cyan!15]   (p2) at (\xL+1.6,\yPilot) {$\mathbf s_2^T$};
\node                                   at (\xL+2.8,\yPilot) {$\cdots$};
\node[pilotblock, fill=violet!15] (pU) at (\xL+4.2,\yPilot) {$\mathbf s_U^T$};

\node[rxblock] (rx) at (\xR+0.65,\yPilot) {$\mathbf Y[k,i,j]$};
\draw[->,very thick] (pU.east)
  -- node[above,font=\scriptsize,inner sep=1pt] {$+\,\mathbf N[k,i,j]$} (rx.west);
\node[font=\scriptsize] at (\xR+0.65,\yPilot-0.55) {$\in\mathbb C^{N\times P}$};

\draw[bracedn] (\xL-0.65,\yPilot-0.45) -- (\xM+1.15,\yPilot-0.45)
  node[midway,below=5pt,font=\scriptsize,align=center]
  {$U$ users transmit pilot sequences of length $P$};

% Right-side panel grid illustration
\begin{scope}[shift={(\rightX,\rightTopY)}, scale=\rightScale]
  \node[font=\small\bfseries] at (1.55,4.0) {Panel grid ($x$-$z$ plane)};
  \draw[step=0.55, gray!30, thin] (0,0) grid (3.3,3.3);
  \draw[->,thick] (-0.25,0) -- (3.8,0) node[right,font=\scriptsize] {$x$};
  \draw[->,thick] (0,-0.25) -- (0,3.8) node[above,font=\scriptsize] {$z$};

  \foreach \ii in {0,...,6}
    \foreach \jj in {0,...,6}
      \fill[gray!50] (\ii*0.55,\jj*0.55) circle (1.0pt);

  \fill[blue!18,rounded corners=1pt] (1.1,1.65) rectangle (1.65,2.2);
  \draw[blue!70,thick] (1.1,1.65) rectangle (1.65,2.2);
  \fill[blue!75] (1.375,1.925) circle (2.3pt);
  \node[font=\scriptsize, blue!80, right] at (1.78,1.93) {$(i,j)$};

  \draw[gray!65, dashed, thick, rounded corners=2pt]
    (0,0) -- (3.3,0) -- (3.3,0.55) -- (0,0.55) -- (0,1.1);
  \draw[bracedn] (0,-0.35) -- (3.3,-0.35)
    node[midway,below=4pt,font=\scriptsize] {$I$ steps};
  \node[font=\scriptsize, rotate=90] at (-0.72,1.65) {$J$ steps};
\end{scope}

% Right-side morphing sketch
\begin{scope}[shift={(\rightX-0.35,\rightBotY)}, scale=\rightScale]
  \node[font=\small\bfseries] at (2.25,2.1) {Morphing snapshot (state $k$)};
  \draw[thick] (0,0) -- (4.3,0);
  \node[font=\scriptsize,anchor=east] at (-0.08,0) {panel};

  \foreach \xx/\dy in {0.35/0.25,0.95/0.58,1.55/0.18,2.15/0.74,2.75/0.40,3.35/0.55,3.95/0.30} {
    \fill[orange!55] (\xx,0) circle (2pt);
    \draw[->,orange!85,thick] (\xx,0.05) -- (\xx,\dy);
    \fill[orange!90] (\xx,\dy) circle (2pt);
  }
  \draw[->,thick] (4.7,0) -- (4.7,1.0) node[right,font=\scriptsize] {$y$};
  \node[font=\scriptsize] at (2.15,-0.48) {$y_n[k]$: normal displacement};
\end{scope}

\end{tikzpicture}%
}
\caption{Twin-time-scale training frame. Slow loop over the $IJ$ planar positions, fast loop over the $K$ morphing states, and simultaneous pilot transmission at each configuration $(k,i,j)$.}
\label{fig:training_frame}
\end{figure*}
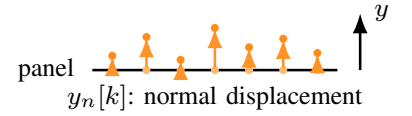

Figure~\ref{fig:training_frame} summarizes the frame hierarchy behind this construction. The panel grid view corresponds to the slow loop over planar translations, i.e., indices $(i,j)$ with cardinalities $(I,J)$, while the morphing snapshot represents one fast-loop state $k\in\{1,\dots,K\}$ at fixed $(i,j)$, so that each complete sweep produces $IJK$ configurations. At each configuration, users transmit pilot sequences of length $P$, yielding one received matrix $\mathbf Y[k,i,j]\in\bbC^{N\times P}$. This maps directly onto the modes of $\mathcal Y$ in Eq.~\eqref{eq:B_tensor_Y}: mode-$1$ ($N$) is the element index, mode-$2$ ($P$) is the pilot-symbol index within a block, mode-$3$ ($K$) is the fast morphing index, and modes-$4$/$5$ ($I,J$) are the slow planar-grid indices. Normalizing the resultant training time by the symbol duration $T_s$ yields the number of training resources
\begin{equation}
\frac{T_{\mathrm{train}}}{T_s}=IJKP.
\label{eq:B_train_time}
\end{equation}

\subsection{Mode-wise \glsentryshort{PARAFAC} Model}
For each user $u$, we rewrite the rank-$L_u$ channel tensor in Eq.~\eqref{eq:B_CP_user} as
\begin{equation}
\mathcal H_u=\mathcal G_u\times_1\mathbf A_u\times_2\mathbf B_u\times_3\mathbf C_u\times_4\mathbf D_u,
\label{eq:B_CP_user_weighted}
\end{equation}
where $\mathcal G_u\in\bbC^{L_u\times L_u\times L_u\times L_u}$ is a superdiagonal core tensor whose only nonzero entries are
\begin{equation}
[\mathcal G_u]_{r,r,r,r}=\beta_{u,r},\qquad r=1,\dots,L_u.
\label{eq:B_core_superdiag}
\end{equation}
Because of the superdiagonal core, this is equivalent to the weighted-factor \gls{CP} form
\begin{equation}
\mathcal H_u=\llbracket \widetilde{\mathbf A}_u,\mathbf B_u,\mathbf C_u,\mathbf D_u\rrbracket,
\qquad \widetilde{\mathbf A}_u\triangleq \mathbf A_u\,\mathrm{Diag}(\boldsymbol\beta_u),
\label{eq:B_CP_user_weighted_equiv}
\end{equation}
in which the gains are absorbed into the spatial factor. Following this convention, $\mathbf H_{u,(n)}$ is the mode-$n$ unfolding of $\mathcal H_u$. For mode ordering $(N,K,I,J)$, the unfoldings are
\begin{equation}
\mathbf H_{u,(1)}=\widetilde{\mathbf A}_u\left(\mathbf D_u\khatri\mathbf C_u\khatri\mathbf B_u\right)^T,
\label{eq:B_unfold_mode1}
\end{equation}
\begin{equation}
\mathbf H_{u,(2)}=\mathbf B_u\left(\mathbf D_u\khatri\mathbf C_u\khatri\widetilde{\mathbf A}_u\right)^T,
\label{eq:B_unfold_mode2}
\end{equation}
\begin{equation}
\mathbf H_{u,(3)}=\mathbf C_u\left(\mathbf D_u\khatri\mathbf B_u\khatri\widetilde{\mathbf A}_u\right)^T,
\label{eq:B_unfold_mode3}
\end{equation}
\begin{equation}
\mathbf H_{u,(4)}=\mathbf D_u\left(\mathbf C_u\khatri\mathbf B_u\khatri\widetilde{\mathbf A}_u\right)^T.
\label{eq:B_unfold_mode4}
\end{equation}
These four unfoldings are the mode-wise \gls{PARAFAC} expressions used for estimating the factors by alternating updates on each mode.

\subsection{Problem Formulation}
\label{sec:problem}

The paper addresses a pair of coupled problems. The first is the acquisition of the
per-user channel parameters from the training tensor, and the second is the
selection of the panel configuration and of the \gls{TPC} that the acquired
parameters make possible.

Let $\boldsymbol\Theta_u \triangleq \{(\mathbf q_{u,\ell},\beta_{u,\ell})\}_{\ell=1}^{L_u}$
collect the wavevectors and complex gains of user $u$, and let
$\mathcal H_u(\boldsymbol\Theta_u)$ denote the channel tensor that
Eq.~\eqref{eq:B_CP_user} generates from them. Given the received tensor
$\mathcal Y$ of Eq.~\eqref{eq:B_Y_tensor} and the known pilot matrix
$\mathbf S$, the estimation problem is formulated as:
\begin{equation}
\begin{aligned}
  \textrm{(P1)}\quad
  \min_{\{L_u,\boldsymbol\Theta_u\}_{u=1}^{U}} \;&
    \sum_{u=1}^{U}
    \big\lVert \widetilde{\mathcal H}_u
      - \mathcal H_u(\boldsymbol\Theta_u)\big\rVert_F^2 \\
  \textrm{s.t.}\quad
  & \lVert \mathbf q_{u,\ell}\rVert = k_0,\quad
    q^{(y)}_{u,\ell} \leq 0, \\
  & L_u \leq L_u^{\max}(N,K,I,J),
\end{aligned}
\label{eq:B_P1}
\end{equation}
where $\widetilde{\mathcal H}_u$ is the despread observation of
Eq.~\eqref{eq:C_ortho_pilots}. The first constraint places every path on the
visible hemisphere of the far-field manifold, and the second restricts the
model order to the identifiability range that Section~\ref{sec:uniqueness}
derives. Problem~(P1) is non-convex in $\boldsymbol\Theta_u$ and it is coupled
with the discrete unknown $L_u$, so it is solved in two stages: the model
order is first estimated by the \gls{MDL} front-end of Section~\ref{sec:mdl}, and
the continuous parameters are then recovered by either pipeline of
Section~\ref{sec:estimation}.

The second problem hinges on the solution of the first. Writing
$\hat{\mathbf H}[k,i,j] \triangleq
[\hat{\mathbf h}_1[k,i,j],\dots,\hat{\mathbf h}_U[k,i,j]]$ for the
multiuser channel predicted at configuration $(k,i,j)$ from
$\{\hat{\boldsymbol\Theta}_u\}$, the design problem becomes:
\begin{equation}
\begin{aligned}
  \textrm{(P2)}\quad
  \max_{(k,i,j),\,\mathbf W} \;&
    \sum_{u=1}^{U}\log_2\!\big(1+\mathrm{SINR}_u(\mathbf W,
      \hat{\mathbf H}[k,i,j])\big) \\
  \textrm{s.t.}\quad
  & \lVert \mathbf w_u \rVert = 1,\quad u = 1,\dots,U, \\
  & (k,i,j) \in \mathcal C,
\end{aligned}
\label{eq:B_P2}
\end{equation}
where $\mathcal C\triangleq\{1,\dots,K\}\times\{1,\dots,I\}\times\{1,\dots,J\}$
is the set of $KIJ$ reachable panel configurations. Problem~(P2) is a mixed
discrete-continuous program, since the configuration index is combinatorial,
while the \gls{TPC} is continuous. We do not attempt a
joint solution. Instead, Section~\ref{sec:beamforming} fixes the \gls{TPC} to
one of three standard linear rules, which reduces (P2) to an exhaustive
search over the $KIJ$ grid points, whose cost
Section~\ref{sec:results_precoder} measures.

Two features distinguish (P1) and (P2) from their counterparts in the
fixed-aperture literature. The feasible set of (P2) is created by the same
mechanical motion that generates the training tensor of (P1), so the diversity
that the \gls{TPC} exploits and the diversity that makes the channel
identifiable are one and the same. And because the solution of (P1) is
parametric, the objective of (P2) can be evaluated at configurations outside
the training grid at no pilot cost, which is what separates this formulation
from one built on a table of measured channel responses.

\section{Channel Estimation Paradigms}

\label{sec:estimation}

This section develops the channel estimation framework, whose overall flow is
stated compactly as Algorithm~\ref{alg:framework}. We first describe the multiuser separation step
via orthogonal pilots and the \gls{MDL} front-end that both estimators share.
We then present the two pipelines, the first one performing a \gls{CP}
decomposition via \gls{ALS} followed by paired \gls{ESPRIT} for parametric
wavevector extraction, while the second is a tensor spatial-smoothing
\gls{MUSIC} estimator that preserves the multilinear structure of the
\gls{GFIM} array manifold. Channel reconstruction at arbitrary panel
configurations and the uniqueness conditions that make it well-posed will
conclude the section.

\begin{algorithm}[t]
\caption{Twin-time-scale \gls{GFIM} estimation framework}
\label{alg:framework}
\begin{algorithmic}[1]
\STATE \textbf{Input:} received tensor $\mathcal Y\in\bbC^{N\times P\times K\times I\times J}$ of Eq.~\eqref{eq:B_Y_tensor}, pilot matrix $\mathbf S$ with $\mathbf S^H\mathbf S=P\mathbf I_U$.
\FOR{$u = 1, \dots, U$}
  \STATE Despread: $\widetilde{\mathcal H}_u \leftarrow \frac{1}{P}\mathcal Y\times_2\mathbf s_u^*$ (Eq.~\eqref{eq:C_ortho_pilots}).
  \STATE $\hat L_u \leftarrow \textsc{MDL}(\widetilde{\mathcal H}_u)$ (Algorithm~\ref{alg:mdl}).
  \STATE Estimate $\{\hat{\mathbf q}_{u,\ell},\hat\beta_{u,\ell}\}_{\ell=1}^{\hat L_u}$ by Algorithm~\ref{alg:als_esprit} or Algorithm~\ref{alg:tensor_music}.
\ENDFOR
\STATE \textbf{Output:} per-user parameters, from which the channel at any
  configuration follows by Eq.~\eqref{eq:C_recon}.
\end{algorithmic}
\end{algorithm}

\subsection{Multiuser Separation via Orthogonal Pilots}
\label{sec:ortho_pilots}

The received-signal tensor in Eq.~\eqref{eq:B_Y_tensor} mixes all $U$
users through the pilot matrix
$\mathbf S\in\bbC^{P\times U}$.
When the pilots satisfy
$\mathbf S^H\mathbf S=P\,\mathbf I_U$,
a matched-filter projection along the pilot mode separates the users.
For user~$u$, forming
\begin{equation}
  \widetilde{\mathcal H}_u
    = \frac{1}{P}\,\mathcal Y \times_2 \mathbf s_u^*
  \label{eq:C_ortho_pilots}
\end{equation}
yields a 4th-order tensor
$\widetilde{\mathcal H}_u\in\bbC^{N\times K\times I\times J}$
that is an unbiased estimate of $\mathcal H_u$ corrupted by filtered
noise.
This despreading step requires $P\geq U$ and reduces the 5-way problem
to $U$ independent 4-way problems, one per user.

\subsection{Automatic Rank Detection via \glsentryshort{MDL}}
\label{sec:mdl}

Both estimation pipelines require the tensor rank $L_u$, which equals the number of propagation paths for user~$u$. In practice, $L_u$ is unknown and must be inferred from the observed data. We adopt the \gls{MDL} criterion~\cite{rissanen1978mdl,wax1985mdl}, applied to the eigenvalue spectrum of the sample covariance of the largest-mode unfolding, the standard information-theoretic approach to model-order selection in array processing, which extends naturally to the multidimensional setting~\cite{dacosta2011modelorder}.

Let $\mathbf X = \widetilde{\mathbf H}_{u,(m)} \in \bbC^{M \times Q}$ denote the mode-$m$ unfolding of $\widetilde{\mathcal H}_u$, where $m = \arg\max_n \dim_n(\widetilde{\mathcal H}_u)$ is chosen as the mode with the largest dimension to maximize the number of available samples. Define $N_s = \max(M,Q)$ and $p = \min(M,Q)$, and let $\lambda_1 \geq \lambda_2 \geq \cdots \geq \lambda_p \geq 0$ denote the eigenvalues of the sample covariance matrix
\begin{equation}
  \hat{\mathbf C} =
  \begin{cases}
    \frac{1}{M}\mathbf X^H \mathbf X, & \text{if } M \geq Q,\\[4pt]
    \frac{1}{Q}\mathbf X \mathbf X^H, & \text{otherwise.}
  \end{cases}
  \label{eq:C_sample_cov}
\end{equation}

For candidate rank $r \in \{0, 1, \dots, R_{\max}\}$, where $R_{\max} < p$, we define both the geometric and arithmetic means of the $p - r$ smallest eigenvalues,
\begin{equation}
  \bar{\lambda}_r^{\mathrm{geo}}
    = \exp\!\left(\frac{1}{p-r}\sum_{i=r+1}^{p}\ln\lambda_i\right),
  \quad
  \bar{\lambda}_r^{\mathrm{ari}}
    = \frac{1}{p-r}\sum_{i=r+1}^{p}\lambda_i.
  \label{eq:C_mdl_means}
\end{equation}
The \gls{MDL} cost function balances the log-likelihood of the noise eigenvalues against a model-complexity penalty, formulated as:
\begin{equation}
  \mathrm{MDL}(r) =
    -N_s(p-r)\ln\frac{\bar{\lambda}_r^{\mathrm{geo}}}
                      {\bar{\lambda}_r^{\mathrm{ari}}}
    + \frac{1}{2}\,r(2p-r)\ln N_s.
  \label{eq:C_MDL}
\end{equation}
The first term quantifies the departure of the noise eigenvalues from a flat profile, the hallmark of white noise, while the second term penalizes model order. The rank estimate is represented by
\begin{equation}
  \hat{L}_u = \arg\min_{0 \leq r \leq R_{\max}}\;\mathrm{MDL}(r).
  \label{eq:C_MDL_select}
\end{equation}

The rationale for this criterion is the following. Under the hypothesis that the tensor has a rank of exactly $r$, the $p - r$ smallest eigenvalues of $\hat{\mathbf C}$ correspond to a noise-only subspace and should be approximately equal, so that their geometric-to-arithmetic mean ratio approaches unity. A rank over-estimate, $r > L_u$, reduces the log-likelihood term but incurs an excessive penalty, whereas a rank underestimate, $r < L_u$, retains signal eigenvalues in the noise set and inflates the log-likelihood. The \gls{MDL} criterion provides a consistent estimator of $L_u$ as $N_s \to \infty$ under mild regularity conditions.

In the context of the \gls{GFIM} tensor, the unfolding along the largest mode
effectively stacks all spatial snapshots from the morphing and translation
grid, which produces a rich sample set even when the individual mode
dimensions are modest. This is a structural advantage of the 4-way model,
because the effective sample size is set by the product $K I J$, or by $N$
when $N$ is smaller, and it is this sample size that governs the reliability
of the \gls{MDL} decision. Section~\ref{sec:results_ce} quantifies that
reliability directly, and Fig.~\ref{fig:rank_vs_snr} confirms the
one-sidedness that Eq.~\eqref{eq:C_MDL} predicts, in that the penalty term
suppresses over-estimation, while an under-estimate remains possible, whenever the
weakest path sinks into the noise floor. Algorithm~\ref{alg:mdl} states the
procedure.

\begin{algorithm}[t]
\caption{\gls{MDL} tensor-rank detection}
\label{alg:mdl}
\begin{algorithmic}[1]
\STATE \textbf{Input:} $\widetilde{\mathcal H}_u \in \bbC^{N \times K \times I \times J}$, maximum order $R_{\max}$.
\STATE $m \leftarrow \arg\max_n \dim_n(\widetilde{\mathcal H}_u)$; \, $\mathbf X \leftarrow \widetilde{\mathbf H}_{u,(m)} \in \bbC^{M\times Q}$.
\STATE $N_s \leftarrow \max(M,Q)$, \, $p \leftarrow \min(M,Q)$.
\STATE Form $\hat{\mathbf C}$ by Eq.~\eqref{eq:C_sample_cov} and take its eigenvalues $\lambda_1 \geq \cdots \geq \lambda_p$.
\FOR{$r = 0, 1, \dots, R_{\max}$}
  \STATE Evaluate $\bar\lambda_r^{\mathrm{geo}}$, $\bar\lambda_r^{\mathrm{ari}}$ (Eq.~\eqref{eq:C_mdl_means}) and $\mathrm{MDL}(r)$ (Eq.~\eqref{eq:C_MDL}).
\ENDFOR
\STATE \textbf{Output:} $\hat L_u = \arg\min_r \mathrm{MDL}(r)$ (Eq.~\eqref{eq:C_MDL_select}).
\end{algorithmic}
\end{algorithm}

\subsection{Per-User \glsentryshort{CP} Factorization via \glsentryshort{ALS} and Paired \glsentryshort{ESPRIT}}
\label{sec:als_esprit}

\begin{figure*}[t]
\centering
\begin{tikzpicture}[
  node distance=10mm and 14mm,
  >=Stealth,
  block/.style={draw, rectangle, align=center, minimum height=10mm,
                minimum width=28mm, font=\small},
  annot/.style={font=\scriptsize, text=black!70, align=center},
  line/.style={->, thick},
]

% --- Row 1: CP decomposition ---
\node[block] (input)
  {4-D tensor\\$\widetilde{\mathcal H}_u\!\in\!\mathbb C^{N\times K\times I\times J}$};

\node[block, right=of input] (gevd)
  {GEVD\\initialization};

\node[block, right=of gevd] (als)
  {ALS iterations\\Eqs.~\eqref{eq:C_ALS_update}--\eqref{eq:C_ALS_closedform}};

\node[block, right=of als] (factors)
  {CP factors\\$\hat{\widetilde{\mathbf A}}_u,\hat{\mathbf B}_u,\hat{\mathbf C}_{x,u},\hat{\mathbf C}_{z,u}$};

% --- Row 2: ESPRIT extraction ---
\node[block, below=14mm of factors] (esprit)
  {Paired ESPRIT\\per factor};

\node[block, left=of esprit] (wavevec)
  {Wavevector\\assembly\\$\hat{\mathbf q}_{u,\ell}$};

\node[block, left=of wavevec] (aoa)
  {Angle\\recovery\\$(\hat\theta,\hat\phi)$};

\node[block, left=of aoa] (gains)
  {LS gain\\extraction\\$\hat{\boldsymbol\beta}_u$};

% --- Arrows ---
\draw[line] (input) -- (gevd);
\draw[line] (gevd) -- (als);
\draw[line] (als) -- (factors);
\draw[line] (factors) -- (esprit);
\draw[line] (esprit) -- (wavevec);
\draw[line] (wavevec) -- (aoa);
\draw[line] (aoa) -- (gains);

% --- Annotations ---
\node[annot, above=2mm of gevd]
  {Generalized EVD\\rank-$L_u$ init};
\node[annot, above=2mm of als]
  {Cyclic mode-wise\\LS updates};
\node[annot, below=2mm of esprit]
  {$\hat{\mathbf F}_{1:}^H\hat{\mathbf F}_{2:}\to z_\ell\to\hat\mu_\ell$\\
   $\hat{\mathbf B}_u\!\to\!\hat q_y$,\;
   $\hat{\mathbf C}_{x,u}\!\to\!\hat q_x$,\;
   $\hat{\mathbf C}_{z,u}\!\to\!\hat q_z$};
\node[annot, below=2mm of wavevec]
  {Eq.~\eqref{eq:C_wavevector_aoa}};
\node[annot, below=2mm of gains]
  {$\hat{\boldsymbol\beta}=\boldsymbol\Phi^\dagger\mathrm{vec}(\widetilde{\mathcal H}_u)$};

\end{tikzpicture}
\caption{Block diagram of the \gls{ALS}\,+\,\gls{ESPRIT} pipeline of Algorithm~\ref{alg:als_esprit}.}
\label{fig:als_esprit_diagram}
\end{figure*}
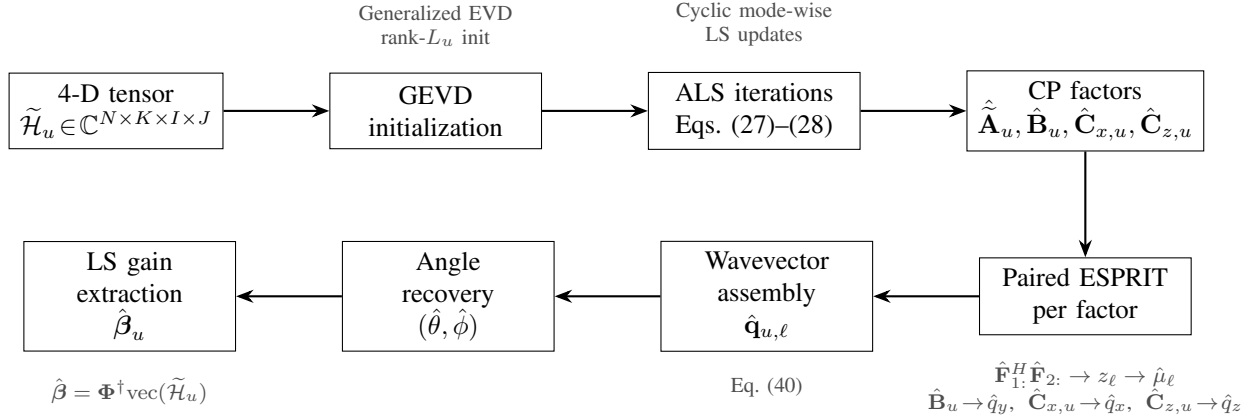

After pilot separation, the per-user tensor
$\widetilde{\mathcal H}_u$ inherits the rank-$L_u$ \gls{CP} structure
of Eq.~\eqref{eq:B_CP_user}.
The factor matrices
$\widetilde{\mathbf A}_u$, $\mathbf B_u$,
$\mathbf C_u$, $\mathbf D_u$
are estimated by the \gls{ALS} algorithm~\cite{kolda2009tensor,sidiropoulos2017tensor}, which cyclically solves
least-squares subproblems for each mode while fixing the remaining
three.
Specifically, the mode-$n$ update at iteration~$t$ solves
\begin{equation}
  \mathbf F^{(n)} \leftarrow
    \arg\min_{\mathbf F}\,
    \left\|
      \widetilde{\mathbf H}_{u,(n)}
      -\mathbf F\,\mathbf Z_{(n)}^T
    \right\|_F^2,
  \label{eq:C_ALS_update}
\end{equation}
where $\mathbf Z_{(n)}$ is the Khatri--Rao product of the remaining
factor matrices in reverse mode order, as given by
Eqs.~\eqref{eq:B_unfold_mode1}--\eqref{eq:B_unfold_mode4}.
The closed-form solution is
\begin{equation}
  \mathbf F^{(n)}
    = \widetilde{\mathbf H}_{u,(n)}\,
      \mathbf Z_{(n)}
      \left(\mathbf Z_{(n)}^T\mathbf Z_{(n)}\right)^{-1}.
  \label{eq:C_ALS_closedform}
\end{equation}
The iterations alternate over modes $n=1,2,3,4$ and repeat until
convergence, typically monitored by the relative change in the
reconstruction error
$\|\widetilde{\mathcal H}_u
  -\llbracket \widetilde{\mathbf A}_u,\mathbf B_u,
             \mathbf C_u,\mathbf D_u\rrbracket\|_F
 /\|\widetilde{\mathcal H}_u\|_F$.

Each subproblem in Eq.~\eqref{eq:C_ALS_update} is an unconstrained linear
least-squares problem whose exact minimizer is Eq.~\eqref{eq:C_ALS_closedform},
so every mode update is non-increasing in the cost and the sequence of
reconstruction errors is monotone and bounded below. This guarantees
convergence of the cost, though not of the iterates, and the convergence
properties of \gls{ALS} for the \gls{CP} model are well documented. Comon
\textit{et al.}~\cite{comon2009als} survey the behavior of the iteration and
the swamps that arise when two factor columns become collinear, and
Uschmajew~\cite{uschmajew2012als} establishes local linear convergence to a
limit point under the condition that the Hessian of the cost restricted to the
non-degenerate directions be positive definite at that point, which for our
model is implied by the Kruskal condition of
Corollary~\ref{cor:dims}. We therefore do not reproduce a convergence analysis
here and instead take the standard guarantees as given, noting only that they
are local and that the limit point may be a spurious stationary point.

The initialization is what governs which limit point is reached in practice.
We initialize the \gls{ALS} procedure by a \gls{GEVD}, which provides a
closed-form and non-iterative starting point that is more robust to local
minima than a random initialization, because it exploits the simultaneous
diagonalization structure of the mode-$n$ slice matrices to produce an initial
set of factors already close to the global optimum. Together with the $T_r$
restarts of Algorithm~\ref{alg:als_esprit}, of which the one attaining the
smallest reconstruction error is retained, this makes the spurious limit points
rare rather than absent. Section~\ref{sec:results_ce} will show that their
residual frequency is what sets the high-\gls{SNR} floor of
Fig.~\ref{fig:nmse_vs_snr}, rather than any limitation of the \gls{CP} model
itself.

Once the \gls{CP} factors have been estimated, the wavevector components are recovered from the Vandermonde structure of each factor matrix via paired \gls{ESPRIT}. Consider a generic factor $\hat{\mathbf F} \in \bbC^{M \times L_u}$ whose columns satisfy $[\hat{\mathbf F}]_{m,\ell} \approx \exp(-\mathrm{j}\,\mu_\ell\, m\, \delta)$ for a uniform step $\delta$ and unknown spatial frequency $\mu_\ell$. Define the selection matrices $\hat{\mathbf F}_1 = \hat{\mathbf F}(1{:}M{-}1,:)$ and $\hat{\mathbf F}_2 = \hat{\mathbf F}(2{:}M,:)$, which extract all but the last and first rows, respectively. The shift-invariance property of the Vandermonde structure implies $\hat{\mathbf F}_2 = \hat{\mathbf F}_1 \, \mathrm{Diag}(\mathbf z)$ with $z_\ell = \exp(-\mathrm{j}\,\mu_\ell\,\delta)$. The spatial frequency is then recovered column-wise as
\begin{equation}
  \hat{z}_\ell
    = \frac{\hat{\mathbf f}_{1,\ell}^H \hat{\mathbf f}_{2,\ell}}
           {\hat{\mathbf f}_{1,\ell}^H \hat{\mathbf f}_{1,\ell}},
  \qquad
  \hat{\mu}_\ell = -\frac{\angle \hat{z}_\ell}{\delta},
  \label{eq:C_esprit_freq}
\end{equation}
where $\hat{\mathbf f}_{1,\ell}$ and $\hat{\mathbf f}_{2,\ell}$ denote the $\ell$-th columns of $\hat{\mathbf F}_1$ and $\hat{\mathbf F}_2$.

The three non-spatial factor matrices directly yield the wavevector
components. Applying Eq.~\eqref{eq:C_esprit_freq} to $\hat{\mathbf B}_u$ with the
morphing spacing $\delta_y$ returns $\hat q^{(y)}_{u,\ell}$, applying it to
$\hat{\mathbf C}_u$ with the translation step $\delta_x$ returns
$\hat q^{(x)}_{u,\ell}$, and applying it to $\hat{\mathbf D}_u$ with the
translation step $\delta_z$ returns $\hat q^{(z)}_{u,\ell}$. The estimates are
paired across the three modes by the column order of the \gls{CP} factors,
which the decomposition fixes up to one common permutation.
The \gls{AoA} pair $(\hat\theta_{u,\ell}, \hat\phi_{u,\ell})$ is then recovered from the assembled wavevector $\hat{\mathbf q}_{u,\ell} = [\hat{q}^{(x)}_{u,\ell},\, \hat{q}^{(y)}_{u,\ell},\, \hat{q}^{(z)}_{u,\ell}]^T$ via
\begin{equation}
  \hat\theta_{u,\ell} = \arccos\!\left(\frac{\hat{q}^{(z)}_{u,\ell}}{k_0}\right),
  \qquad
  \hat\phi_{u,\ell} = \arctan\!\left(\frac{\hat{q}^{(y)}_{u,\ell}}{\hat{q}^{(x)}_{u,\ell}}\right),
  \label{eq:C_aoa_recovery}
\end{equation}
where $k_0 = 2\pi/\lambda$.

When only a single morphing state is available during training, the factor
$\hat{\mathbf B}_u$ collapses to a scalar per path and
$\hat{q}^{(y)}_{u,\ell}$ is not directly observable. In this case we invoke the
far-field norm constraint $\|\mathbf q_{u,\ell}\| = k_0$, which together with
the visible-hemisphere convention $q^{(y)}_{u,\ell}\leq 0$ of
Section~\ref{sec:problem} determines the sign and gives
\begin{equation}
  \hat{q}^{(y)}_{u,\ell}
    = -\sqrt{\max\!\big(k_0^2 - (\hat{q}^{(x)}_{u,\ell})^2 - (\hat{q}^{(z)}_{u,\ell})^2,\;0\big)}.
  \label{eq:C_qy_norm}
\end{equation}
This closure is available only because the panel is planar and the sources lie
in one half-space. It is the same relationship that
Remark~\ref{rem:degenerate} invokes to rule out degenerate geometries, and
Section~\ref{sec:results_precoder} shows that it is also what caps the
beamforming value of the morphing mode, since a component that the in-plane
ones already determine cannot separate two users that translation fails to separate.

After determining the wavevectors, we extract the complex path gains $\boldsymbol\beta_u = [\beta_{u,1}, \dots, \beta_{u,L_u}]^T$ via a \gls{LS} fit to the observed tensor. Let us define the Khatri--Rao steering matrix $\boldsymbol\Phi \in \bbC^{NKIJ \times L_u}$ whose $\ell$-th column is the vectorized rank-one steering tensor, as follows:
\begin{equation}
  [\boldsymbol\Phi]_{:,\ell}
    = \vecop\!\big(\hat{\mathbf a}_{u,\ell} \circ \hat{\mathbf b}_{u,\ell}
                    \circ \hat{\mathbf c}_{x,u,\ell} \circ \hat{\mathbf c}_{z,u,\ell}\big),
  \label{eq:C_steering_kr}
\end{equation}
where each vector is evaluated at the estimated wavevector. The \gls{LS} estimate is
\begin{equation}
  \hat{\boldsymbol\beta}_u
    = \boldsymbol\Phi^\dagger\, \vecop(\widetilde{\mathcal H}_u),
  \label{eq:C_ls_gains}
\end{equation}
with $\boldsymbol\Phi^\dagger = (\boldsymbol\Phi^H \boldsymbol\Phi)^{-1} \boldsymbol\Phi^H$ denoting the Moore--Penrose pseudo-inverse.

Because Eq.~\eqref{eq:C_esprit_freq} recovers $\mu_\ell$ from the phase of
$\hat z_\ell$, the mode step must satisfy $|\mu_\ell|\delta < \pi$ for the
estimate to be unambiguous, which for the translation modes is the familiar
half-wavelength limit $\delta_x,\delta_z < \lambda/2$ and for the morphing mode
reads $\delta_y < \lambda/2$. Section~\ref{sec:results_ce} will return to this
bound, since Fig.~\ref{fig:nmse_vs_step} shows that the fit degrades only
mildly at the limit, while the extrapolation of Eq.~\eqref{eq:C_recon}
degrades sharply, which is exactly what an aliased phase should do.

The complete \gls{ALS}\,+\,\gls{ESPRIT} pipeline is summarized in
Algorithm~\ref{alg:als_esprit} and illustrated in
Fig.~\ref{fig:als_esprit_diagram}.

\begin{algorithm}[t]
\caption{\glsentryshort{ALS}\,+\,\glsentryshort{ESPRIT} channel estimation}
\label{alg:als_esprit}
\begin{algorithmic}[1]
\STATE \textbf{Input:} $\widetilde{\mathcal H}_u \in \bbC^{N \times K \times I \times J}$; model order $L_u$, supplied or taken from Algorithm~\ref{alg:mdl}; number of restarts $T_r$; tolerance $\epsilon$.
\FOR{$t_r = 1, \dots, T_r$}
  \STATE Initialize $\widetilde{\mathbf A}_u$, $\mathbf B_u$, $\mathbf C_u$, $\mathbf D_u$ by \gls{GEVD}.
  \REPEAT
    \FOR{$n = 1, 2, 3, 4$}
      \STATE $\mathbf Z_{(n)} \leftarrow$ Khatri--Rao product of the other three factors, in the order of Eqs.~\eqref{eq:B_unfold_mode1}--\eqref{eq:B_unfold_mode4}.
      \STATE $\mathbf F^{(n)} \leftarrow \widetilde{\mathbf H}_{u,(n)}\,\mathbf Z_{(n)}(\mathbf Z_{(n)}^T \mathbf Z_{(n)})^{-1}$ \quad (Eq.~\eqref{eq:C_ALS_closedform})
    \ENDFOR
  \UNTIL{the relative reconstruction error changes by less than $\epsilon$.}
\ENDFOR
\STATE Keep the restart of smallest reconstruction error.
\STATE Apply paired \gls{ESPRIT} (Eq.~\eqref{eq:C_esprit_freq}) to $\hat{\mathbf B}_u$, $\hat{\mathbf C}_u$, $\hat{\mathbf D}_u$ with the steps $\delta_y$, $\delta_x$, $\delta_z$, which gives $\hat{q}^{(y)}_{u,\ell}$, $\hat{q}^{(x)}_{u,\ell}$, $\hat{q}^{(z)}_{u,\ell}$.
\IF{$K = 1$}
  \STATE Recover $\hat{q}^{(y)}_{u,\ell}$ from the norm constraint of Eq.~\eqref{eq:C_qy_norm}.
\ENDIF
\STATE Compute $(\hat\theta_{u,\ell},\hat\phi_{u,\ell})$ from Eq.~\eqref{eq:C_aoa_recovery}.
\STATE Extract $\hat{\boldsymbol\beta}_u$ from Eq.~\eqref{eq:C_ls_gains}.
\STATE \textbf{Output:} $\{\hat{\mathbf q}_{u,\ell},\, \hat\beta_{u,\ell},\, \hat\theta_{u,\ell},\, \hat\phi_{u,\ell}\}_{\ell=1}^{L_u}$.
\end{algorithmic}
\end{algorithm}

\subsection{Tensor Spatial-Smoothing \glsentryshort{MUSIC}}
\label{sec:tensor_music}

The second estimation pipeline replaces the parametric \gls{CP} decomposition
by a subspace approach that preserves the multilinear structure of the
\gls{GFIM} array manifold through multi-mode spatial smoothing, as
Fig.~\ref{fig:tensor_music_diagram} illustrates. The method extends classical
\gls{MUSIC} to the tensor domain by constructing a covariance matrix from
overlapping sub-tensors extracted along every mode simultaneously.

\begin{figure*}[t]
\centering
\begin{tikzpicture}[
  node distance=9mm and 7mm,
  >=Stealth,
  block/.style={draw, rectangle, align=center, minimum height=10mm,
                minimum width=25mm, font=\small},
  annot/.style={font=\scriptsize, text=black!70, align=center},
  line/.style={->, thick},
]

% --- Row 1: Main pipeline ---
\node[block] (input)
  {4-D tensor\\$\widetilde{\mathcal H}_u\!\in\!\mathbb C^{N\times K\times I\times J}$};

\node[block, right=of input] (reshape)
  {Reshape to 5-D\\$(N_x,N_z,K,I,J)$};

\node[block, right=of reshape] (smooth)
  {Multi-mode\\spatial smoothing};

\node[block, right=of smooth] (cov)
  {Covariance\\$\hat{\mathbf R}$\\+ FB averaging};

\node[block, right=of cov] (evd)
  {EVD\\$\hat{\mathbf R}=\mathbf U\boldsymbol\Lambda\mathbf U^H$};

% --- Row 2: MUSIC spectrum ---
\node[block, below=12mm of evd] (noise)
  {Noise subspace\\$\mathbf U_n$};

\node[block, left=of noise] (steer)
  {Kronecker\\steering grid\\$\mathbf a(\theta,\phi)$};

\node[block, left=of steer] (spec)
  {MUSIC spectrum\\$P(\theta,\phi)$};

\node[block, left=of spec] (peaks)
  {Peak detection\\$L$ largest peaks};

\node[block, left=of peaks] (output)
  {$\{\hat\theta_{u,\ell},\hat\phi_{u,\ell}\}_{\ell=1}^{L_u}$};

% --- Arrows ---
\draw[line] (input) -- (reshape);
\draw[line] (reshape) -- (smooth);
\draw[line] (smooth) -- (cov);
\draw[line] (cov) -- (evd);
\draw[line] (evd) -- (noise);
\draw[line] (noise) -- (steer);
\draw[line] (steer) -- (spec);
\draw[line] (spec) -- (peaks);
\draw[line] (peaks) -- (output);

% --- Annotations ---
\node[annot, above=2mm of reshape]
  {UPA structure\\$N\!=\!N_x N_z$};
\node[annot, above=2mm of smooth]
  {Subarrays\\$(L_{n_x},L_{n_z},L_k,L_i,L_j)$};
\node[annot, above=2mm of cov]
  {$\hat{\mathbf R}_{\mathrm{fb}} = \tfrac{1}{2}(\hat{\mathbf R}+\mathbf J\hat{\mathbf R}^*\mathbf J)$};
\node[annot, below=2mm of steer]
  {$\mathbf a = \mathbf a_{n_x}\!\otimes\!\mathbf a_{n_z}\!\otimes\!\mathbf b_k\!\otimes\!\mathbf c_i\!\otimes\!\mathbf d_j$};
\node[annot, below=2mm of spec]
  {$P(\theta,\phi)=\frac{1}{\|\mathbf U_n^H\mathbf a(\theta,\phi)\|^2}$};

\end{tikzpicture}
\caption{Block diagram of the tensor spatial-smoothing \gls{MUSIC} pipeline of Algorithm~\ref{alg:tensor_music}.}
\label{fig:tensor_music_diagram}
\end{figure*}
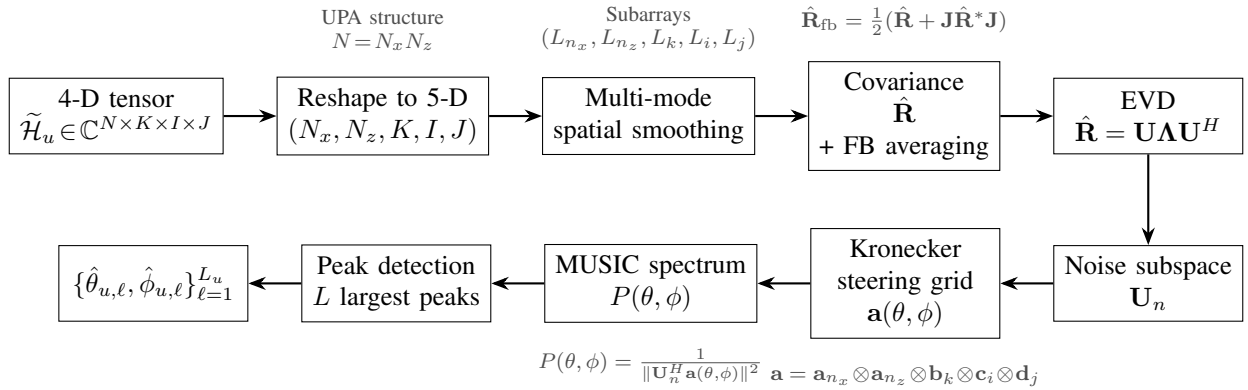

The \gls{GFIM} employs a \gls{UPA} having $N = N_x N_z$ elements, where $N_x$ and $N_z$ denote the number of elements along the $x$- and $z$-axes, respectively. To exploit this two-dimensional spatial regularity, the 4-D per-user tensor $\widetilde{\mathcal H}_u \in \bbC^{N \times K \times I \times J}$ is reshaped into a 5-D tensor,
\begin{equation}
  \mathcal H_u^{(5)} \in \bbC^{N_x \times N_z \times K \times I \times J},
  \label{eq:C_reshape_5d}
\end{equation}
where the element mode $N$ is factored into its \gls{UPA} components. This restructuring enables independent spatial smoothing along the $x$- and
$z$-array dimensions, and thereby increases the effective number of snapshots
available for covariance estimation. Fig.~\ref{fig:spatial_smoothing}
illustrates the three steps that follow.

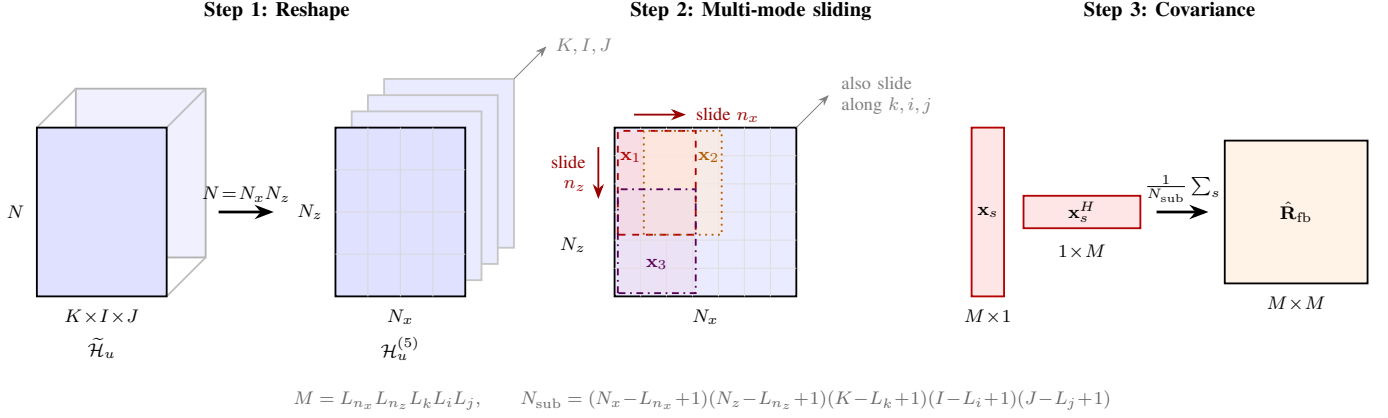
\begin{figure*}[t]
\centering
\resizebox{\textwidth}{!}{%
\begin{tikzpicture}[
  >=Stealth,
  every node/.style={font=\small},
  annot/.style={font=\scriptsize, text=black!60},
  arrowlabel/.style={font=\scriptsize, midway, above},
  stephead/.style={font=\footnotesize\bfseries},
]

% ===== Step headers (one row, above every drawing) =====
\node[stephead] at (3.7, 4.4) {Step 1: Reshape};
\node[stephead] at (11.0, 4.4) {Step 2: Multi-mode sliding};
\node[stephead] at (17.45, 4.4) {Step 3: Covariance};

% ===== Step 1: 4-D tensor =====
\begin{scope}
  \fill[blue!6] (0.6,0.6) -- (2.6,0.6) -- (2.6,3.2) -- (0.6,3.2) -- cycle;
  \draw[thick, gray!60] (0.6,0.6) -- (2.6,0.6) -- (2.6,3.2) -- (0.6,3.2) -- cycle;
  \foreach \p in {(0,0), (2.0,0), (2.0,2.6), (0,2.6)}
    \draw[thick, gray!50] \p -- ++(0.6,0.6);
  \fill[blue!12] (0,0) rectangle (2.0,2.6);
  \draw[thick] (0,0) rectangle (2.0,2.6);
  \node[left] at (-0.05,1.3) {\scriptsize $N$};
  \node[below] at (1.0,-0.05) {\scriptsize $K\!\times\!I\!\times\!J$};
  \node at (1.0,-0.8) {\scriptsize $\widetilde{\mathcal H}_u$};
\end{scope}

\draw[->, very thick] (2.8,1.3) -- node[arrowlabel] {$N\!=\!N_x N_z$} (3.65,1.3);

% ===== Step 1: 5-D tensor (stacked slabs) =====
\begin{scope}[xshift=4.6cm]
  \foreach \d in {0.75,0.5,0.25} {
    \fill[blue!8] (\d, \d) rectangle (\d+2.0, \d+2.6);
    \draw[thick, gray!45] (\d, \d) rectangle (\d+2.0, \d+2.6);
  }
  \fill[blue!12] (0,0) rectangle (2.0,2.6);
  \draw[thick] (0,0) rectangle (2.0,2.6);
  \foreach \x in {0.5, 1.0, 1.5}  \draw[gray!30, thin] (\x,0) -- (\x,2.6);
  \foreach \y in {0.65, 1.3, 1.95} \draw[gray!30, thin] (0,\y) -- (2.0,\y);
  \node[below] at (1.0,-0.05) {\scriptsize $N_x$};
  \node[left] at (-0.05,1.3) {\scriptsize $N_z$};
  \draw[->, gray] (2.75,3.35) -- (3.25,3.85) node[right, font=\scriptsize] {$K,I,J$};
  \node at (1.0,-0.8) {\scriptsize $\mathcal H_u^{(5)}$};
\end{scope}

% ===== Step 2: sliding windows =====
\begin{scope}[xshift=8.9cm]
  \fill[blue!8] (0,0) rectangle (2.8,2.6);
  \draw[thick] (0,0) rectangle (2.8,2.6);
  \foreach \x in {0.4, 0.8, 1.2, 1.6, 2.0, 2.4} \draw[gray!25, thin] (\x,0) -- (\x,2.6);
  \foreach \y in {0.433, 0.867, 1.3, 1.733, 2.167} \draw[gray!25, thin] (0,\y) -- (2.8,\y);

  % Window 1, window 2 (shifted along n_x), window 3 (shifted along n_z)
  \fill[red!15, opacity=0.7] (0.05,0.95) rectangle (1.25,2.55);
  \fill[orange!15, opacity=0.6] (0.45,0.95) rectangle (1.65,2.55);
  \fill[violet!15, opacity=0.5] (0.05,0.05) rectangle (1.25,1.65);
  \draw[thick, red!70!black, dashed] (0.05,0.95) rectangle (1.25,2.55);
  \draw[thick, orange!70!black, dotted] (0.45,0.95) rectangle (1.65,2.55);
  \draw[thick, violet!70!black, dashdotted] (0.05,0.05) rectangle (1.25,1.65);
  % Labels sit where only one window is present
  \node[red!70!black, font=\scriptsize] at (0.25,2.15) {$\mathbf x_1$};
  \node[orange!70!black, font=\scriptsize] at (1.45,2.15) {$\mathbf x_2$};
  \node[violet!70!black, font=\scriptsize] at (0.65,0.5) {$\mathbf x_3$};

  % Sliding directions, drawn outside the face
  \draw[->, thick, red!60!black] (0.3,2.8) -- (1.1,2.8)
    node[right, font=\scriptsize] {slide $n_x$};
  \draw[->, thick, red!60!black] (-0.25,2.3) -- (-0.25,1.5);
  \node[left, font=\scriptsize, red!60!black, align=right] at (-0.3,1.9) {slide\\ $n_z$};

  \node[below] at (1.4,-0.05) {\scriptsize $N_x$};
  \node[left] at (-0.3,0.8) {\scriptsize $N_z$};

  % Depth: the window also slides along the three configuration modes
  \draw[->, gray] (2.8,2.6) -- (3.3,3.1);
  \node[right, font=\scriptsize, gray, align=left] at (3.35,3.1) {also slide\\ along $k,i,j$};
\end{scope}

% ===== Step 3: outer products averaged into the covariance =====
\begin{scope}[xshift=14.4cm]
  \fill[red!10] (0,0) rectangle (0.5,2.6);
  \draw[thick, red!70!black] (0,0) rectangle (0.5,2.6);
  \node[font=\scriptsize] at (0.25,1.3) {$\mathbf x_s$};
  \node[below] at (0.25,-0.05) {\scriptsize $M\!\times\!1$};

  \fill[red!10] (0.8,1.05) rectangle (2.6,1.55);
  \draw[thick, red!70!black] (0.8,1.05) rectangle (2.6,1.55);
  \node[font=\scriptsize] at (1.7,1.3) {$\mathbf x_s^H$};
  \node[below] at (1.7,0.95) {\scriptsize $1\!\times\!M$};

  \draw[->, very thick] (2.85,1.3) -- node[arrowlabel] {$\frac{1}{N_{\mathrm{sub}}}\sum_s$} (3.65,1.3);

  \fill[orange!10] (3.9,0.2) rectangle (6.1,2.4);
  \draw[thick] (3.9,0.2) rectangle (6.1,2.4);
  \node[font=\scriptsize] at (5.0,1.3) {$\hat{\mathbf R}_{\mathrm{fb}}$};
  \node[below] at (5.0,0.15) {\scriptsize $M\!\times\!M$};
\end{scope}

% ===== Bottom annotation =====
\node[annot] at (10.25,-1.6)
  {$M = L_{n_x} L_{n_z} L_k L_i L_j, \qquad
   N_{\mathrm{sub}} = (N_x\!-\!L_{n_x}\!+\!1)(N_z\!-\!L_{n_z}\!+\!1)(K\!-\!L_k\!+\!1)(I\!-\!L_i\!+\!1)(J\!-\!L_j\!+\!1)$};

\end{tikzpicture}%
}
\caption{Multi-mode spatial smoothing of the 5-D tensor of Eq.~\eqref{eq:C_reshape_5d} into the sample covariance of Eq.~\eqref{eq:C_cov_smooth}.}
\label{fig:spatial_smoothing}
\end{figure*}

Let us now define the subarray sizes $(L_{n_x}, L_{n_z}, L_k, L_i, L_j)$, one per mode, each chosen as approximately half the corresponding mode dimension. The 5-D tensor is partitioned into overlapping sub-tensors by sliding a window of size $(L_{n_x}, L_{n_z}, L_k, L_i, L_j)$ independently along each mode, which is the multi-mode counterpart of classical spatial smoothing~\cite{shan1985smoothing}. Let $M = L_{n_x} L_{n_z} L_k L_i L_j$ denote the vectorized sub-tensor length. For each valid set of offsets $(s_{n_x}, s_{n_z}, s_k, s_i, s_j)$, the sub-tensor
\begin{equation}
  \mathbf x_s = \vecop\!\Big(
    \mathcal H_u^{(5)}(s_{n_x}{:}s_{n_x}{+}L_{n_x}{-}1,\ldots,s_j{:}s_j{+}L_j{-}1)
  \Big) \in \bbC^{M}
  \label{eq:C_subtensor}
\end{equation}
is extracted and contributes to the sample covariance matrix,
\begin{equation}
  \hat{\mathbf R}
    = \frac{1}{N_{\mathrm{sub}}}
      \sum_{s \in \mathcal{S}} \mathbf x_s \mathbf x_s^H
    \in \bbC^{M \times M},
  \label{eq:C_cov_smooth}
\end{equation}
where $\mathcal{S}$ is the set of all valid offset tuples and $N_{\mathrm{sub}} = |\mathcal{S}| = \prod_{d} (D_d - L_d + 1)$, with $D_d$ denoting the $d$-th mode dimension.

The number of independent sub-tensors grows multiplicatively with the number
of modes, which is a central advantage of the tensor formulation. In the
5-mode \gls{GFIM} setting even modest per-mode dimensions yield a large
effective sample size, so that the conditioning of $\hat{\mathbf R}$ improves
substantially relative to conventional vector-based spatial smoothing. The
price is aperture, because each mode is smoothed by shortening it, and
Section~\ref{sec:results_ce} shows this trade in operation, in that the
resolution of tensor \gls{MUSIC} falls behind \gls{ESPRIT} once the path
count approaches the identifiability ceiling.

To further improve the rank of $\hat{\mathbf R}$ and reduce the effective noise floor, forward--backward averaging~\cite{pillai1989fb} is applied,
\begin{equation}
  \hat{\mathbf R}_{\mathrm{fb}}
    = \frac{1}{2}\big(\hat{\mathbf R} + \mathbf J\, \hat{\mathbf R}^* \mathbf J\big),
  \label{eq:C_fb_avg}
\end{equation}
where $\mathbf J \in \bbR^{M \times M}$ is the exchange (anti-diagonal identity) matrix. This operation exploits the conjugate symmetry of the steering manifold under uniform spacing and effectively doubles the number of snapshots.

The eigen-decomposition of $\hat{\mathbf R}_{\mathrm{fb}} = \mathbf U \boldsymbol\Lambda \mathbf U^H$ partitions the eigenvectors into a signal subspace $\mathbf U_s \in \bbC^{M \times L_u}$ (the $L_u$ eigenvectors associated with the largest eigenvalues) and a noise subspace $\mathbf U_n \in \bbC^{M \times (M - L_u)}$. Since the signal steering vectors are orthogonal to the noise subspace, the \gls{MUSIC} pseudo-spectrum is defined as
\begin{equation}
  P_{\mathrm{MUSIC}}(\theta, \phi)
    = \frac{1}{\|\mathbf U_n^H \mathbf a(\theta, \phi)\|^2},
  \label{eq:C_music_spectrum}
\end{equation}
where $\mathbf a(\theta, \phi) \in \bbC^{M}$ is the Kronecker steering vector for the sub-tensor geometry. For the 5-D structure, this vector preserves the separable \gls{GFIM} manifold,
\begin{equation}
  \mathbf a(\theta, \phi)
    = \mathbf a_{n_x}(q_x) \kron \mathbf a_{n_z}(q_z) \kron
      \mathbf b_k(q_y) \kron \mathbf c_i(q_x) \kron \mathbf d_j(q_z),
  \label{eq:C_kron_steer_music}
\end{equation}
where $q_x = k_0 \sin\theta\cos\phi$, $q_y = k_0 \sin\theta\sin\phi$, $q_z = k_0 \cos\theta$, and each sub-vector is a Vandermonde segment having a length equal to the corresponding subarray size. Specifically, $[\mathbf a_{n_x}(q_x)]_m = \exp(-\mathrm{j}\, q_x\, m\, d_{\mathrm{upa}})$ for $m = 0, \dots, L_{n_x}-1$, where $d_{\mathrm{upa}}$ is the \gls{UPA} element spacing, and the remaining sub-vectors are built in the same way with the physical step of their mode, namely $d_{\mathrm{upa}}$ again for $\mathbf a_{n_z}$, $\delta_y$ for $\mathbf b_k$, $\delta_x$ for $\mathbf c_i$ and $\delta_z$ for $\mathbf d_j$.

The pseudo-spectrum is evaluated on a two-dimensional grid of $G \times G$
candidate angles $(\theta, \phi)$, and the $L_u$ highest peaks, identified by a
local-maximum filter with a guard band, yield the \gls{AoA} estimates
$\{(\hat\theta_{u,\ell}, \hat\phi_{u,\ell})\}_{\ell=1}^{L_u}$. The wavevectors
$\hat{\mathbf q}_{u,\ell}$ then follow from Eq.~\eqref{eq:C_wavevector_aoa} and
the gains from the same \gls{LS} fit of Eq.~\eqref{eq:C_ls_gains} used by the
first pipeline, so that both estimators deliver the identical parameter set to
the reconstruction of Eq.~\eqref{eq:C_recon}. The angular grid caps the
attainable resolution at roughly $180^\circ/G$ per axis, which is the origin
of the grid-limited behavior reported in Section~\ref{sec:results_ce}.

The complete pipeline is summarized in Algorithm~\ref{alg:tensor_music}.

\begin{algorithm}[t]
\caption{Tensor spatial-smoothing \glsentryshort{MUSIC}}
\label{alg:tensor_music}
\begin{algorithmic}[1]
\STATE \textbf{Input:} $\widetilde{\mathcal H}_u \in \bbC^{N \times K \times I \times J}$; model order $L_u$, supplied or taken from Algorithm~\ref{alg:mdl}; grid density $G$.
\STATE Reshape to 5-D: $\mathcal H_u^{(5)} \in \bbC^{N_x \times N_z \times K \times I \times J}$ (Eq.~\eqref{eq:C_reshape_5d}).
\STATE Choose subarray sizes $(L_{n_x}, L_{n_z}, L_k, L_i, L_j)$, each about half its mode dimension, and set $M = L_{n_x}L_{n_z}L_kL_iL_j$.
\FOR{each valid offset tuple $s \in \mathcal S$}
  \STATE Extract and vectorize the sub-tensor $\mathbf x_s$ (Eq.~\eqref{eq:C_subtensor}).
\ENDFOR
\STATE $\hat{\mathbf R} \leftarrow \frac{1}{N_{\mathrm{sub}}}\sum_{s\in\mathcal S}\mathbf x_s\mathbf x_s^H$ (Eq.~\eqref{eq:C_cov_smooth}).
\STATE Apply forward--backward averaging, $\hat{\mathbf R}_{\mathrm{fb}} = \frac{1}{2}(\hat{\mathbf R} + \mathbf J \hat{\mathbf R}^* \mathbf J)$ (Eq.~\eqref{eq:C_fb_avg}).
\STATE Eigen-decompose $\hat{\mathbf R}_{\mathrm{fb}}$ and collect the $M - L_u$ smallest eigenvectors into $\mathbf U_n$.
\STATE Evaluate $P_{\mathrm{MUSIC}}(\theta, \phi)$ on the $G \times G$ grid, with the Kronecker steering vector of Eq.~\eqref{eq:C_kron_steer_music}.
\STATE Detect the $L_u$ largest peaks, which give $\{(\hat\theta_{u,\ell}, \hat\phi_{u,\ell})\}_{\ell=1}^{L_u}$.
\STATE Map the angles to $\hat{\mathbf q}_{u,\ell}$ by Eq.~\eqref{eq:C_wavevector_aoa} and extract $\hat{\boldsymbol\beta}_u$ by \gls{LS} (Eq.~\eqref{eq:C_ls_gains}).
\STATE \textbf{Output:} $\{\hat{\mathbf q}_{u,\ell},\, \hat\beta_{u,\ell},\, \hat\theta_{u,\ell},\, \hat\phi_{u,\ell}\}_{\ell=1}^{L_u}$.
\end{algorithmic}
\end{algorithm}

\subsection{Comparison of the Two Pipelines}

The \gls{ALS}\,+\,\gls{ESPRIT} and tensor \gls{MUSIC} pipelines offer
complementary trade-offs, which Table~\ref{tab:pipeline_comparison}
summarizes. The \gls{ALS} approach exploits the exact \gls{CP} structure of
the channel tensor and recovers the wavevector components independently per
mode via \gls{ESPRIT}, so that it delivers high-resolution estimates whenever
the \gls{CP} model is well conditioned. Its convergence is, however,
sensitive to the initialization, and an occasional spurious limit point is
what bounds its mean error at high \gls{SNR}. Tensor \gls{MUSIC} instead
relies on the eigen-decomposition of a spatially smoothed covariance and
requires no explicit tensor factorization, so that it has no convergence
failure mode and its error keeps falling with \gls{SNR}, at the price of a
resolution set by the subarray dimensions and the grid density. At low
\gls{SNR} the two are indistinguishable, because both inherit the same
\gls{MDL} model order and it is the model order, not the estimator, that
governs the error there. The smoothed covariance also needs a fast mode long
enough to smooth over, whereas \gls{ESPRIT} needs only one shift invariance
and, through Eq.~\eqref{eq:C_qy_norm}, none at all. The computational cost of
tensor \gls{MUSIC} is dominated by the $G^2$ steering-vector evaluations and
by the eigen-decomposition of $\hat{\mathbf R}_{\mathrm{fb}}$, whereas
\gls{ALS} scales with the number of iterations and with the Khatri--Rao
product dimensions. Section~\ref{sec:results_ce} confirms each of these
predictions in turn, since the two estimators coincide below $10$~dB and at
low path count, they separate in favor of \gls{ESPRIT} as $L_u$ approaches
the Kruskal ceiling and in favor of \gls{MUSIC} above $20$~dB. Tensor
\gls{MUSIC} needs $K \geq 4$ before it matches \gls{ALS} on the morphing
axis.

\begin{table}[t]
\centering
\caption{Comparison of the two estimation pipelines.}
\label{tab:pipeline_comparison}
\renewcommand{\arraystretch}{1.15}
\begin{tabular}{|l|c|c|}
\hline
\textbf{Property} & \textbf{\glsentryshort{ALS}\,+\,\glsentryshort{ESPRIT}} & \textbf{Tensor \glsentryshort{MUSIC}} \\
\hline
Model assumption & Exact \gls{CP} & Subspace \\
Initialization & \gls{GEVD} & Not required \\
Resolution & Super-resolution & Grid-limited \\
Low \gls{SNR} & \gls{MDL}-limited & \gls{MDL}-limited \\
High \gls{SNR} & Convergence floor & Monotone \\
Near the Kruskal ceiling & Retains accuracy & Saturates \\
Minimum fast mode & $K = 1$ (Eq.~\eqref{eq:C_qy_norm}) & $K \geq 4$ \\
Computational cost & $\mathcal O(T_{\mathrm{it}} \cdot NKIJ L_u)$ & $\mathcal O(G^2 M + M^3)$ \\
\hline
\end{tabular}
\end{table}

\subsection{Relationship Between Wavevectors and Angles of Arrival}

Each column~$\ell$ of the spatial factor matrix $\mathbf A_u$ encodes
the wavevector $\mathbf q_{u,\ell}\in\bbR^3$ through the element-wise
phase
$[\mathbf A_u]_{n,\ell}=\exp(-\mathrm{j}\,\mathbf q_{u,\ell}^T
  \mathbf p_n^{(0)})$.
Under far-field propagation at wavelength~$\lambda$, the wavevector is
related to the \gls{AoA}
$(\theta_{u,\ell},\phi_{u,\ell})$ by
\begin{equation}
  \mathbf q_{u,\ell}
    =\frac{2\pi}{\lambda}
    \begin{bmatrix}
      \sin\theta_{u,\ell}\cos\phi_{u,\ell}\\
      \sin\theta_{u,\ell}\sin\phi_{u,\ell}\\
      \cos\theta_{u,\ell}
    \end{bmatrix},
  \label{eq:C_wavevector_aoa}
\end{equation}
where $\theta_{u,\ell}$ is the elevation angle measured from the
$z$-axis and $\phi_{u,\ell}$ is the azimuth angle measured from the
$x$-axis in the $x$-$y$ plane.

\subsection{Channel Reconstruction for Arbitrary Configurations}
\label{sec:reconstruction}

Once the factor matrices
$\hat{\mathbf A}_u$, $\hat{\mathbf B}_u$,
$\hat{\mathbf C}_u$, $\hat{\mathbf D}_u$
and the path gains $\hat{\boldsymbol\beta}_u$ have been estimated, the
channel vector at any configuration $(k',i',j')$, including those not
in the training grid, can be reconstructed as
\begin{equation}
  \hat{\mathbf h}_u[k',i',j']
    = \hat{\mathbf A}_u\,\mathrm{Diag}(\hat{\boldsymbol\beta}_u)
      \left(
        \hat{\mathbf c}_{z,u}(j')
        \hadamard\hat{\mathbf c}_{x,u}(i')
        \hadamard\hat{\mathbf b}_u(k')
      \right),
  \label{eq:C_recon}
\end{equation}
where $\hat{\mathbf b}_u(k')$, $\hat{\mathbf c}_{x,u}(i')$,
$\hat{\mathbf c}_{z,u}(j')$ are the steering vectors evaluated at the
new configuration indices using the estimated wavevectors.
This extrapolation capability is a key advantage of the tensor-based approach,
because the training overhead $IJKP$ of Eq.~\eqref{eq:B_train_time} covers
only the grid configurations used for estimation, while the reconstructed
channel remains valid at arbitrary configurations within the angular support.
Section~\ref{sec:beamforming} builds the configuration search on precisely
this property, and Section~\ref{sec:results_ce} measures the accuracy of
Eq.~\eqref{eq:C_recon} at a position half a step beyond the training grid, so
what is reported is extrapolation, not interpolation.

\subsection{Identifiability of the Twin-Time-Scale \glsentryshort{CP} Model}
\label{sec:uniqueness}

This subsection establishes when the factors of Eq.~\eqref{eq:B_CP_user} are
determined by the data, which is the precondition under which
problem~(P1) of Eq.~\eqref{eq:B_P1} is well posed. We first recall the
$d$-way Kruskal condition, then compute the Kruskal ranks of the four
\gls{GFIM} factors from the array geometry, and finally translate the result
into the system dimensions and discuss what it means in practice.

Recall that the Kruskal rank $k_{\mathbf M}$ of a matrix $\mathbf M$ is the
largest integer $r$ such that every set of $r$ columns of $\mathbf M$ is
linearly independent, so that $k_{\mathbf M}\leq\operatorname{rank}(\mathbf M)$
always. The uniqueness result we use is the $d$-way generalization, due to
Sidiropoulos and Bro~\cite{sidiropoulos2000uniqueness}, of Kruskal's
theorem for three-way arrays~\cite{kruskal1977}.

\begin{lemma}
\label{lem:kruskal}
Let $\mathcal X = \llbracket \mathbf F^{(1)},\dots,\mathbf F^{(d)}\rrbracket$
be a $d$-way \gls{CP} model of rank $L$. If
\begin{equation}
  \sum_{m=1}^{d} k_{\mathbf F^{(m)}} \geq 2L + (d-1),
  \label{eq:C_kruskal_dway}
\end{equation}
then the decomposition is essentially unique, that is, any other rank-$L$
\gls{CP} model of $\mathcal X$ has factors that coincide with
$\{\mathbf F^{(m)}\}$ up to one common column permutation and $d$ diagonal
scalings, whose product is the identity.
\end{lemma}

For $d = 4$ the right-hand side of Eq.~\eqref{eq:C_kruskal_dway} is $2L_u+3$,
so the condition results in:
\begin{equation}
  k_{\widetilde{\mathbf A}_u}+k_{\mathbf B_u}+k_{\mathbf C_u}
  +k_{\mathbf D_u}\geq 2L_u+3.
  \label{eq:C_kruskal}
\end{equation}
What makes Eq.~\eqref{eq:C_kruskal} usable here is that the \gls{GFIM}
geometry fixes all four Kruskal ranks in closed form, through the following
classical property of Vandermonde matrices.

\begin{lemma}
\label{lem:vandermonde}
Let $\mathbf V\in\bbC^{M\times L}$ have entries
$[\mathbf V]_{m,\ell}=z_\ell^{\,m-1}$ with generators
$z_1,\dots,z_L\in\bbC\setminus\{0\}$. If the generators are pairwise
distinct, then $k_{\mathbf V}=\min(M,L)$.
\end{lemma}

\begin{IEEEproof}
Let $r=\min(M,L)$ and take any $r$ columns, indexed by
$\mathcal I\subseteq\{1,\dots,L\}$ with $|\mathcal I|=r$. The submatrix
formed by the first $r$ rows of those columns is a square Vandermonde matrix
with nodes $\{z_\ell\}_{\ell\in\mathcal I}$, whose determinant is
$\prod_{\ell<\ell',\,\ell,\ell'\in\mathcal I}(z_{\ell'}-z_\ell)\neq 0$
because the nodes are distinct. The $r$ columns are therefore linearly
independent, and since no set of more than $\min(M,L)$ columns of an
$M\times L$ matrix can be independent, we have $k_{\mathbf V}=r$.
\end{IEEEproof}

The three motion factors of the \gls{GFIM} are exactly of this form. Writing
the mode generators as
\begin{equation}
  b_{u,\ell}=e^{-\mathrm{j}q^{(y)}_{u,\ell}\delta_y},\quad
  c_{u,\ell}=e^{-\mathrm{j}q^{(x)}_{u,\ell}\delta_x},\quad
  d_{u,\ell}=e^{-\mathrm{j}q^{(z)}_{u,\ell}\delta_z},
  \label{eq:C_generators}
\end{equation}
the columns of $\mathbf B_u$, $\mathbf C_u$ and $\mathbf D_u$ are the
Vandermonde vectors generated by $b_{u,\ell}$, $c_{u,\ell}$ and
$d_{u,\ell}$, since the commanded steps are uniform,
$\Delta_x[i]=(i-1)\delta_x$ and $\Delta_z[j]=(j-1)\delta_z$, and the piston
morphing gives $y[k]=(k-1)\delta_y$. This yields the following
characterization.

\begin{proposition}
\label{prop:kranks}
Consider the model of Eq.~\eqref{eq:B_CP_user} with mode steps satisfying the
unambiguity condition $\delta_x,\delta_y,\delta_z<\lambda/2$ and gains
$\beta_{u,\ell}\neq 0$ for all $\ell$. If the $L_u$ wavevectors have pairwise
distinct $y$-projections, then $k_{\mathbf B_u}=\min(K,L_u)$, and likewise
pairwise distinct $x$- and $z$-projections give $k_{\mathbf C_u}=\min(I,L_u)$
and $k_{\mathbf D_u}=\min(J,L_u)$. For the spatial factor of an
$N_x\times N_z$ \gls{UPA}, we have
\begin{equation}
  k_{\widetilde{\mathbf A}_u}
    \;\geq\;\min\!\big(k_{\mathbf A_x}+k_{\mathbf A_z}-1,\;L_u\big),
  \label{eq:C_krank_kr}
\end{equation}
with equality to $\min(N,L_u)$ for almost every realization of the angles
under a continuous angular distribution.
\end{proposition}

\begin{IEEEproof}
Consider $\mathbf B_u$. Two generators coincide, $b_{u,\ell}=b_{u,\ell'}$,
if and only if $(q^{(y)}_{u,\ell}-q^{(y)}_{u,\ell'})\delta_y$ is an integer
multiple of $2\pi$. Since $|q^{(y)}|\leq k_0=2\pi/\lambda$ for every visible
path, the difference satisfies
$|q^{(y)}_{u,\ell}-q^{(y)}_{u,\ell'}|\delta_y\leq 2k_0\delta_y<2\pi$ under
$\delta_y<\lambda/2$, so the only admissible multiple is zero and the
generators coincide precisely when the $y$-projections do. Distinct
projections therefore give distinct generators, and
Lemma~\ref{lem:vandermonde} gives $k_{\mathbf B_u}=\min(K,L_u)$. The same
argument applies verbatim to $\mathbf C_u$ and $\mathbf D_u$ with
$\delta_x$ and $\delta_z$. For the spatial factor, the \gls{UPA} manifold
factors as $\mathbf A_u=\mathbf A_x\khatri\mathbf A_z$ with
$\mathbf A_x\in\bbC^{N_x\times L_u}$ and
$\mathbf A_z\in\bbC^{N_z\times L_u}$ Vandermonde, and
Eq.~\eqref{eq:C_krank_kr} is the Khatri--Rao Kruskal-rank bound
of~\cite[Lemma~1]{sidiropoulos2000uniqueness}. Every $\min(N,L_u)$-column
subset of $\mathbf A_u$ has a non-vanishing minor that is a non-trivial
analytic function of the angles, so the set on which it vanishes has Lebesgue
measure zero, which gives the almost-sure equality. Finally,
$\widetilde{\mathbf A}_u=\mathbf A_u\mathrm{Diag}(\boldsymbol\beta_u)$ and a
diagonal scaling with non-zero entries leaves every Kruskal rank unchanged.
\end{IEEEproof}

\begin{corollary}
\label{cor:dims}
Under the hypotheses of Proposition~\ref{prop:kranks}, the per-user \gls{CP}
model of Eq.~\eqref{eq:B_CP_user} is essentially unique whenever we have:
\begin{equation}
\begin{split}
  &\min(N,L_u)+\min(K,L_u)\\
  &\qquad+\min(I,L_u)+\min(J,L_u)\geq 2L_u+3.
\end{split}
  \label{eq:C_kruskal_dim}
\end{equation}
The largest $L_u$ satisfying Eq.~\eqref{eq:C_kruskal_dim} is denoted by
$L_u^{\max}(N,K,I,J)$ and it is the bound that constrains problem~(P1).
\end{corollary}

\begin{table}[t]
\centering
\caption{Identifiability bound $L_u^{\max}$ of Corollary~\ref{cor:dims}
  for selected system dimensions.}
\label{tab:uniqueness}
\renewcommand{\arraystretch}{1.1}
\begin{tabular}{|c|c|c|c|c|}
\hline
$N$ & $K$ & $I$ & $J$ & $L_u^{\max}$\\
\hline
16 & 4 & 4 & 4 &  9 \\
32 & 4 & 8 & 8 & 17 \\
64 & 2 & 4 & 4 &  7 \\
64 & 8 & 8 & 8 & 21 \\
\hline
\end{tabular}
\end{table}

Table~\ref{tab:uniqueness} evaluates $L_u^{\max}$ for representative
dimensions. Three consequences deserve comment, and the numerical campaign of
Section~\ref{sec:results} tests each of them.

\begin{remark}
\label{rem:fourmodes}
For $L_u$ larger than every mode dimension, Eq.~\eqref{eq:C_kruskal_dim}
reads $N+K+I+J\geq 2L_u+3$, so each mode contributes its full dimension to
the budget, while the requirement grows at $2L_u$. A $d$-way model needs
$\sum_m I_m \geq 2L+(d-1)$, so adding a mode of dimension $I_{d+1}$ is
profitable whenever $I_{d+1}\geq 1$, and strictly so when $I_{d+1}\geq 2$.
This is why the fast morphing mode helps even at $K=2$, and it is the
structural reason for preferring the two-motion protocol over a single-motion
version that spends the same $KIJ$ budget on one axis.
\end{remark}

\begin{remark}
\label{rem:degenerate}
Proposition~\ref{prop:kranks} requires distinct projections per axis, and it
is worth asking how restrictive that is. Two paths sharing a single
projection, say $q^{(z)}_{u,1}=q^{(z)}_{u,2}$, collapse $k_{\mathbf D_u}$ to
one but leave the other three modes intact. At the reference geometry
$(N,K,I,J)=(16,4,4,4)$ with $L_u=3$ this gives $3+3+3+1=10\geq 9$, so
Eq.~\eqref{eq:C_kruskal_dim} still holds, and the model remains identifiable.
Identifiability is lost only if two modes collapse simultaneously, which by
Eq.~\eqref{eq:C_generators} requires two paths to share two of the three
wavevector components. The visible-hemisphere constraint
$\lVert\mathbf q_{u,\ell}\rVert=k_0$ with $q^{(y)}_{u,\ell}\leq 0$ makes the
third component a function of the other two, so sharing two components means
sharing all three, that is, being the same path. The twin-time-scale model is
therefore identifiable for every admissible geometry with distinct paths, and
not merely for generic ones. What a near-degenerate geometry costs is
conditioning rather than uniqueness, which is a finite-\gls{SNR} effect that
Eq.~\eqref{eq:C_kruskal_dim} does not see.
\end{remark}

\begin{remark}
\label{rem:sufficiency}
Eq.~\eqref{eq:C_kruskal_dim} is sufficient and not necessary, and it concerns
exact uniqueness in the noiseless case. A crude necessary condition follows
from parameter counting, since a rank-$L_u$ 4-way \gls{CP} model has
$L_u(N+K+I+J-3)$ free parameters against $NKIJ$ observations, giving
$L_u \leq NKIJ/(N+K+I+J-3)$, which at the reference geometry permits
$L_u\leq 40$ against the Kruskal value of $9$. The gap between the two shows
that Kruskal is the binding constraint at these dimensions. Neither bound,
however, orders two protocols by their estimation accuracy, since two designs
may be ranked one way by $L_u^{\max}$ and the opposite way by the conditioning
of the \gls{ALS} fit at a finite \gls{SNR}. Section~\ref{sec:results_precoder}
exhibits precisely this inversion, where the translation-only protocol has the
larger Kruskal ceiling yet the larger estimation penalty.
\end{remark}

\begin{remark}
\label{rem:dof}
The counting above treats the factors as unstructured. The \gls{GFIM} model is
far more constrained, since each path contributes two angular degrees of
freedom and one complex gain, that is, $4L_u$ real unknowns per user against
$2NKIJ$ real observations. This is the formal statement of the training
efficiency claimed in Section~\ref{sec:reconstruction}, because a parametric
model identified from $IJKP$ pilot symbols predicts the channel over the whole
continuum of panel configurations, whereas an unstructured estimator would
have to measure each configuration it intends to use.
\end{remark}

\subsection{Training Allocation Under the Identifiability Constraint}
\label{sec:allocation}

Corollary~\ref{cor:dims} bounds the number of recoverable paths for a given
geometry. The design question runs the other way, since the training frame
$IJKP$ of Eq.~\eqref{eq:B_train_time} is what the coherence time allows and
the panel dimensions $(K,I,J)$ are what the designer chooses within it. This
subsection asks how a fixed motion budget $B\triangleq KIJ$ should be split
among the three motion modes, and what Eq.~\eqref{eq:C_kruskal_dim} says
about that split. The pilot length $P$ is fixed separately by the
orthogonality requirement $P\geq U$ of Section~\ref{sec:ortho_pilots}, so
the allocation problem concerns $B$ alone.

Two regimes of Eq.~\eqref{eq:C_kruskal_dim} admit a closed form. Let
$D_{\max}\triangleq\max(K,I,J)$, $S\triangleq K+I+J-3$ and
$T\triangleq\lfloor (N+K+I+J-3)/2\rfloor$.

\begin{proposition}
\label{prop:ceiling}
Under the hypotheses of Proposition~\ref{prop:kranks}, the bound
$L_u^{\max}$ of Corollary~\ref{cor:dims} satisfies
\begin{equation}
  L_u^{\max} =
  \begin{cases}
    S, & \text{if } D_{\max}\leq S\leq N,\\[3pt]
    T, & \text{if } N<S \text{ and } T\geq D_{\max}.
  \end{cases}
  \label{eq:C_ceiling}
\end{equation}
\end{proposition}

\begin{IEEEproof}
In the first regime, take $L=S$. Since $D_{\max}\leq S\leq N$, every minimum
in Eq.~\eqref{eq:C_kruskal_dim} except the first saturates at its mode
dimension and the first equals $L$, so the left-hand side is $L+K+I+J=2L+3$
and the condition holds with equality. For any $L>S$ the left-hand side is at
most $L+K+I+J<2L+3$, because $\min(N,L)\leq L$, so $S$ is the largest
admissible order. In the second regime, take
$L=T$. The hypothesis $N<S$ gives $T>N$, and $T\geq D_{\max}$ by
assumption, so all four minima saturate and the condition reads
$N+K+I+J\geq 2L+3$, which holds for $L=T$ by the definition of $T$ and
fails for $L=T+1$.
\end{IEEEproof}

The first regime is the one of practical panels, in which the array is large
relative to the motion grid, and it states that the ceiling grows with the
sum of the motion dimensions and not with their product. For a fixed budget $B=KIJ$ that sum is
largest when the factorization is most unbalanced, so the Kruskal ceiling
rewards concentrating the budget on one motion axis.
Table~\ref{tab:allocation} makes this explicit for the $B=64$ budget of the
numerical campaign. The balanced reference protocol $(4,4,4)$ has the lowest
ceiling of all non-degenerate allocations, while $(1,4,16)$ nearly doubles
it with the same number of training symbols. The one allocation the ceiling
excludes outright is the one that concentrates everything, since two modes of
unit length reduce the 4-way tensor to a matrix and
Eq.~\eqref{eq:C_kruskal_dim} then certifies nothing, which is the
$(64,1,1)$ row and the morphing-only protocol of
Section~\ref{sec:results_precoder}.

\begin{table}[t]
\centering
\caption{Identifiability ceiling $L_u^{\max}$ for every allocation of the
  motion budget $B = KIJ = 64$ at $N = 16$. Rows are ordered by the ceiling,
  and $(K,I,J)$ is listed in non-decreasing order because
  Eq.~\eqref{eq:C_kruskal_dim} is symmetric in the three motion modes.}
\label{tab:allocation}
\renewcommand{\arraystretch}{1.1}
\begin{tabular}{|c|c|c|c|}
\hline
$(K,I,J)$ sorted & $K+I+J$ & $L_u^{\max}$ & Regime of Eq.~\eqref{eq:C_ceiling}\\
\hline
$(1,4,16)$ & 21 & 17 & second \\
$(2,2,16)$ & 20 & 16 & second \\
$(1,2,32)$ & 35 & 16 & neither \\
$(1,8,8)$  & 17 & 14 & first \\
$(2,4,8)$  & 14 & 11 & first \\
$(4,4,4)$  & 12 &  9 & first \\
$(1,1,64)$ & 66 &  0 & degenerate \\
\hline
\end{tabular}
\end{table}

The ceiling is, however, a feasibility constraint and not a design objective,
and Section~\ref{sec:results_precoder} quantifies the difference. There the
specific allocation of $(2,8,4)$, whose ceiling is $11$, delivers a smaller estimation
penalty than $(1,16,4)$, whose ceiling is $17$, at every path count above
four. Unbalancing the budget raises the number of paths that can be told
apart in principle while lowering the conditioning with which the factors
are separated at a finite \gls{SNR}, because a mode of length one contributes
no shift invariance to Eq.~\eqref{eq:C_esprit_freq} and a mode of length two
contributes a single one. The particular allocation rule that the analysis and the
measurements together support is therefore to choose $(K,I,J)$ so that
Eq.~\eqref{eq:C_ceiling} exceeds the largest path count the deployment must
resolve, with a margin, and then to spend whatever budget remains on balancing the
modes rather than on lengthening the longest one. Every mode should be at
least two states long, since that is the minimum at which it contributes a
shift invariance, and no two modes should be left at unit length.

\section{Multiuser Beamforming with Configuration Selection}
\label{sec:beamforming}

The channel estimation framework of Section~\ref{sec:estimation} recovers the
per-user wavevectors $\hat{\mathbf q}_{u,\ell}$ and path gains
$\hat{\beta}_{u,\ell}$ from the training tensor. This section shows how to
exploit those parametric estimates for downlink multiuser beamforming. Because
the tensor-based estimators provide a generative channel model rather than a
table of measured responses, the \gls{BS} can evaluate the channel at any
configuration $(k,i,j)$ through Eq.~\eqref{eq:C_recon}, including
configurations never visited during training, and can therefore select the best
one by exhaustively searching the $K \times I \times J$ grid. This is the
operational payoff of the identifiability analysis of
Section~\ref{sec:uniqueness}, since a configuration search over a
non-identifiable model would optimize a channel that the data do not
determine.

\subsection{Downlink Signal Model}

Consider the downlink of a system in which the \gls{BS}, equipped with the $N$-element \gls{GFIM}, simultaneously serves $U$ single-antenna users. At a given panel configuration $(k,i,j)$, the \gls{BS} transmits the precoded signal
\begin{equation}
  \mathbf x = \sqrt{P_t}\,\mathbf W \mathbf d,
  \label{eq:D_tx_signal}
\end{equation}
where $\mathbf d = [d_1, \dots, d_U]^T \in \bbC^U$ is the vector of data symbols with $\mathbb{E}[\mathbf d \mathbf d^H] = \mathbf I_U$, $P_t$ is the total transmit power, and $\mathbf W = [\mathbf w_1, \dots, \mathbf w_U] \in \bbC^{N \times U}$ is the precoding matrix whose columns are normalized to unit norm, $\|\mathbf w_u\| = 1$ for all $u$.

The received signal at user $u$ is
\begin{equation}
  r_u = \sqrt{P_t}\,\mathbf h_u^H[k,i,j]\,\mathbf W \mathbf d + n_u,
  \label{eq:D_rx_signal}
\end{equation}
where $\mathbf h_u[k,i,j] \in \bbC^N$ is the channel vector from the \gls{BS} to user $u$ at configuration $(k,i,j)$, and $n_u \sim \mathcal{CN}(0,\sigma_n^2)$ is additive noise. Defining $\mathrm{SNR} \triangleq P_t / \sigma_n^2$ and equal per-user power allocation $P_u = P_t / U$, the \gls{SINR} of user $u$ is
\begin{equation}
  \mathrm{SINR}_u
    = \frac{P_u \,|\mathbf h_u^H \mathbf w_u|^2}
           {\sum_{v \neq u} P_u\, |\mathbf h_u^H \mathbf w_v|^2 + \sigma_n^2}.
  \label{eq:D_sinr}
\end{equation}
The sum-rate (sum spectral efficiency) under Gaussian signalling is formulated as:
\begin{equation}
  R_{\mathrm{sum}} = \sum_{u=1}^{U} \log_2(1 + \mathrm{SINR}_u) \quad [\text{bps/Hz}].
  \label{eq:D_sumrate}
\end{equation}

\subsection{Linear Precoding Schemes}
\label{sec:precoders}

We consider three standard linear \glspl{TPC} that span the interference-management design space~\cite{spencer2004zf,peel2005rzf}. Let $\mathbf H = [\mathbf h_1, \dots, \mathbf h_U] \in \bbC^{N \times U}$ denote the aggregate channel matrix at the selected configuration.

\subsubsection{Zero-forcing}
The \gls{ZF} \gls{TPC} eliminates multiuser interference by inverting the channel's Gram matrix,
\begin{equation}
  \mathbf W_{\mathrm{ZF}}
    = \mathbf H\left(\mathbf H^H \mathbf H\right)^{-1},
  \label{eq:D_zf}
\end{equation}
followed by column-wise normalization $\mathbf w_u \leftarrow \mathbf w_u / \|\mathbf w_u\|$. This \gls{TPC} requires $N \geq U$ and achieves zero inter-user interference, i.e., $\mathbf h_u^H \mathbf w_v = 0$ for all $v \neq u$ before normalization. The inversion nonetheless amplifies noise whenever the channel is ill-conditioned.

\subsubsection{Regularized zero-forcing}
The \gls{MMSE} \gls{TPC}, also known as regularized \gls{ZF}~\cite{peel2005rzf}, introduces a regularization term that trades residual interference for reduced noise amplification,
\begin{equation}
  \mathbf W_{\mathrm{MMSE}}
    = \mathbf H\left(\mathbf H^H \mathbf H
      + \frac{U}{\mathrm{SNR}}\,\mathbf I_U\right)^{-1},
  \label{eq:D_mmse}
\end{equation}
again followed by column normalization. The regularization weight $U / \mathrm{SNR}$ balances interference suppression against noise enhancement, and the \gls{TPC} reduces to \gls{ZF} as $\mathrm{SNR} \to \infty$. Section~\ref{sec:results_precoder} shows that the two coincide numerically until the load reaches $U = N/2$, which is where the inversion in Eq.~\eqref{eq:D_zf} first becomes ill-conditioned.

\subsubsection{Matched filtering}
The \gls{MRC} \gls{TPC} maximizes the desired-signal power by matched filtering, formulated as:
\begin{equation}
  \mathbf W_{\mathrm{MRC}} = \mathbf H,
  \label{eq:D_mrc}
\end{equation}
with column normalization. This \gls{TPC} is optimal in the single-user or noise-dominated regime but does not suppress inter-user interference; hence the numerator and the interference term of Eq.~\eqref{eq:D_sinr} grow together, and the sum-rate saturates as $U$ or the $\mathrm{SNR}$ increases. The three \glspl{TPC} are summarized in Algorithm~\ref{alg:precoders}.

\begin{algorithm}[t]
\caption{Linear precoding and sum-rate evaluation}
\label{alg:precoders}
\begin{algorithmic}[1]
\STATE \textbf{Input:} $\mathbf H \in \bbC^{N \times U}$, $\mathrm{SNR}$, \gls{TPC} type $\in \{\mathrm{ZF}, \mathrm{MMSE}, \mathrm{MRC}\}$.
\IF{\gls{TPC} $= \mathrm{ZF}$}
  \STATE $\mathbf W \leftarrow \mathbf H(\mathbf H^H \mathbf H)^{-1}$
\ELSIF{\gls{TPC} $= \mathrm{MMSE}$}
  \STATE $\mathbf W \leftarrow \mathbf H(\mathbf H^H \mathbf H + \frac{U}{\mathrm{SNR}}\mathbf I_U)^{-1}$
\ELSIF{\gls{TPC} $= \mathrm{MRC}$}
  \STATE $\mathbf W \leftarrow \mathbf H$
\ENDIF
\STATE Normalize: $\mathbf w_u \leftarrow \mathbf w_u / \|\mathbf w_u\|$, $\forall u$.
\STATE Compute $\mathrm{SINR}_u$ from Eq.~\eqref{eq:D_sinr}, $\forall u$.
\STATE $R_{\mathrm{sum}} \leftarrow \sum_{u=1}^{U} \log_2(1 + \mathrm{SINR}_u)$.
\STATE \textbf{Output:} $R_{\mathrm{sum}}$.
\end{algorithmic}
\end{algorithm}

\subsection{Channel Prediction and Configuration Selection}
\label{sec:prediction}

A key benefit of the tensor-based estimation framework is its ability to predict the channel at any \gls{GFIM} configuration without additional pilot transmissions. Using the estimated per-user parameters $\{\hat{\mathbf q}_{u,\ell}, \hat{\beta}_{u,\ell}\}_{\ell=1}^{L_u}$ from either the \gls{ALS}\,+\,\gls{ESPRIT} or tensor \gls{MUSIC} pipeline, the predicted channel vector at an arbitrary configuration $(k', i', j')$ is (cf.\ Eq.~\eqref{eq:C_recon})
\begin{equation}
  \big[\hat{\mathbf h}_u[k', i', j']\big]_n
    = \sum_{\ell=1}^{L_u} \hat{\beta}_{u,\ell}\,
      e^{-\mathrm{j}\,\hat{\mathbf q}_{u,\ell}^T
        \left(\mathbf p_n^{(0)} + \boldsymbol\delta[i',j'] + y[k']\mathbf e_y\right)},
  \label{eq:D_predict}
\end{equation}
for $n = 1,\dots,N$, where $\mathbf e_y = [0,1,0]^T$ and the position model of
Eq.~\eqref{eq:B0_fullpos} is evaluated at the new configuration indices. This
is Eq.~\eqref{eq:C_recon} written entry by entry, and the two are the same
prediction, since the factor matrices of the reconstruction are themselves
generated by the estimated wavevectors.

The \gls{BS} exploits this prediction capability to select the panel configuration that maximizes the downlink sum-rate. For each precoding criterion, the optimal configuration is found by exhaustive search,
\begin{equation}
  (k^\star, i^\star, j^\star)
    = \arg\max_{k,i,j}\;
      R_{\mathrm{sum}}\!\big(\hat{\mathbf H}[k,i,j]\big),
  \label{eq:D_exh_search}
\end{equation}
where $\hat{\mathbf H}[k,i,j] = [\hat{\mathbf h}_1[k,i,j], \dots, \hat{\mathbf h}_U[k,i,j]]$. This search requires $K I J$ evaluations, each involving one channel prediction by Eq.~\eqref{eq:D_predict}, one \gls{TPC} computation, and one \gls{SINR} evaluation by Algorithm~\ref{alg:precoders}. No additional pilot overhead is incurred at any point because the search uses only the parametric model estimated during training, so the training cost remains the $IJKP$ symbols of Eq.~\eqref{eq:B_train_time} regardless of the size of the searched grid. Two caveats follow the estimates through the search and are quantified in Section~\ref{sec:results_precoder}. The rate that Eq.~\eqref{eq:D_exh_search} maximizes is a predicted rate and is optimistically biased, since the same estimate both chooses the configuration and scores it, and the rate that the system actually obtains is the one delivered by the selected \gls{TPC} on the true channel. The complete beamforming pipeline is given in Algorithm~\ref{alg:beamforming}.

\begin{algorithm}[t]
\caption{Configuration-selective multiuser beamforming}
\label{alg:beamforming}
\begin{algorithmic}[1]
\STATE \textbf{Input:} Per-user estimates $\{\hat{\mathbf q}_{u,\ell}, \hat{\beta}_{u,\ell}\}$ for $u = 1, \dots, U$; geometry $(N, K, I, J)$; \gls{TPC} type.
\STATE $R^\star \leftarrow 0$
\FOR{$k = 1, \dots, K$}
  \FOR{$i = 1, \dots, I$}
    \FOR{$j = 1, \dots, J$}
      \FOR{$u = 1, \dots, U$}
        \STATE Predict $\hat{\mathbf h}_u[k,i,j]$ via Eq.~\eqref{eq:D_predict}.
      \ENDFOR
      \STATE Assemble $\hat{\mathbf H}[k,i,j] = [\hat{\mathbf h}_1, \dots, \hat{\mathbf h}_U]$.
      \STATE Compute $R_{\mathrm{sum}}$ via Algorithm~\ref{alg:precoders}.
      \IF{$R_{\mathrm{sum}} > R^\star$}
        \STATE $R^\star \leftarrow R_{\mathrm{sum}}$; $(k^\star, i^\star, j^\star) \leftarrow (k, i, j)$.
      \ENDIF
    \ENDFOR
  \ENDFOR
\ENDFOR
\STATE \textbf{Output:} Optimal configuration $(k^\star, i^\star, j^\star)$, sum-rate $R^\star$.
\end{algorithmic}
\end{algorithm}

\subsection{Fixed-Panel Baseline}
\label{sec:traditional}

To quantify the benefit of \gls{GFIM} reconfigurability, we compare it against a traditional baseline in which the array is held at the reference position $\mathbf p_n^{(0)}$, with neither translation nor morphing. The \gls{BS} spends the whole training frame on the same fixed channel and averages,
\begin{equation}
  \hat{\mathbf h}_u^{\mathrm{trad}}
    = \mathbf h_u[1,1,1]
    + \mathbf n_u^{\mathrm{avg}},
  \label{eq:D_trad}
\end{equation}
where the effective noise variance is reduced by the averaging factor $N_{\mathrm{pilot}} = N I J K$, which yields $\mathrm{SNR}_{\mathrm{eff}} = \mathrm{SNR} \cdot N_{\mathrm{pilot}}$. The baseline is thus generous by construction, since it collects the same aperture and the identical training frame as the \gls{GFIM} curves, converts all of it into noise suppression, and still cannot exploit the diversity that panel reconfiguration provides. The aggregate channel $\hat{\mathbf H}^{\mathrm{trad}} = [\hat{\mathbf h}_1^{\mathrm{trad}}, \dots, \hat{\mathbf h}_U^{\mathrm{trad}}]$ is then used directly for precoding with no configuration selection, so that any gap to it in Section~\ref{sec:results_precoder} measures reconfiguration and nothing else.

\subsection{Complexity Analysis}
\label{sec:complexity}

The overall computational cost per coherence block splits into an estimation phase and a configuration search phase. The estimation phase runs once per coherence block and costs, for each user, $\mathcal O(T_{\mathrm{it}} \cdot NKIJ \cdot L_u)$ with \gls{ALS}\,+\,\gls{ESPRIT} or $\mathcal O(G^2 M + M^3)$ with tensor \gls{MUSIC}, where $M = L_{n_x} L_{n_z} L_k L_i L_j$ is the number of subarray snapshots entering the smoothed covariance. The configuration search phase then performs $KIJ$ evaluations of channel prediction, at cost $\mathcal O(NUL)$ per configuration, each followed by precoding at cost $\mathcal O(NU^2)$ for \gls{ZF}/\gls{MMSE} or $\mathcal O(NU)$ for \gls{MRC}, so that the total search cost is $\mathcal O[KIJ \cdot N(UL + U^2)]$, which is modest for typical system dimensions. Section~\ref{sec:results_precoder} measures the resulting split and finds the search to be under $1\%$ of the per-coherence-block cost, so that the estimator dominates. The search can be accelerated further by gradient-based optimization if the discrete grid is replaced by a continuous parametrization of the panel displacements.

\section{Numerical Results}
\label{sec:results}

This section evaluates the framework based on the Monte-Carlo campaign summarized in
Table~\ref{tab:sim_setup}. All sweeps share a single reference operating point
and each varies one axis at a time, and every point is the average of $1000$ independent trials
with the scenario and the noise redrawn per trial from a fixed seed.

% Simulation configuration table.  Values mirror the campaign specification
% (campaigns.yaml / spec.json) that produced data/channel_estimation/*.csv;
% the provenance hashes are listed in data/channel_estimation/PROVENANCE.csv.
\begin{table}[t]
\centering
\caption{Simulation configuration. The swept parameter of each figure overrides its entry.}
\label{tab:sim_setup}
\renewcommand{\arraystretch}{1.15}
\begin{tabular}{|l|c|}
\hline
\textbf{Parameter} & \textbf{Value} \\
\hline
\multicolumn{2}{|l|}{\textit{Array and system}} \\
\hline
Panel geometry $N_x \times N_z$      & $4 \times 4$ ($N = 16$) \\
Element spacing                      & $\lambda/2$, isotropic \\
Users $U$ (single antenna)           & $2$ \\
Pilot length $P$ (orthogonal)        & $4$ / $16$$^{\mathrm{a}}$ \\
Wavelength $\lambda$                 & $1$ (normalized) \\
\hline
\multicolumn{2}{|l|}{\textit{Twin-time-scale motion protocol}} \\
\hline
Translation grid $I \times J$        & $4 \times 4$ \\
Translation steps $\Delta_x, \Delta_z$ & $0.5\lambda$ \\
Morphing states $K$                  & $4$ \\
Morphing step $\Delta_y$             & $0.5\lambda$ \\
Training frame $IJKP$                & $256$ / $1024$ symbols$^{\mathrm{a}}$ \\
\hline
\multicolumn{2}{|l|}{\textit{Propagation}} \\
\hline
Paths per user $L_u$                 & $3$ \\
Elevation $\theta_{u,\ell}$          & $\mathcal U[45^\circ, 135^\circ]$ \\
Azimuth $\phi_{u,\ell}$              & $\mathcal U[-153^\circ, -27^\circ]$ \\
Min.\ path separation                & $0.3$~rad \\
Path gains $\beta_{u,\ell}$          & i.i.d.\ $\mathcal{CN}$, $\|\boldsymbol\beta_u\| = 1$ \\
\hline
\multicolumn{2}{|l|}{\textit{Estimators}} \\
\hline
ALS restarts (GEVD init)             & $10$ \\
MUSIC angular grid $G$               & $181 \times 181$ \\
MUSIC covariance                     & forward--backward \\
Model order $L_u$                    & MDL / genie \\
\hline
\multicolumn{2}{|l|}{\textit{Monte Carlo}} \\
\hline
Trials per point                     & $1000$ \\
Seed                                 & $0$ \\
Reference SNR                        & $10$~dB \\
Prediction position                  & half step past the grid \\
\hline
\multicolumn{2}{|l|}{\textit{Swept axes (one at a time)}} \\
\hline
SNR                                  & $-10$ to $30$~dB, step $5$~dB \\
Paths per user $L_u$                 & $\{1, 2, 3, 4, 6, 8, 9\}$ \\
Morphing states $K$                  & $1, \dots, 8$ \\
SNR of the $L_u$ sweep               & $\{0, 10, 20\}$~dB \\
\hline
\end{tabular}

\vspace{2pt}
\begin{minipage}{\columnwidth}
\footnotesize\raggedright
$^{\mathrm{a}}$ The channel-estimation campaigns of
Section~\ref{sec:results_ce} use $P = 4$; those of
Section~\ref{sec:results_precoder} use $P = 16$, since orthogonal pilots
require $P \geq U$ and the user axis runs to $U = 16$. Holding $P$ at the
larger value keeps the training budget constant along every \gls{TPC} axis.
\end{minipage}
\end{table}

Four metrics are reported, and each of them scores a specific equation of the
preceding sections. The channel \gls{NMSE} scores the fit of the estimated
tensor on the training grid,
\begin{equation}
  \mathrm{NMSE}
    = \mathbb E\!\left[
      \frac{\lVert \widehat{\mathcal H}_u - \mathcal H_u \rVert_F^2}
           {\lVert \mathcal H_u \rVert_F^2}\right],
  \label{eq:E_nmse}
\end{equation}
where $\widehat{\mathcal H}_u$ is assembled from the estimated parameters by
Eq.~\eqref{eq:B_CP_user} and $\mathcal H_u$ is the true channel tensor. The
prediction \gls{NMSE} replaces both tensors in Eq.~\eqref{eq:E_nmse} by the
channel vector of Eq.~\eqref{eq:C_recon} evaluated at a panel position half a step
beyond the training grid, so that it scores extrapolation rather than
interpolation and thereby tests the claim of
Section~\ref{sec:reconstruction} directly. The \gls{AoA} \gls{RMSE} scores
the angular parameters given by Eq.~\eqref{eq:C_aoa_recovery} after the estimated
paths are matched to the true ones. The achieved sum-rate, finally, is
Eq.~\eqref{eq:D_sumrate} evaluated on the true channel at the configuration
that Eq.~\eqref{eq:D_exh_search} selected from the estimate. The expectation
in Eq.~\eqref{eq:E_nmse} is taken in the linear domain and converted to dB
afterwards, which matters above $20$~dB and is revisited below.

Three references are used throughout. The genie-aided reference assumes
perfect \gls{CSI}, meaning the true wavevectors and path gains, together with
the exhaustive configuration search of Algorithm~\ref{alg:beamforming}, and it
therefore upper-bounds the achievable sum-rate. The estimated-\gls{CSI} curves
run the \gls{ALS}\,+\,\gls{ESPRIT} or the tensor \gls{MUSIC} pipeline of
Algorithm~\ref{alg:framework} on the noisy training tensor, with the model
order either detected by \gls{MDL} or supplied by a genie, and feed the
predicted channels to the same configuration search. The traditional
reference of Section~\ref{sec:traditional} keeps the array fixed and spends
the identical pilot budget on noise averaging, so that it isolates the gain
that comes from reconfiguration, rather than from training length.

\subsection{Channel Estimation}
\label{sec:results_ce}

Fig.~\ref{fig:nmse_vs_snr} shows the channel \gls{NMSE} of
Eq.~\eqref{eq:E_nmse} against the receive \gls{SNR}. Between $0$ and $10$~dB
all three curves fall at roughly $1.6$~dB/dB, and above $10$~dB they settle
onto a common slope close to $1$~dB/dB, which confirms that both pipelines
attain the same error exponent once the model order is resolved. Below $0$~dB
the \gls{MDL} variants fall $6$ to $8$~dB short of the genie-rank bound, and
Fig.~\ref{fig:rank_vs_snr} identifies the cause. At $-10$~dB the criterion of
Eq.~\eqref{eq:C_MDL_select} returns $\hat L_u = 1$ in every realization, so
that two of the three paths are simply not reconstructed. The penalty is
therefore a detection effect rather than an estimation one, since the genie curve
keeps its slope through the same region.

% Channel NMSE against receive SNR, reference geometry (Table~\ref{tab:sim_setup}).
% Data: data/channel_estimation/nmse_vs_snr.csv, 1000 trials per point.
% The plotted metric is the trial mean in the linear domain, converted to dB.
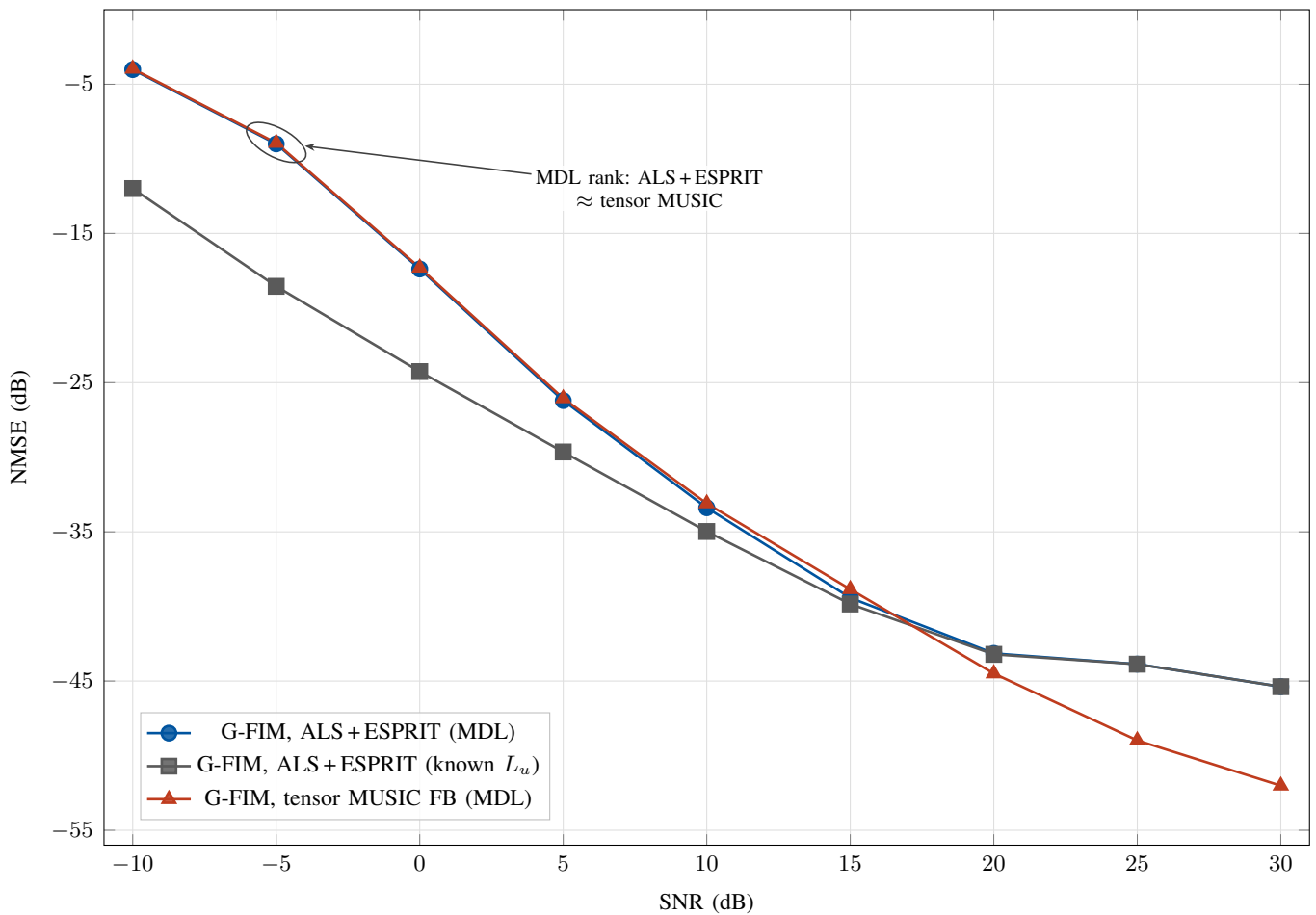
\begin{figure}[t]
\centering
\begin{tikzpicture}
\begin{axis}[
  resultaxis,
  xlabel={SNR (dB)},
  ylabel={NMSE (dB)},
  xmin=-11, xmax=31, xtick={-10,-5,...,30},
  ymin=-56, ymax=0, ytick={-55,-45,...,-5},
  legend pos=south west,
]
  \addplot[als]      table[col sep=comma, x=snr_db, y=als_mdl_nmse_db]
    {\cedata{nmse_vs_snr.csv}};
  \addlegendentry{\glsentryshort{GFIM}, ALS\,+\,ESPRIT (MDL)}
  \addplot[alsknown] table[col sep=comma, x=snr_db, y=als_known_nmse_db]
    {\cedata{nmse_vs_snr.csv}};
  \addlegendentry{\glsentryshort{GFIM}, ALS\,+\,ESPRIT (known $L_u$)}
  \addplot[music]    table[col sep=comma, x=snr_db, y=music_fb_mdl_nmse_db]
    {\cedata{nmse_vs_snr.csv}};
  \addlegendentry{\glsentryshort{GFIM}, tensor MUSIC FB (MDL)}
  % Curve group: both MDL-based estimators coincide up to 15 dB.
  \curvegroup[-28]{-5,-8.9}{9mm}{4mm}{8,-12}{MDL rank: ALS\,+\,ESPRIT\\$\approx$ tensor MUSIC}
\end{axis}
\end{tikzpicture}
\caption{Channel \gls{NMSE} against receive \gls{SNR} at the reference geometry of Table~\ref{tab:sim_setup}. Trial mean of the linear \gls{NMSE}, in dB.}
\label{fig:nmse_vs_snr}
\end{figure}

Above $20$~dB the mean of the two \gls{ALS} curves flattens near $-45$~dB
while tensor \gls{MUSIC} continues to improve. This crossing is an artifact of
a single trial out of $1000$ whose residual error dominates the linear mean of
Eq.~\eqref{eq:E_nmse}. Excluding that trial moves the \gls{ALS} curve at
$30$~dB from $-45.4$ to $-50.8$~dB, and the trial median of the same data
decreases monotonically to $-55.4$~dB. The \gls{CP} model is therefore not the
limiting factor at high \gls{SNR}, and what limits the mean is the occasional
convergence of \gls{ALS} to a spurious stationary point that the \gls{GEVD}
initialization of Section~\ref{sec:als_esprit} does not always avert.

The model-order behavior behind the low-\gls{SNR} gap is detailed in
Fig.~\ref{fig:rank_vs_snr}, which plots the empirical distribution of
$\hat L_u$ around the true order $L_u = 3$. Two features are worth noting. The
criterion never overestimates, in that $\Pr\{\hat L_u > 3\} = 0$ at every
\gls{SNR}, which is the behavior that the complexity penalty of
Eq.~\eqref{eq:C_MDL} is designed to produce and confirms that the effective
sample size $N_s$ argued for in Section~\ref{sec:mdl} is indeed large at these
dimensions. The transition is also sharp, since correct detection rises from
$0$ at $-10$~dB to $0.63$ at $0$~dB and to $0.96$ at $10$~dB, and this is what
makes the \gls{MDL} and genie curves of Fig.~\ref{fig:nmse_vs_snr}
indistinguishable above $10$~dB.

% Empirical distribution of the MDL model order against SNR.
% Data: data/channel_estimation/rank_detection_vs_snr.csv, built from the
% per-trial estimates of the SNR campaign: 1000 trials x 2 users = 2000
% per-user samples per point.  True order L_u = 3.  The MDL front-end is
% shared by both pipelines, so a single family of curves applies to each.
% Styling note: this figure compares MODEL ORDERS, not estimators, so it uses
% the rthree/rtwo/rone palette and never the als/alsknown/music one.
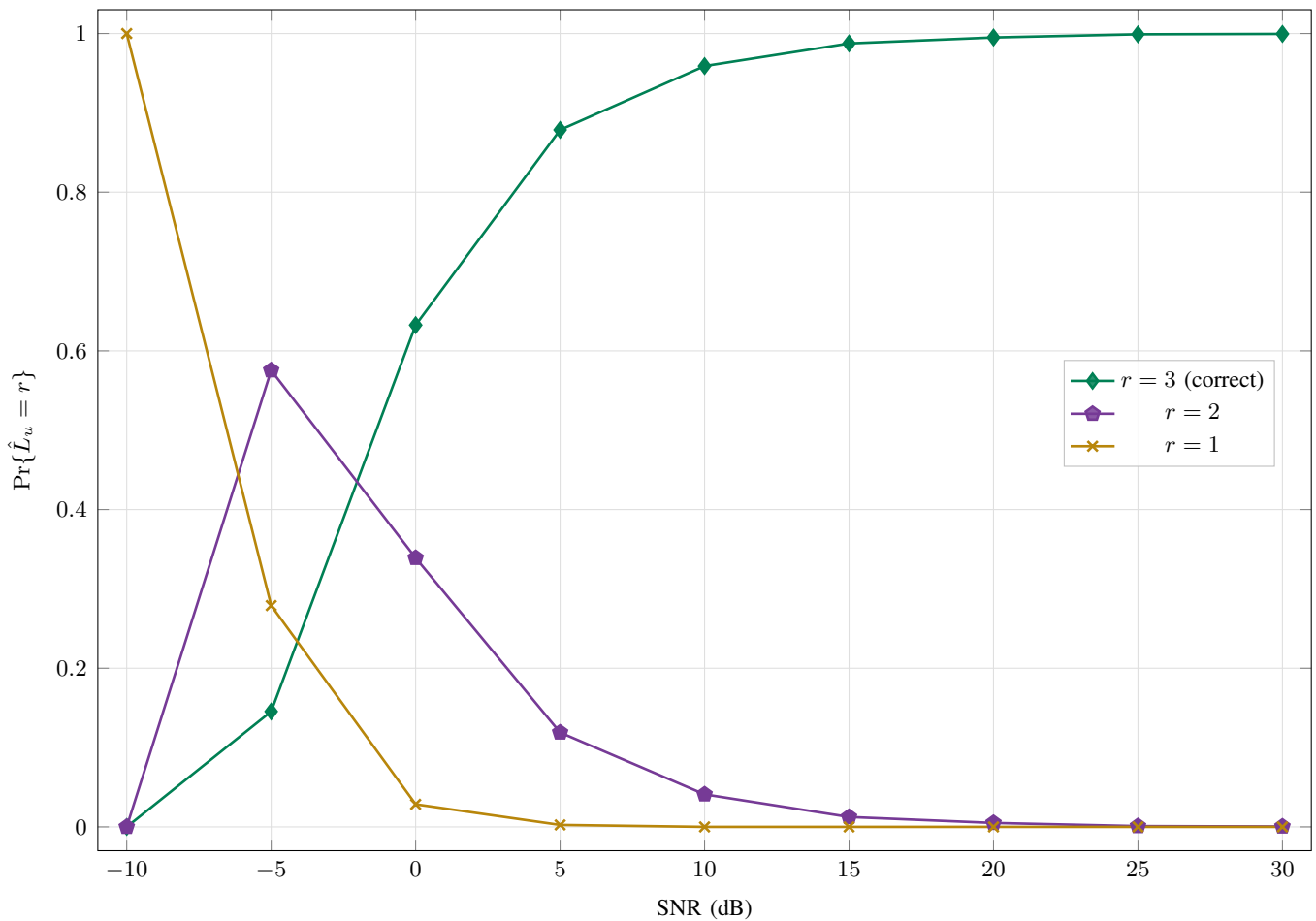
\begin{figure}[t]
\centering
\begin{tikzpicture}
\begin{axis}[
  resultaxis,
  xlabel={SNR (dB)},
  ylabel={$\Pr\{\hat L_u = r\}$},
  xmin=-11, xmax=31, xtick={-10,-5,...,30},
  ymin=-0.03, ymax=1.03, ytick={0,0.2,...,1.0},
  legend style={at={(0.97,0.52)}, anchor=east},
]
  \addplot[rthree]
    table[col sep=comma, x=snr_db, y=als_mdl_p_rank3]
    {\cedata{rank_detection_vs_snr.csv}};
  \addlegendentry{$r = 3$ (correct)}
  \addplot[rtwo]
    table[col sep=comma, x=snr_db, y=als_mdl_p_rank2]
    {\cedata{rank_detection_vs_snr.csv}};
  \addlegendentry{$r = 2$}
  \addplot[rone]
    table[col sep=comma, x=snr_db, y=als_mdl_p_rank1]
    {\cedata{rank_detection_vs_snr.csv}};
  \addlegendentry{$r = 1$}
\end{axis}
\end{tikzpicture}
\caption{Empirical distribution of the \gls{MDL} model order $\hat L_u$ against \gls{SNR}, true order $L_u = 3$.}
\label{fig:rank_vs_snr}
\end{figure}

Fig.~\ref{fig:nmse_vs_morphing} isolates the contribution of the fast
morphing mode by varying the number of morphing states $K$ at
$\mathrm{SNR} = 10$~dB. Both pipelines improve monotonically by roughly
$10$~dB from $K = 1$ to $K = 8$. Part of this improvement is plain noise
averaging, because the training frame $IJKP$ of Eq.~\eqref{eq:B_train_time}
grows with $K$ along this axis. The remainder is the extra mode that the fast
displacement adds to the \gls{CP} model of Eq.~\eqref{eq:B_CP_user}, which
raises the Kruskal bound of Eq.~\eqref{eq:C_kruskal_dim} from $L_u \leq 6$ at
$K = 1$ to $L_u \leq 13$ at $K = 8$ and improves the conditioning of the
factor matrices along the way. Tensor \gls{MUSIC} lags the \gls{ALS} pipeline
by $2$ to $6$~dB for $K \le 3$ and catches up at $K = 4$, because the fast mode
has to be long enough to be smoothed over before it contributes to the
covariance estimate of Eq.~\eqref{eq:C_cov_smooth}.

% Channel NMSE against the number of fast morphing states K, at SNR = 10 dB.
% Data: data/channel_estimation/nmse_vs_morphing.csv.
% NOTE: the training frame IJKP grows with K along this axis, so part of the
% gain is plain noise averaging; the caption says so.
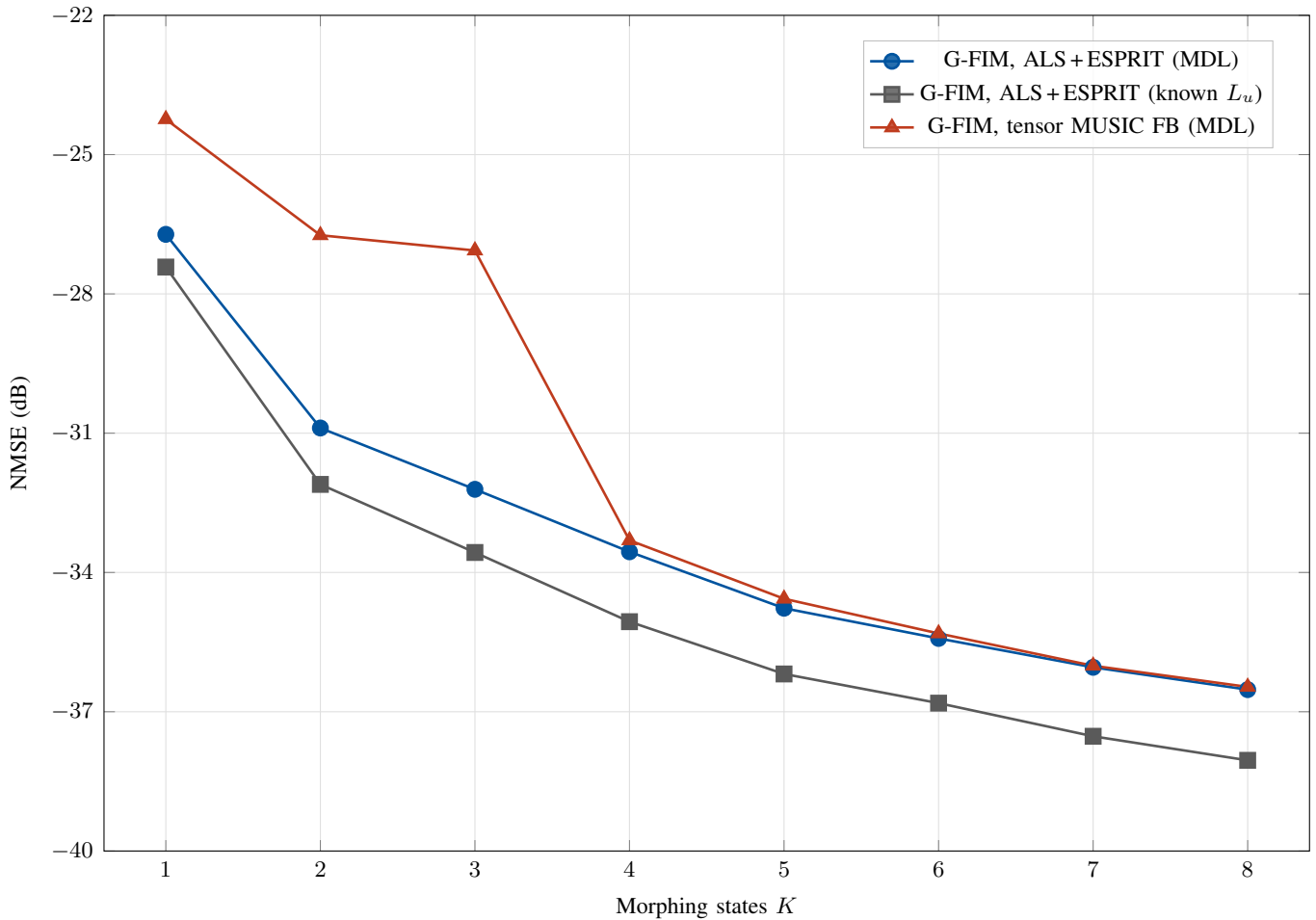
\begin{figure}[t]
\centering
\begin{tikzpicture}
\begin{axis}[
  resultaxis,
  xlabel={Morphing states $K$},
  ylabel={NMSE (dB)},
  xmin=0.6, xmax=8.4, xtick={1,2,...,8},
  ymin=-40, ymax=-22, ytick={-40,-37,...,-22},
  legend pos=north east,
]
  \addplot[als]      table[col sep=comma, x=k, y=als_mdl_nmse_db]
    {\cedata{nmse_vs_morphing.csv}};
  \addlegendentry{\glsentryshort{GFIM}, ALS\,+\,ESPRIT (MDL)}
  \addplot[alsknown] table[col sep=comma, x=k, y=als_known_nmse_db]
    {\cedata{nmse_vs_morphing.csv}};
  \addlegendentry{\glsentryshort{GFIM}, ALS\,+\,ESPRIT (known $L_u$)}
  \addplot[music]    table[col sep=comma, x=k, y=music_fb_mdl_nmse_db]
    {\cedata{nmse_vs_morphing.csv}};
  \addlegendentry{\glsentryshort{GFIM}, tensor MUSIC FB (MDL)}
\end{axis}
\end{tikzpicture}
\caption{Channel \gls{NMSE} against the number of morphing states $K$ at $\mathrm{SNR} = 10$~dB. The training frame $IJKP$ grows with $K$.}
\label{fig:nmse_vs_morphing}
\end{figure}

The companion question of how far the panel has to move on the fast time-scale
is answered by Fig.~\ref{fig:nmse_vs_step}, which varies the morphing step
$\Delta_y$ at fixed $K = 4$. This axis is free in training length, because
$IJKP$ does not depend on the step, and it therefore isolates the conditioning
of the fast mode. The \gls{ALS} pipeline is essentially flat and gains
$0.6$~dB over the whole range, because the shift invariance that
Eq.~\eqref{eq:C_esprit_freq} exploits need only be observable, not large.
Tensor \gls{MUSIC} instead improves by almost $6$~dB from $0.05\lambda$ to
$0.45\lambda$, since its smoothed covariance resolves the fast mode only after the
consecutive morphing states are separated by an appreciable fraction of a
wavelength. At the upper end of the range the \gls{ESPRIT} phase reaches the
unambiguous limit $\delta_y < \lambda/2$ derived in
Section~\ref{sec:als_esprit}. The fitting \gls{NMSE} is unaffected there, yet
the extrapolation is not, and at $\Delta_y = 0.5\lambda$ the \gls{ALS}
prediction \gls{NMSE} degrades from $-31.4$ to $-26.7$~dB while the \gls{AoA}
\gls{RMSE} rises from $0.30^\circ$ to $0.95^\circ$. This is the signature the
aliasing bound predicts, since a wrapped phase still fits the sampled grid and
fails only off it. A step slightly below half a wavelength is therefore the
right operating point for both pipelines.

% Channel NMSE against the fast-morphing step along the panel normal, at
% SNR = 10 dB.  Data: data/channel_estimation/nmse_vs_step.csv.
% NOTE: unlike the K axis, the training length IJKP does not depend on the
% step, so this axis costs the same everywhere -- the variation is pure
% conditioning of the fast mode.
\begin{figure}[t]
\centering
\begin{tikzpicture}
\begin{axis}[
  resultaxis,
  xlabel={Morphing step $\Delta_y/\lambda$},
  ylabel={NMSE (dB)},
  xmin=0.02, xmax=0.53, xtick={0.1,0.2,0.3,0.4,0.5},
  minor xtick={0.05,0.15,0.25,0.35,0.45},
  scaled x ticks=false,
  x tick label style={/pgf/number format/.cd, fixed, precision=1},
  ymin=-36, ymax=-26, ytick={-36,-34,...,-26},
  legend pos=north east,
]
  \addplot[als]      table[col sep=comma, x=morph_step_lambda, y=als_mdl_nmse_db]
    {\cedata{nmse_vs_step.csv}};
  \addlegendentry{\glsentryshort{GFIM}, ALS\,+\,ESPRIT (MDL)}
  \addplot[alsknown] table[col sep=comma, x=morph_step_lambda, y=als_known_nmse_db]
    {\cedata{nmse_vs_step.csv}};
  \addlegendentry{\glsentryshort{GFIM}, ALS\,+\,ESPRIT (known $L_u$)}
  \addplot[music]    table[col sep=comma, x=morph_step_lambda, y=music_fb_mdl_nmse_db]
    {\cedata{nmse_vs_step.csv}};
  \addlegendentry{\glsentryshort{GFIM}, tensor MUSIC FB (MDL)}
\end{axis}
\end{tikzpicture}
\caption{Channel \gls{NMSE} against the morphing step $\Delta_y$ at $\mathrm{SNR} = 10$~dB and $K = 4$. The training frame does not depend on $\Delta_y$.}
\label{fig:nmse_vs_step}
\end{figure}
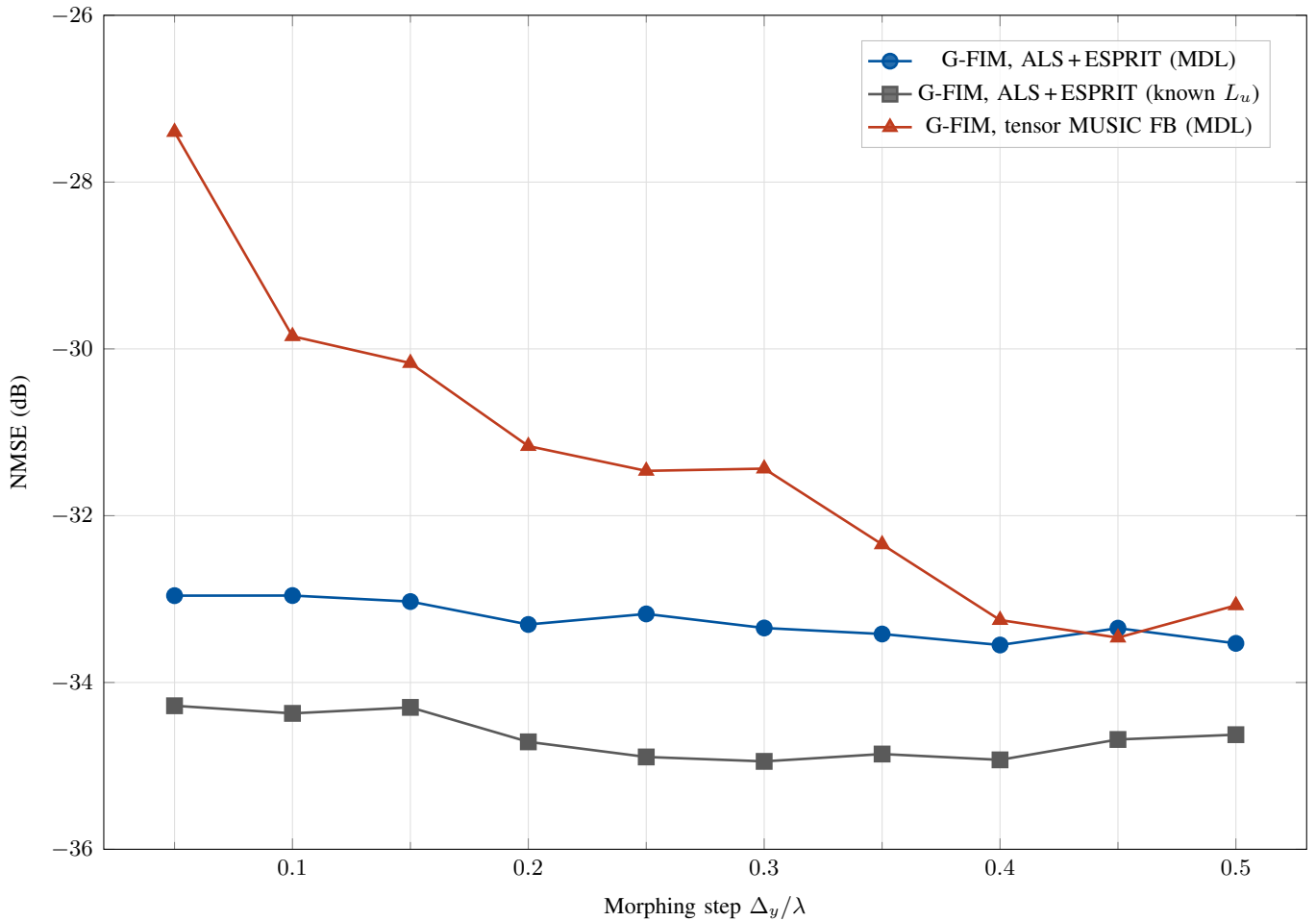

Finally, Fig.~\ref{fig:nmse_vs_paths} sweeps the number of paths per user at
$\mathrm{SNR} \in \{0, 10, 20\}$~dB, which tests the identifiability analysis
of Section~\ref{sec:uniqueness} against measured error. The two pipelines are
indistinguishable up to $L_u = 3$, where the \gls{CP} model is comfortably
identifiable at $(N,K,I,J) = (16,4,4,4)$. Beyond $L_u = 4$ they separate and
the separation widens with \gls{SNR}, so that at $20$~dB and at the Kruskal
bound $L_u = 9$ of Table~\ref{tab:uniqueness} the \gls{ALS} pipeline retains
$-35.6$~dB while tensor \gls{MUSIC} saturates at $-27.0$~dB. The reason is
that the \gls{MUSIC} spectrum of Eq.~\eqref{eq:C_music_spectrum} is
grid-limited and its spatial smoothing trades aperture for snapshots, whereas
Eq.~\eqref{eq:C_esprit_freq} resolves the wavevector components off the grid.
Neither curve breaks at $L_u = 9$, which is consistent with
Eq.~\eqref{eq:C_kruskal_dim} being a sufficient condition for uniqueness and
not a cliff in accuracy. At $0$~dB the ordering is irrelevant, since both are
noise-limited and lose about $3$~dB per doubling of $L_u$.

% Channel NMSE against the number of paths per user, at SNR in {0,10,20} dB.
% Data: data/channel_estimation/nmse_vs_paths.csv.  Colour and marker encode
% the estimator (as in every other result figure), while the line type encodes
% the SNR: solid = 0 dB, dashed = 10 dB, dotted = 20 dB.  Each pair is also
% annotated in the plot, so the legend only has to name the two estimators.
\begin{figure}[t]
\centering
\begin{tikzpicture}
\begin{axis}[
  resultaxis,
  xlabel={Paths per user $L_u$},
  ylabel={NMSE (dB)},
  xmin=0.6, xmax=9.6, xtick={1,2,3,4,6,8,9},
  ymin=-54, ymax=-4, ytick={-50,-40,...,-10},
  legend columns=2,
  legend style={at={(0.5,1.03)}, anchor=south, draw=gray!50},
]
  \addplot[als, snrlow]   table[col sep=comma, x=n_paths, y=als_mdl_nmse_db_snr0]
    {\cedata{nmse_vs_paths.csv}};
  \addlegendentry{\glsentryshort{GFIM}, ALS\,+\,ESPRIT}
  \addplot[music, snrlow] table[col sep=comma, x=n_paths, y=music_fb_mdl_nmse_db_snr0]
    {\cedata{nmse_vs_paths.csv}};
  \addlegendentry{\glsentryshort{GFIM}, tensor MUSIC FB}

  \addplot[als, snrmid, forget plot]
    table[col sep=comma, x=n_paths, y=als_mdl_nmse_db_snr10]
    {\cedata{nmse_vs_paths.csv}};
  \addplot[music, snrmid, forget plot]
    table[col sep=comma, x=n_paths, y=music_fb_mdl_nmse_db_snr10]
    {\cedata{nmse_vs_paths.csv}};

  \addplot[als, snrhigh, forget plot]
    table[col sep=comma, x=n_paths, y=als_mdl_nmse_db_snr20]
    {\cedata{nmse_vs_paths.csv}};
  \addplot[music, snrhigh, forget plot]
    table[col sep=comma, x=n_paths, y=music_fb_mdl_nmse_db_snr20]
    {\cedata{nmse_vs_paths.csv}};

  % SNR of each estimator pair, as curve groups.
  \curvegroup[38]{4,-14.3}{9mm}{4mm}{2,-7}{$0$~dB}
  \curvegroup[38]{1.5,-38}{8mm}{4mm}{2.8,-28}{$10$~dB}
  \curvegroup[22]{1.6,-48.7}{8mm}{4mm}{3.4,-52}{$20$~dB}
\end{axis}
\end{tikzpicture}
\caption{Channel \gls{NMSE} against the number of paths per user $L_u$ at $\mathrm{SNR} \in \{0, 10, 20\}$~dB, model order from \gls{MDL}. Color and marker identify the estimator, and the line type identifies the \gls{SNR}, being solid for $0$~dB, dashed for $10$~dB and dotted for $20$~dB.}
\label{fig:nmse_vs_paths}
\end{figure}
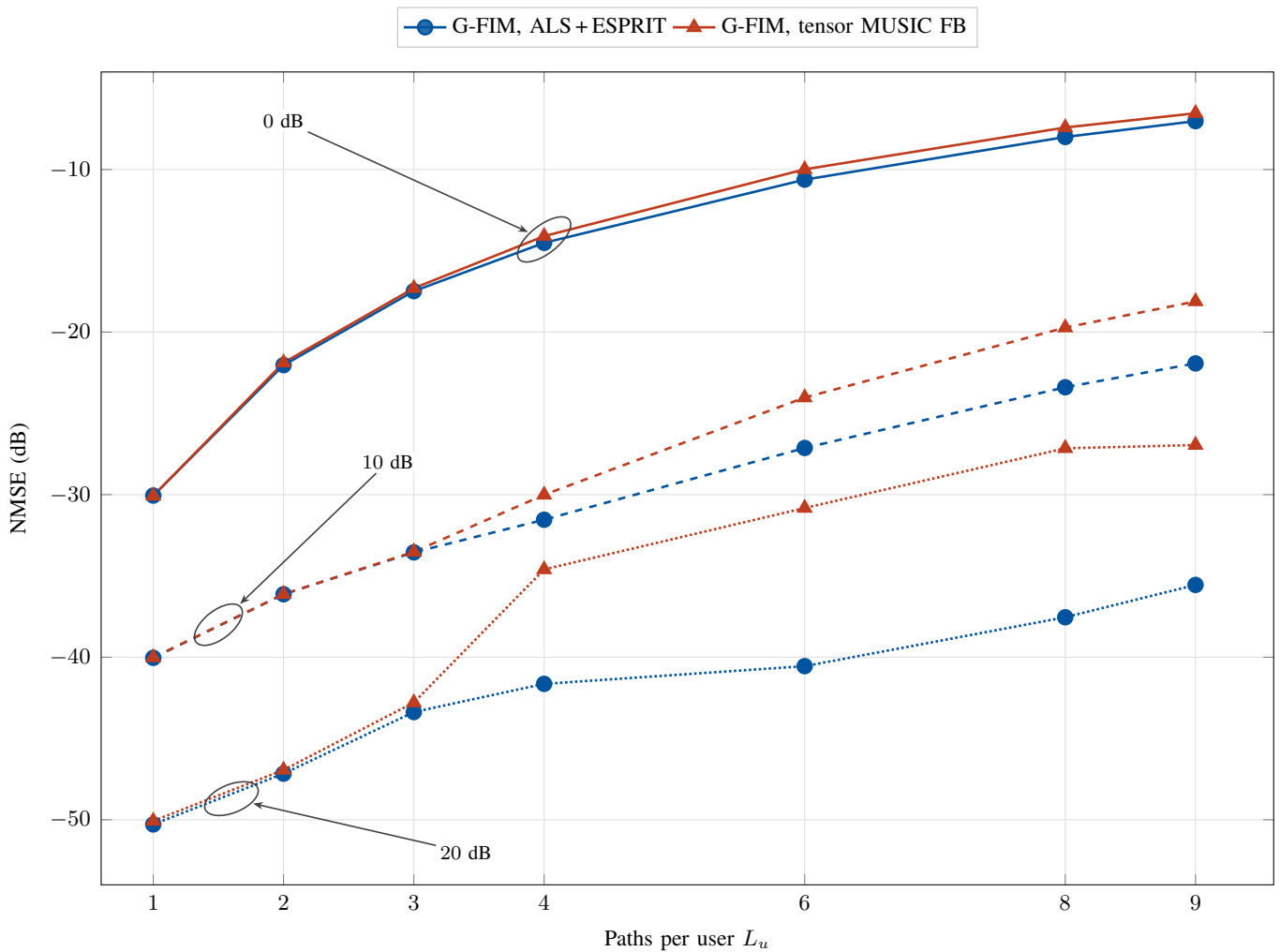

\subsection{\glsentryshort{TPC} Design}
\label{sec:results_precoder}

The campaigns of this subsection use a pilot length of $P = 16$, hence a training
frame of $IJKP = 1024$ symbols, rather than the $P = 4$ of
Table~\ref{tab:sim_setup}. The orthogonality requirement $P \geq U$ of
Section~\ref{sec:ortho_pilots} forces this, since the user axis of
Fig.~\ref{fig:sumrate_vs_users} runs to $U = 16$, and holding $P$ at the
larger value keeps the training budget constant along every axis, so that a
curve reports load, or path count, or motion protocol, and never a change of
budget. Throughout, the plotted quantity is the achieved rate defined above,
which is the \gls{TPC} built from the estimate and applied to the true channel
at the configuration the search selected. The predicted rate is recorded but
not plotted, because it is optimistically biased whenever the same estimate
both chooses the configuration and scores it, as anticipated in
Section~\ref{sec:prediction}. That bias is measured here and never exceeds
$0.06$~bps/Hz above $-5$~dB, which is small enough that selecting on an
estimate is not the limiting step.

Fig.~\ref{fig:sumrate_vs_snr} shows the evolution of the sum-rate against \gls{SNR} for
$U = 2$. Two things are worth noting. The estimated-\gls{CSI} curve is
indistinguishable from the genie above $-5$~dB, to within $0.01$~bps/Hz, which
says that the search of Eq.~\eqref{eq:D_exh_search} makes the same choice from
an estimate as it would from the truth and that the estimation stage is
effectively free once the model order is resolved. The second is that
\gls{MRC} saturates on inter-user interference at $17.1$~bps/Hz where
\gls{ZF} keeps its slope, which is precisely the behavior that
Eq.~\eqref{eq:D_sinr} predicts for a matched filter whose interference term
scales with the desired power.

% Achieved downlink sum-rate against receive SNR, U = 2.
% Data: data/precoder_design/sumrate_vs_snr.csv, 1000 trials per point.
% The plotted metric is `achieved_rate': the precoder built from the estimate,
% applied to the TRUE channel at the configuration the search selected.  The
% predicted rate is not plotted -- it is optimistically biased, since the same
% estimate chose the configuration and then scored it.
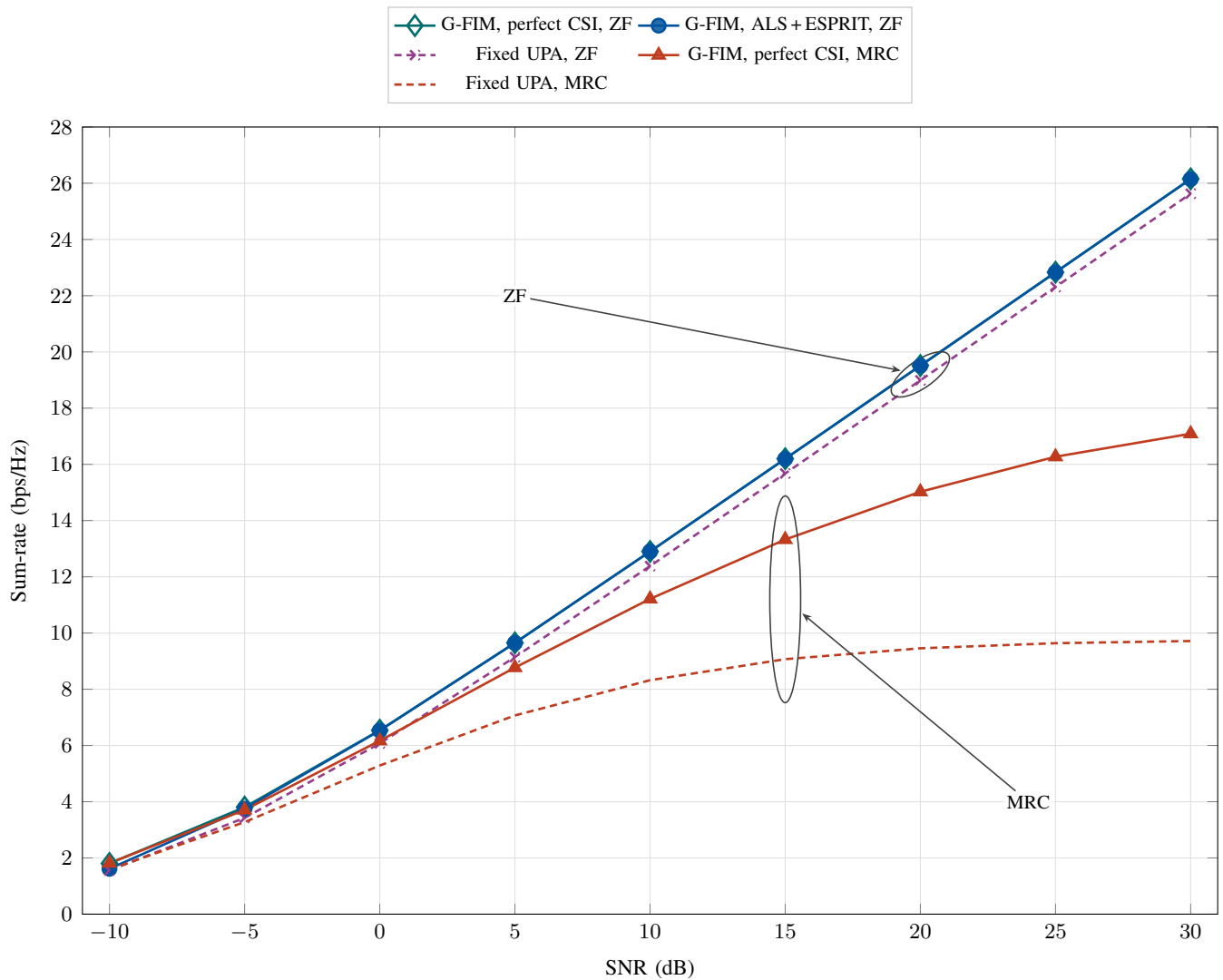
\begin{figure}[t]
\centering
\begin{tikzpicture}
\begin{axis}[
  resultaxis,
  xlabel={SNR (dB)},
  ylabel={Sum-rate (bps/Hz)},
  xmin=-11, xmax=31, xtick={-10,-5,...,30},
  ymin=0, ymax=28,
  legend columns=2,
  legend style={at={(0.5,1.03)}, anchor=south, draw=gray!50, font=\scriptsize},
]
  \addplot[genie] table[col sep=comma, x=snr_db, y=genie_zf_achieved_rate]
    {\bfdata{sumrate_vs_snr.csv}};
  \addlegendentry{\glsentryshort{GFIM}, perfect CSI, ZF}
  \addplot[als]   table[col sep=comma, x=snr_db, y=als_mdl_zf_achieved_rate]
    {\bfdata{sumrate_vs_snr.csv}};
  \addlegendentry{\glsentryshort{GFIM}, ALS\,+\,ESPRIT, ZF}
  \addplot[trad]  table[col sep=comma, x=snr_db, y=traditional_paper_zf_achieved_rate]
    {\bfdata{sumrate_vs_snr.csv}};
  \addlegendentry{Fixed UPA, ZF}
  \addplot[mrc]   table[col sep=comma, x=snr_db, y=genie_mrc_achieved_rate]
    {\bfdata{sumrate_vs_snr.csv}};
  \addlegendentry{\glsentryshort{GFIM}, perfect CSI, MRC}
  \addplot[mrc, densely dashed, mark=none]
    table[col sep=comma, x=snr_db, y=traditional_paper_mrc_achieved_rate]
    {\bfdata{sumrate_vs_snr.csv}};
  \addlegendentry{Fixed UPA, MRC}
  % Curve groups: the three ZF curves overlap, the two MRC curves fan out.
  \curvegroup[35]{20,19.2}{10mm}{4mm}{5,22}{ZF}
  % The two MRC curves are 4.2 bps/Hz apart at 15 dB, so this ellipse is sized
  % in axis units (fit) to enclose both in the one- and two-column layouts.
  \node[cgroup, fit={(axis cs:15,8.6) (axis cs:15,13.8)}, minimum width=4.5mm,
        inner sep=0pt] (cgm) {};
  \node[cglabel] (cgml) at (axis cs:24,4) {MRC};
  \draw[cgarrow] (cgml) -- (cgm);
\end{axis}
\end{tikzpicture}
\caption{Achieved sum-rate against receive \gls{SNR} for $U = 2$.}
\label{fig:sumrate_vs_snr}
\end{figure}

The gap between each \gls{GFIM} curve and its fixed-\gls{UPA} counterpart
in Fig.~\ref{fig:sumrate_vs_snr} is the reconfiguration gain over the
baseline of Section~\ref{sec:traditional}, which is the same $4 \times 4$
panel with the same pilot budget held at the reference position. Under
\gls{ZF} the gain settles at $0.52$~bps/Hz from $5$~dB onward and it is almost
independent of the \gls{SNR}, which is why the two \gls{ZF} curves are nearly
superimposed at this scale. Under \gls{MRC} it keeps escalating and reaches
$7.4$~bps/Hz at $30$~dB, because the fixed panel saturates on inter-user
interference earlier than the movable one, which has a configuration to move
to when the two users fail to separate. The \gls{ZF} figure is a modest number, and
Fig.~\ref{fig:sumrate_vs_users} shows why, since at $U = 2$ the array is
lightly loaded and there is little for the extra spatial degrees of freedom to
do.

% Achieved sum-rate against the served user count, SNR = 10 dB.
% Data: data/precoder_design/sumrate_vs_users.csv, 1000 trials per point.
% P = 16 at every point, so the training frame is 1024 symbols throughout and
% the axis is load rather than budget.
%
% The three precoders are drawn under perfect CSI.  One estimated-CSI curve
% (ALS+ESPRIT, ZF) is added so that the 6.22 bps/Hz reconfiguration gain the
% text quotes at U = 16 comes from a curve that is actually on the figure: it
% is als_mdl minus traditional, not genie minus traditional (6.29).  The
% estimated curve reuses the ZF colour with a hollow marker and a dashed line,
% so that colour still means "precoder" and the line type means "CSI source".
\begin{figure}[t]
\centering
\begin{tikzpicture}
\begin{axis}[
  resultaxis,
  xlabel={No. of served users $U$},
  ylabel={Sum-rate (bps/Hz)},
  xmin=1, xmax=17, xtick={2,4,...,16},
  ymin=0, ymax=34,
  legend columns=2,
  legend style={at={(0.5,1.03)}, anchor=south, draw=gray!50},
]
  \addplot[zf]   table[col sep=comma, x=n_users, y=genie_zf_achieved_rate_snr10]
    {\bfdata{sumrate_vs_users.csv}};
  \addlegendentry{\glsentryshort{GFIM}, perfect CSI, ZF}
  \addplot[mmse] table[col sep=comma, x=n_users, y=genie_mmse_achieved_rate_snr10]
    {\bfdata{sumrate_vs_users.csv}};
  \addlegendentry{\glsentryshort{GFIM}, perfect CSI, MMSE}
  \addplot[mrc]  table[col sep=comma, x=n_users, y=genie_mrc_achieved_rate_snr10]
    {\bfdata{sumrate_vs_users.csv}};
  \addlegendentry{\glsentryshort{GFIM}, perfect CSI, MRC}
  \addplot[zf, densely dashed, mark=o, mark options={solid}]
    table[col sep=comma, x=n_users, y=als_mdl_zf_achieved_rate_snr10]
    {\bfdata{sumrate_vs_users.csv}};
  \addlegendentry{\glsentryshort{GFIM}, ALS\,+\,ESPRIT, ZF}
  \addplot[trad] table[col sep=comma, x=n_users, y=traditional_paper_zf_achieved_rate_snr10]
    {\bfdata{sumrate_vs_users.csv}};
  \addlegendentry{Fixed UPA, ZF}
  % Curve group: ZF under perfect and estimated CSI coincide at every load.
  \curvegroup[-52]{14,21}{9mm}{4mm}{8,5}{ZF: perfect CSI $\approx$\\ALS\,+\,ESPRIT}
\end{axis}
\end{tikzpicture}
\caption{Achieved sum-rate against the number of served users $U$ at $\mathrm{SNR} = 10$~dB and $IJKP = 1024$ symbols.}
\label{fig:sumrate_vs_users}
\end{figure}
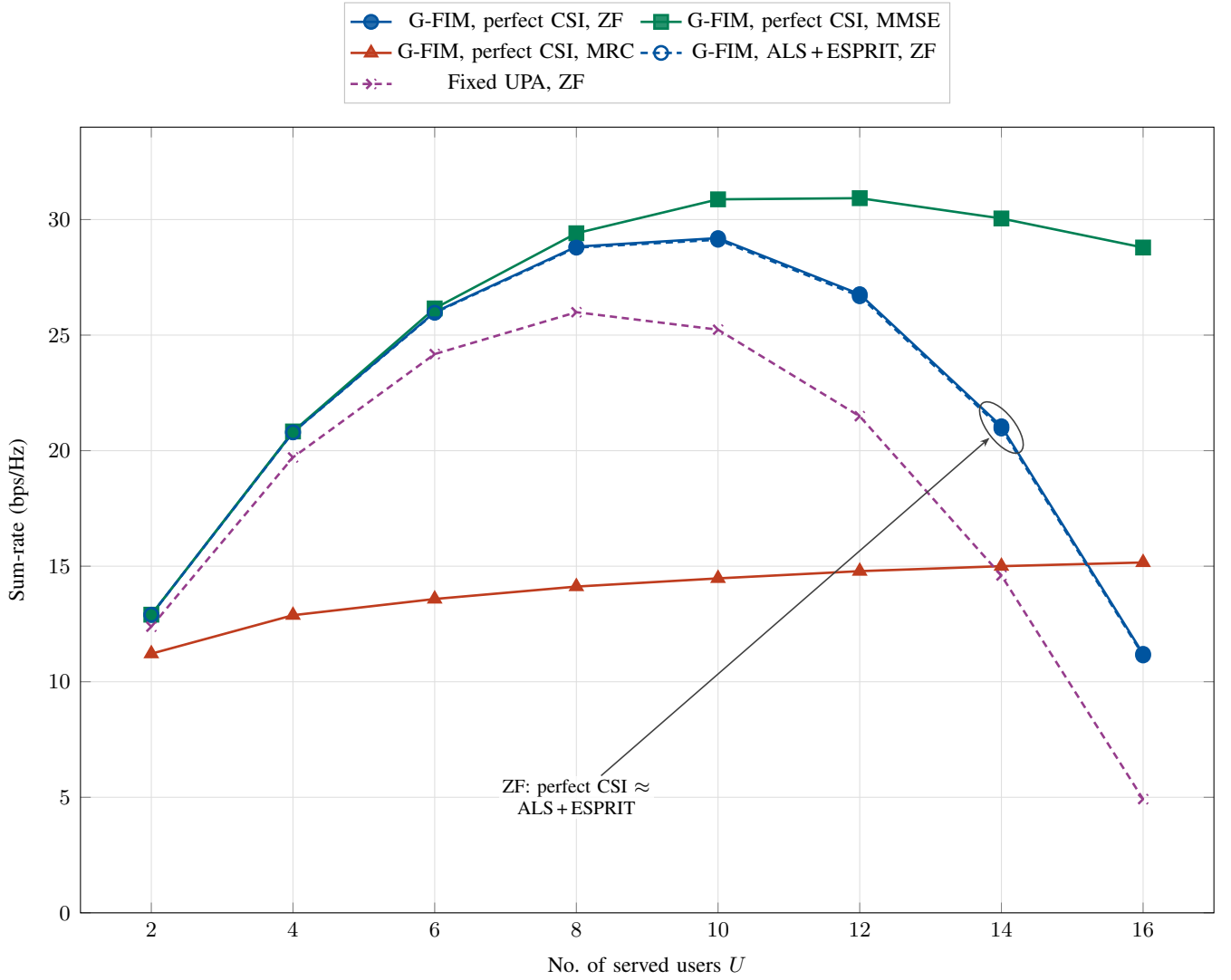

Sweeping the user count to the array size separates the \glspl{TPC} and the
references at once. The three \glspl{TPC} are shown under perfect \gls{CSI},
and \gls{ZF} is shown also under \gls{ALS}\,+\,\gls{ESPRIT} estimates, from
which it is indistinguishable at every load, so that the conclusions below do
not depend on the \gls{CSI} source. \gls{ZF} and \gls{MMSE} are numerically
the same \gls{TPC}
up to $U = 8$, where they differ by $0.58$~bps/Hz, and they separate only past
$U = N/2$. From there \gls{ZF} peaks at $29.2$~bps/Hz at $U = 10$ and then
collapses to $11.2$ at $U = N = 16$, where $\mathbf H$ is square, and the
inversion of Eq.~\eqref{eq:D_zf} amplifies noise without bound, while
\gls{MMSE} reaches $30.9$ and still holds $28.8$ at the same point. This is
the trade-off that the regularization weight of Eq.~\eqref{eq:D_mmse} is there to
make, and it is invisible over the range $U \leq 8$ that is usually reported.
The reconfiguration gain grows with load in steps, rising from $0.52$~bps/Hz at
$U = 2$ to $6.22$ at $U = 16$, and to $19.6$~bps/Hz at
$\mathrm{SNR} = 20$~dB. The fixed panel degrades fastest of all. This is
because its geometry is frozen at the reference position: when two users
happen to be poorly separated there, the fixed panel has no other
configuration to move to, whereas the \gls{GFIM} can select a panel position
at which their channels become better separated. As the load grows, such
poorly separated pairs become more frequent, and so does the advantage of
being able to move.

\subsubsection{Contribution of each motion mechanism}

The reference protocol spends its training budget on both mechanisms at once,
with $K = I = J = 4$. Fig.~\ref{fig:sumrate_vs_paths} asks whether both
are needed. It compares motion protocols at a fixed $KIJ = 64$, so that each
protocol offers the configuration search the same $64$ candidates. The
comparison therefore reveals which kind of diversity produced the candidates, rather
than how many there are.

% Achieved sum-rate against the path count, per motion protocol, at equal
% training budget K I J = 64.  Data: data/precoder_design/arch_sumrate_vs_paths.csv,
% U = 8, SNR = 20 dB, 1000 trials per point.
% Column names carry the protocol as k<K>i<I>j<J>: k1i16j4 is translation only,
% k2i8j4 is the combination, k64i1j1 is morphing only.
%
% Single axis (0-60 bps/Hz) for all protocols.  Translation-only and the
% combination differ by tenths of a bps/Hz and are bundled by an ellipse; the
% morphing-only estimate (dotted) collapses from 47.5 to 18 bps/Hz already at
% L_u = 2 while its genie stays above 53 -- the I = J = 1 case that Corollary 1
% cannot certify for any L_u.
\begin{figure}[t]
\centering
\begin{tikzpicture}
\begin{axis}[
  resultaxis,
  height=0.78\columnwidth,
  xlabel={Paths per user $L_u$},
  ylabel={Sum-rate (bps/Hz)},
  xmin=0.5, xmax=8.5, xtick={1,2,3,4,6,8},
  ymin=0, ymax=60, ytick={0,10,...,60},
  legend columns=2,
  legend style={at={(0.5,1.03)}, anchor=south, draw=gray!50, font=\scriptsize},
]
  \addplot[genie, ptransl]
    table[col sep=comma, x=n_paths, y=genie_zf_achieved_rate_k1i16j4_snr20]
    {\bfdata{arch_sumrate_vs_paths.csv}};
  \addlegendentry{\glsentryshort{GFIM} transl., perfect CSI}
  \addplot[als, ptransl]
    table[col sep=comma, x=n_paths, y=als_mdl_zf_achieved_rate_k1i16j4_snr20]
    {\bfdata{arch_sumrate_vs_paths.csv}};
  \addlegendentry{\glsentryshort{GFIM} transl., ALS\,+\,ESPRIT}
  \addplot[genie, pboth]
    table[col sep=comma, x=n_paths, y=genie_zf_achieved_rate_k2i8j4_snr20]
    {\bfdata{arch_sumrate_vs_paths.csv}};
  \addlegendentry{\glsentryshort{GFIM} both, perfect CSI}
  \addplot[als, pboth]
    table[col sep=comma, x=n_paths, y=als_mdl_zf_achieved_rate_k2i8j4_snr20]
    {\bfdata{arch_sumrate_vs_paths.csv}};
  \addlegendentry{\glsentryshort{GFIM} both, ALS\,+\,ESPRIT}
  \addplot[genie, pmorph]
    table[col sep=comma, x=n_paths, y=genie_zf_achieved_rate_k64i1j1_snr20]
    {\bfdata{arch_sumrate_vs_paths.csv}};
  \addlegendentry{\glsentryshort{GFIM} morph., perfect CSI}
  \addplot[als, pmorph]
    table[col sep=comma, x=n_paths, y=als_mdl_zf_achieved_rate_k64i1j1_snr20]
    {\bfdata{arch_sumrate_vs_paths.csv}};
  \addlegendentry{\glsentryshort{GFIM} morph., ALS\,+\,ESPRIT}
  \addplot[trad]
    table[col sep=comma, x=n_paths, y=traditional_paper_zf_achieved_rate_k4i4j4_snr20]
    {\bfdata{arch_sumrate_vs_paths.csv}};
  \addlegendentry{Fixed UPA, no motion}
  % Curve groups: the five curves at the top coincide.
  \curvegroup{4,54.9}{11mm}{4.5mm}{6.2,40}{Perfect CSI (all three),\\ALS\,+\,ESPRIT (transl., both)}
  \curvegroup[-40]{2.6,15}{4mm}{12mm}{5.2,24}{Morph. only, ALS\,+\,ESPRIT}
\end{axis}
\end{tikzpicture}
\caption{Achieved sum-rate against the number of paths per user for the equal-budget motion protocols of Table~\ref{tab:allocation}, $KIJ = 64$. Translation alone $(K,I,J) = (1,16,4)$, combined motion $(2,8,4)$, morphing alone $(64,1,1)$, and the fixed panel. \gls{ZF}, $U = 8$, $\mathrm{SNR} = 20$~dB.}
\label{fig:sumrate_vs_paths}
\end{figure}
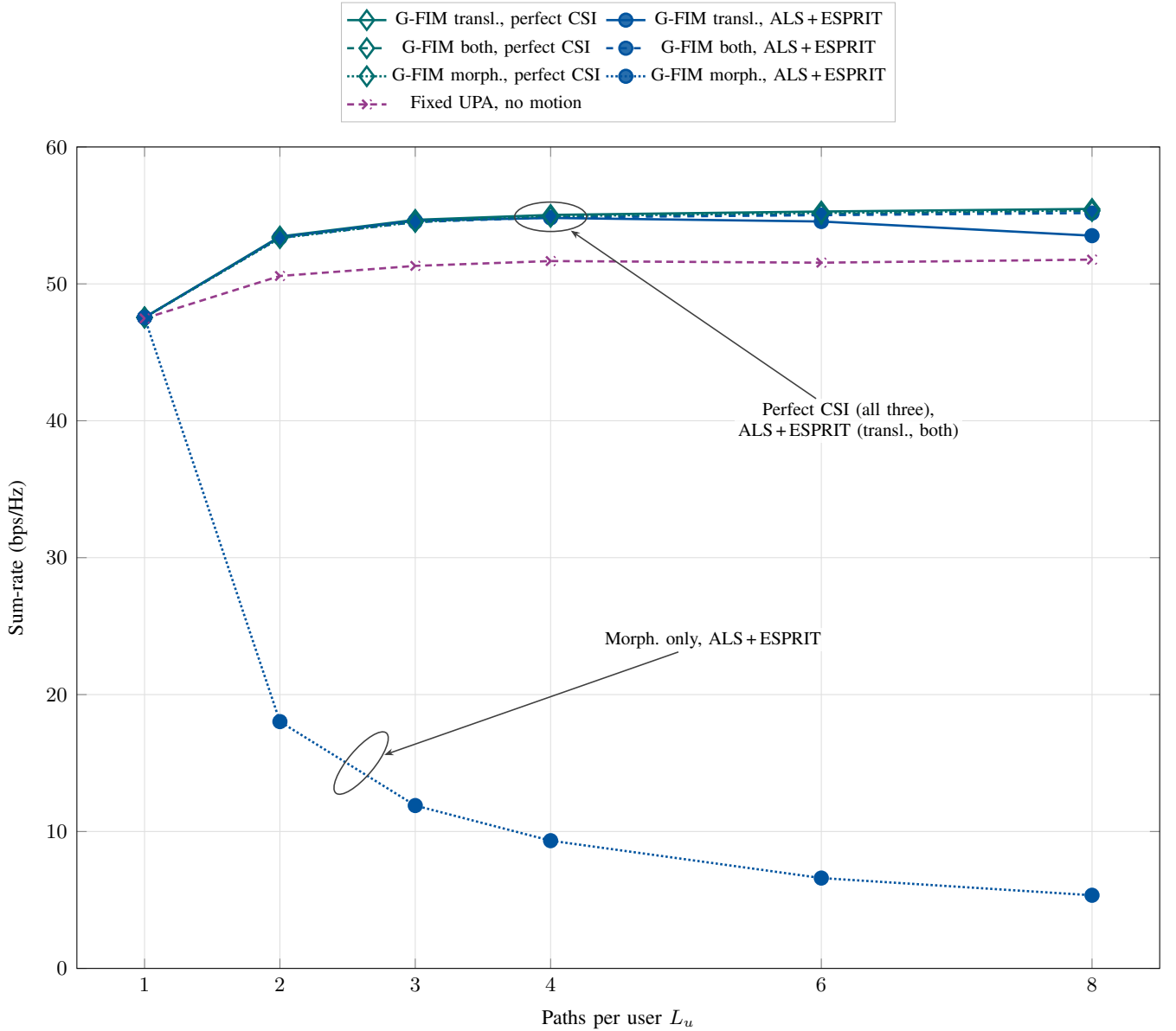

Under perfect \gls{CSI} the answer is no, since translation
alone, with $K = 1$ and $I = 16$, is the best of the protocols at every path
count, although by less than $0.2$~bps/Hz. This is a consequence of the
geometry rather than an accident. The model of
Eq.~\eqref{eq:C_wavevector_aoa} confines sources to the front half-space,
where the norm constraint $\lVert\mathbf q\rVert = k_0$ forces
$q^{(y)} = -\sqrt{k_0^2 - (q^{(x)})^2 - (q^{(z)})^2}$, which is the same
relation that Eq.~\eqref{eq:C_qy_norm} exploits when $K = 1$. The
normal-direction component is thus determined by the in-plane ones, so that
displacing the panel along its normal measures something the in-plane
translation already implies and cannot separate two users that translation
cannot.

In practice the ordering reverses once the channel is rich. Past $L_u = 4$ the
translation-only curve turns down, from $54.82$~bps/Hz at $L_u = 4$ to
$53.52$ at $L_u = 8$, while the combination continues to rise and reaches
$55.17$. The gap between the genie and the estimated curve of each protocol
explains why, since at $L_u = 8$ translation alone gives up $1.95$~bps/Hz to
estimation error against $0.15$~bps/Hz for the protocol with $K = 2$ and
$I = 8$. With $K = 1$ the training tensor is effectively 3-way, and the
morphing mode restores the fourth and thereby buys conditioning exactly where
the factor matrices are hardest to separate. This is not the identifiability
bound of Table~\ref{tab:allocation} speaking, because translation alone has
the higher Kruskal ceiling of the two, at $L_u \leq 17$ against $L_u \leq 11$.
It is instead the distinction drawn in Section~\ref{sec:allocation}, namely
that Eq.~\eqref{eq:C_kruskal_dim} certifies uniqueness in the noiseless case
and says nothing about how well conditioned the fit is at a finite \gls{SNR},
and it is the measurement behind the allocation rule stated there.

Morphing alone, the dotted curve that collapses in
Fig.~\ref{fig:sumrate_vs_paths}, is unusable, and instructively so. Its
estimate collapses from $47.5$~bps/Hz at $L_u = 1$ to $18.0$ at $L_u = 2$ and
to $5.3$ at $L_u = 8$, while its genie stays above $53$ throughout, so the
geometry is as good as the others and the estimate is not. This is the
degenerate case identified in Section~\ref{sec:uniqueness}, in which setting
$I = J = 1$ collapses the 4-way tensor to an $N \times K$ matrix that
Eq.~\eqref{eq:C_kruskal_dim} cannot certify for any $L_u$, so that the
factorization is not unique, and no estimator can recover the factors from the
data alone. The geometry of the protocol is sound and its estimation problem
is ill-posed. The fast morphing mode therefore earns its place in this
architecture as an estimation mechanism that supplements the slow translation,
rather than as a substitute for it or as a source of beamforming gain in its
own right.

What dominates all of these differences is moving at all. For $L_u \geq 2$ the
combination beats the fixed panel by $2.8$ to $3.5$~bps/Hz at $U = 8$, an
order of magnitude more than the differences between the protocols. Two
qualifications attach to that statement. At $L_u = 1$ the gain falls to
$0.07$~bps/Hz, because a single path leaves the extra spatial degrees of
freedom to exploit, and translation alone falls to $1.75$~bps/Hz at
$L_u = 8$, which is the same estimation penalty seen above and is now measured
against the immobile reference.

\subsubsection{Cost of the configuration search}

Measured against the estimator fit, the exhaustive search of
Algorithm~\ref{alg:beamforming} is inexpensive, taking $0.006$~s against
$0.504$~s at $U = 2$ and $0.013$~s against $3.60$~s at $U = 16$, so that it
stays under $1\%$ of the per-coherence-block cost in every case. This confirms
the two-phase split predicted in Section~\ref{sec:complexity}, in which the
$\mathcal O[KIJ \cdot N(UL + U^2)]$ search is dominated by the estimator that
precedes it.

\section{Conclusion}
\label{sec:conclusion}

This paper developed a training-efficient channel estimation framework for a
multiuser uplink served by a \gls{GFIM} whose panel translates rigidly over a
slow two-dimensional grid while its elements morph along the surface normal on
a fast time-scale. The superposition of the two motions was shown to induce a
4th-order \gls{CP} structure in the per-user channel tensor, from which two
estimation pipelines were derived, one pairing \gls{ALS} with \gls{ESPRIT}
and one based on multi-mode spatial smoothing followed by \gls{MUSIC}, both
sharing an \gls{MDL} front-end for model-order selection. The Kruskal
conditions of Eq.~\eqref{eq:C_kruskal_dim} bound the number of recoverable
paths in terms of $(N,K,I,J)$, and the parametric form of the estimates makes
the channel available at panel configurations that were never trained, which
in turn supports a configuration-selective multiuser \gls{TPC} that costs no
pilot resource.

The numerical campaign tested each of these derivations rather than merely
illustrating them. Model-order detection was found never to overestimate and
to become reliable above $0$~dB, which is where the \gls{MDL} and genie curves
merge. The morphing step behaved exactly as the unambiguous-phase limit of
\gls{ESPRIT} predicts, in that the fit is indifferent to aliasing on the
training grid while the extrapolation is not. The identifiability analysis was
confirmed in both directions, since accuracy degrades smoothly rather than
abruptly at the Kruskal ceiling, and since the morphing-only protocol, for
which the bound certifies nothing, fails as an estimation problem while
remaining sound as a geometry. On the \gls{TPC} side, estimated \gls{CSI}
matched perfect \gls{CSI} to within $0.01$~bps/Hz above $-5$~dB, and the gain
over a fixed panel of identical aperture and identical training budget grew
from $0.52$~bps/Hz at two users to $6.22$~bps/Hz at sixteen, which locates the
value of reconfiguration in the loaded regime.

Two conclusions carry beyond the specific architecture. The first is that the
fast morphing mode earns its place as an estimation mechanism that conditions
the factorization, and not as a source of beamforming gain, because the
half-space norm constraint makes the normal-direction wavenumber a function of
the in-plane ones. The second is that a sufficient identifiability condition
orders protocols differently from finite-\gls{SNR} conditioning, so that the
protocol with the higher Kruskal ceiling can be the worse one to deploy.
Extending the model to near-field propagation~\cite{cui2022nearfield}, where
the wavevector varies across the aperture, and the translation mode ceases to
be Vandermonde, and replacing the discrete configuration grid by a continuous
parametrization of the panel displacement, are the natural next steps.

\appendices

% Uncomment the line below to include the CRB appendix
%\input{crb}

\bibliographystyle{IEEEtran}
\bibliography{references}

\end{document}